\documentclass[11pt]{article}

\usepackage{geometry}
\usepackage[T1]{fontenc}
\usepackage[utf8]{inputenc} 

\usepackage{lmodern}    

\usepackage{fancyhdr}
\usepackage{xcolor}

\usepackage{amsmath}
\usepackage{amssymb}
\usepackage{bbm}

\numberwithin{equation}{section}

\DeclareMathAlphabet{\dutchcal}{U}{dutchcal}{m}{n}
\SetMathAlphabet{\dutchcal}{bold}{U}{dutchcal}{b}{n}
\DeclareMathAlphabet{\dutchbcal} {U}{dutchcal}{b}{n}

\usepackage{amsthm}
\newtheorem{theorem}{Theorem}
\newtheorem{lemma}[theorem]{Lemma}
\newtheorem{corollary}[theorem]{Corollary}
\newtheorem{proposition}[theorem]{Proposition}

\newtheorem{definition}[theorem]{Definition}

\usepackage{braket}    
\usepackage{physics}

\usepackage{hyperref}
\hypersetup{
    colorlinks=true,
    urlcolor=black,
    citecolor=blue,   
    pdftitle={Locally-Measured R\'enyi Divergences},
    pdfpagemode=FullScreen,
}
\usepackage{cite}

\begin{document}

\title{Locally-Measured R\'enyi Divergences}
\author{Tobias Rippchen\footnote{tobias.rippchen@rwth-aachen.de} , Sreejith Sreekumar, Mario Berta} 
\date{Institute for Quantum Information\\
                    RWTH Aachen University}
\maketitle

\begin{abstract}
    We propose an extension of the classical R\'enyi divergences to quantum states through an optimization over probability distributions induced by restricted sets of measurements. In particular, we define the notion of locally-measured R\'enyi divergences, where the set of allowed measurements originates from variants of locality constraints between (distant) parties $A$ and $B$. We then derive variational bounds on the locally-measured R\'enyi divergences and systematically discuss when these bounds become exact characterizations. As an application, we evaluate the locally-measured R\'enyi divergences on variants of highly symmetric data-hiding states, showcasing the reduced distinguishing power of locality-constrained measurements. For $n$-fold tensor powers, we further employ our variational formulae to derive corresponding additivity results, which gives the locally-measured R\'enyi divergences operational meaning as optimal rate exponents in asymptotic locally-measured hypothesis testing.
\end{abstract}




\section{Introduction}

In classical information theory, R\'enyi divergences \cite{Renyi} are fundamental quantities in one-shot information theory with direct operational meaning in hypothesis testing \cite{Csiszar}. They also serve as a parent quantity to many other relevant entropic measures. Hence, it is natural to investigate quantum extensions of these divergences. Due to the non-commutative nature of quantum physics, however, multiple plausible extensions exist (cf.\ e.g.\ \cite{Müller-Lennert, Wilde, Petz} and \cite[Chapter 4]{Tomamichel}). In this work, we follow a measurement-based approach to extend the classical R\'enyi divergence into the quantum domain. The particular method based on restricted measurement sets was previously applied to the relative entropy \cite{Piani} and the Schatten one-norm \cite{Matthews}. The process works as follows: the measurement turns a given quantum state \(\rho\) into a classical measure \(\mu_\rho^M\), to which we can apply the classical R\'enyi divergence. We then maximize the outcome over the set of allowed measurements \(\mathcal{M} \subseteq \text{ALL}\). This leads us to define the measured R\'enyi divergences for $\alpha>0$ as
\begin{equation}
    \text{$D_{\alpha}^{\mathcal{M}} \left(\rho\middle\|\sigma\right) := \sup_{M\in\mathcal{M}} D_\alpha \left( \mu_\rho^M \middle\| \mu_\sigma^M \right)$ with the classical $D_{\alpha}(\mu\|\nu):=\frac{1}{\alpha-1}\log\sum_{z\in \mathcal{Z}} \mu(z)^\alpha \nu(z)^{1-\alpha}$,}
\end{equation}
and the respective limits $\alpha\to1,+\infty$ (see Sec.\ \ref{Sec:Preliminaries} for precise definitions).

We put a particular focus on measurement sets defined by various levels of locality constraints, which give rise to the locally-measured R\'enyi divergences. This definition is motivated by an operational task that arises in quantum hypothesis testing. For this, consider two distant laboratories that want to differentiate between two hypotheses modeled by bipartite quantum states. To achieve this, the laboratories conduct measurements on the states. The implementation of global measurements, however, might not be feasible and, consequently, they are limited to local measurements, potentially supported by additional classical communication. Naturally, the question now arises as to how well this restricted set performs compared to all measurements. A particularly interesting example of states in this context are data-hiding states (cf.\ \cite{Terhal, Eggeling, DiVincenzo, DiVincenzo2}). These are perfectly distinguishable under global measurements but are locally almost indistinguishable. The study of this problem is an active area of research in the field of quantum information (see, e.g.,\ \cite{Matthews, Cheng, Brandao, Matthews2, Owari1, Owari2, Hayashi1, Calsamiglia1,Calsamiglia2,Nathanson,Yu1,Yu2,LiWang,Akibue}). Besides hypothesis testing, one further motivation for the study of these quantities is that correlation measures based on locally-measured quantities can be the crucial ingredient for lifting classical to quantum entropy inequalities (see, e.g., \cite{Piani, Matthews, Berta2, Li, Brandao2, Li2}).

In this work, we show that under certain additivity conditions the locally-measured R\'enyi divergences attain operational significance as optimal rate exponents in the strong converse regime of restricted hypothesis testing. We show that these conditions are satisfied for examples of highly symmetric data-hiding states and compute the Stein's as well as the strong converse exponent under locality-constrained measurements for these states. The main technical tool we employ for this are exact variational characterizations. In this context, the most useful result is the variational characterization of the PPT-measured max-divergence given by (see Sec.\ \ref{Sec:Variational} for precise definitions)
\begin{align}
     D_{\max}^\text{PPT} \left(\rho \middle\| \sigma \right) = \log \sup_{\stackrel{\omega>0}{\omega^\Gamma \geq 0}}  \bigg\{  \tr[\rho \omega] \, \bigg| \,  \tr[\sigma \omega] = 1 \bigg\} = \log \inf_{\stackrel{\lambda >0}{X,Y \in \mathcal{P}}} \bigg\{  \lambda \, \bigg| \, \lambda \sigma - \rho =  X + Y^\Gamma \bigg\} \, .
\end{align} 
The paper is structured as follows. In Sec.\ \ref{Sec:Preliminaries}, we formally define the measured R\'enyi divergences and make the connection to the existing literature. We then survey the mathematical properties of these divergences in Sec.\ \ref{Sec:Properties}. In particular, we show that suitable modifications of the original axioms, used by R\'enyi in \cite{Renyi} to derive the classical quantity, hold for the measured quantities. Sec.\ \ref{Sec:Variational} is devoted to the main technical result of this paper, which gives an upper bound on the divergences in terms of variational formulae. This generalizes a result for \(\alpha=1\) of \cite[Lemma 3]{Berta2}. We then discuss the conditions of when these bounds are tight and show that an exact characterization is available for two classes. We further prove  that in general there is a gap for the locality-constrained sets, which answers a question left open even for $\alpha=1$ in \cite[Section 2.2]{Berta2}. Moreover, we derive a dual program at the point \(\alpha=+\infty\). Lastly, in Sec.\ \ref{Sec:DataHiding}, we connect the locally-measured R\'enyi divergences to the problem of restricted hypothesis testing. For two families of states, we then employ the variational characterization to explicitly compute the measured R\'enyi divergences and use these results to give the Stein's and strong converse exponent (see Sec.\ \ref{Sec:Isotropic_States} and App.\ \ref{App:Werner_States}). As one example, we obtain for isotropic states with \(q \leq \frac{1}{d^2}\) that
\begin{align}
    &\zeta_\text{Stein}^\mathcal{M}\left(\Phi,\dutchcal{i}(q);\varepsilon\right) = \log(\frac{d+1}{qd+1}) && \text{and} & \zeta_\text{SC}^\mathcal{M}\left(\Phi,\dutchcal{i}(q);r\right) = r- \log(\frac{d+1}{qd+1}) 
\end{align}
for all \(\varepsilon\in(0,1)\) and all \(r \geq \log(\frac{d+1}{qd+1}) \) (see Sec.\ \ref{Sec:DataHiding} for precise definitions). 

This shows in particular that the single-letter characterization of testing $\Phi ^{\otimes n}$ against $(\Phi^{\perp})^{\otimes n}$ from \cite[Corollaries 11 \& 13]{Cheng} stays stable when going (not too far) away from the extremal states. We emphasize that the central insight is based on the additivity of the corresponding locally-measured R\'enyi divergences\,--\,which we derive via its variational characterizations. Note that similar additivity questions have also been asked in, e.g., \cite{Li, Cheng, Matthews2}.


\section{Definitions}\label{Sec:Preliminaries} 

\subsection{Notation}

Let us first introduce some notation that will be used throughout this manuscript. We consider quantum systems of finite\footnote{The sole exception is App.\ \ref{App:Infinite-Dimensional_Version}, where we discuss an infinite-dimensional extension in a self-contained manner.} dimension \(d \in \mathbb{N}\) to which we associate a Hilbert space \(\dutchcal{H} \simeq \mathbbm{C}^d\). If the quantum system consists of \(N\) subsystems, the Hilbert space factors as \( \dutchcal{H}:= \bigotimes_{i=1}^N \dutchcal{H}_i \), where \(\dutchcal{H}_i\) denotes the space of the \(i\)-th subsystem. We use capital letters \(A, B, C\), etc.\ to label these Hilbert spaces. Further, any indices \(x,y,z\) we will use are always meant to be taken from some implicitly defined finite alphabets \(\mathcal{X}\), \(\mathcal{Y}\) and \(\mathcal{Z}\). 

Linear operators \(X\) are labeled by subscripts to indicate which space they act on if its not clear from context. Moreover, when we introduce an operator \(X_{AB}\) acting on \(A \otimes B\) we implicitly also introduce its marginals \(X_A\) and \(X_B\) defined via the respective partial traces of \(X_{AB}\) over \(B\) and \(A\), respectively. Further, we use \(\mathcal{H} \) to denote the set of hermitian linear operators on \(\dutchcal{H}\) and extend real functions to these operators in the usual way by applying them on the spectrum of the operator coinciding with the domain of the function. Moreover, we use \(\geq\) and \(>\) to denote the Löwner order on operators, e.g.\ an operator \(X\) has full support if and only if \(X > 0\) and is positive semi-definite if and only if \(X \geq 0\). Let \(\mathcal{P}\) denote the set of positive semi-definite operators acting on \(\dutchcal{H}\) and \(\mathcal{S}\) the subset of states, i.e.\ positive semi-definite operators with unit trace. Lastly, \(\mathcal{U}\) denotes the set of unitary operators on \(\dutchcal{H}\).

Another fundamental concept which we need in the following is that of a measurement on the quantum system. These are described by a positive operator-valued measure (POVM) \(M\) over a finite alphabet \(\mathcal{Z}\).  Formally, \(M\) is a map from \(\mathcal{Z}\) to \(\mathcal{P}\) which satisfies \(M^z\geq0\) and \(\sum_{z \in \mathcal{Z}} M^z = 1_\dutchcal{H} \), where \( 1_\dutchcal{H} \) denotes the unit operator on \(\dutchcal{H}\). If, additionally, all elements of the POVM are projectors, we call it a projection-valued measure (PVM). We denote by ALL the set of all possible measurements that can be performed on the quantum system.  For a given \(\rho \in \mathcal{P}\), the POVM induces a positive measure \( \mu_\rho^M\) over \(\mathcal{Z}\) according to the Born rule: \(\mu_\rho^M(z) = \tr[\rho M^z] \). The measurement can also be viewed as a completely positive map from the set of positive operators to measures via
\begin{equation}\label{Eq:Measurement_Map}
    M : \rho \mapsto \sum_{z\in\mathcal{Z}} \ketbra{z}{z} \Tr[\rho M^z] \, ,
\end{equation}
where the measure is represented as an operator diagonal in the "label" basis \(\{\ket{z}\}\) of the space \(\mathbbm{C}^{|\mathcal{Z}|}\).


\subsection{Measured R\'enyi Divergence}

Let us start by recalling the definition of the classical R\'enyi divergence. Given two probability distributions \(\mu\) and \(\nu\) over a finite alphabet \(\mathcal{Z}\), we define the functional \(Q_\alpha\) in terms of an order parameter \(\alpha\in(0,1)\) as
\begin{equation}\label{Def:Q_Quantity}
    Q_\alpha\left(\mu\middle\|\nu\right) := \sum_{z\in \mathcal{Z}} \mu(z)^\alpha \nu(z)^{1-\alpha} \, .
\end{equation}
Observe that if \(\mu\perp\nu\), i.e.\ \(\mu\) is orthogonal to \(\nu\),\footnote{\(\mu\perp\nu\) iff there exists \(Z\subseteq\mathcal{Z}\) such that \(\mu(Z)=1\) and \(\nu(Z)=0\).} we have \(Q_\alpha(\mu\|\nu) = 0\). For orders \(\alpha > 1 \), Eq.\ \eqref{Def:Q_Quantity} diverges if \(\mu\not\ll\nu\), i.e.\ \(\mu\) is not absolutely continuous with respect to \(\nu\).\footnote{\(\mu\ll\nu\) iff for all \(z\in\mathcal{Z}\) \(\nu(z)=0\) implies \(\mu(z) = 0\).} Consequently, we employ Eq.\ \eqref{Def:Q_Quantity} to define \(Q_\alpha\) for orders \(\alpha\in(1,\infty)\) if \(\mu\ll\nu\) and else set \(Q_\alpha(\mu\|\nu)=+\infty\). Using this auxiliary quantity, we then define the R\'enyi divergence \cite{Renyi} of order \(\alpha\in(0,1)\cup(1,\infty)\) as
\begin{equation}\label{Def:Renyi_Divergence}
    D_\alpha\left(\mu\middle\|\nu\right) := \frac{1}{\alpha-1} \log Q_\alpha(\mu\|\nu) \, ,
\end{equation}
if either \(\alpha\in(0,1)\) and \(\mu\not\perp\nu\) or \(\alpha\in(1,\infty)\) and \(\mu\ll\nu\). Otherwise, we set \(D_\alpha(\mu\|\nu)=+\infty\). This definition is extended to the orders \(\alpha=1\) and \(\alpha=\infty\) based on a continuity argument. In the limit \(\alpha\to1\), Eq.\ \eqref{Def:Renyi_Divergence} converges to the well-known Kullback-Leibler divergence \cite{KullbackLeibler} defined as
\begin{align}
    D\left(\mu\middle\|\nu\right) := \left\{ 
    \begin{array}{ll}
        \sum_{z\in\mathcal{Z}} \mu(z) \log \left(\frac{\mu(z)}{\nu(z)}\right) & \text{if} \, \, \mu\ll\nu \\
        +\infty & \text{else}
    \end{array}
    \right. \, .
\end{align} 
Therefore, we identify \(D_1(\mu\|\nu) := D(\mu\|\nu)\). For \(\alpha \to \infty\), we obtain the max-divergence defined as
\begin{align}
    D_{\max}\left(\mu\middle\|\nu\right) := \sup_{z \in \mathcal{Z}} \log\left(\frac{\mu(z)}{\nu(z)}\right) 
\end{align}
and thus we set \(D_\infty(\mu\|\nu) := D_{\max}(\mu\|\nu)\). 

We are now set to define the measured R\'enyi divergences by lifting the classical quantity.
\begin{definition}\label{Def:Measured_Renyi_Divergence}
    For two quantum states \(\rho\in\mathcal{S}\) and  \(\sigma\in\mathcal{S}\), the measured R\'enyi divergence of order \(\alpha>0\) with respect to a subset of POVMs \(\mathcal{M}\subseteq\text{ALL}\) is defined as
    \begin{align}\label{Eq:Measured_Renyi_Divergence}
        D_{\alpha}^{\mathcal{M}} \left(\rho\middle\|\sigma\right) := \sup_{M\in\mathcal{M}} D_\alpha \left( \mu_\rho^M \middle\| \mu_\sigma^M \right) \, .
    \end{align}
\end{definition}

Mirroring the classical definition, we may also rewrite the measured R\'enyi divergence for orders \(\alpha\in(0,1)\cup(1,\infty)\) as
\begin{equation}\label{Def:Measured_Q_Connection}
    D_{\alpha}^\mathcal{M} \left(\rho \middle\| \sigma \right) = \frac{1}{\alpha-1} \log Q_{\alpha}^{\mathcal{M}} \left(\rho \middle\| \sigma \right) \, ,
\end{equation}
where we defined \( Q_{\alpha}^{\mathcal{M}}\) for \(\alpha \in(0,1)\) as
\begin{equation}
    Q_{\alpha}^{\mathcal{M}} \left(\rho\middle\|\sigma\right) := \inf_{M\in\mathcal{M}} Q_\alpha \left( \mu_\rho^M\middle\|\mu_\sigma^M\right)
\end{equation} 
and for \(\alpha\in(1,\infty)\) as the same expression but with the infimum replaced by a supremum. 


\subsection{Locally-Measured R\'enyi Divergences}\label{Sec:Local_Measurement_Sets} 

If the quantum system under consideration consists of multiple parties, we can identify some special classes of restricted measurements that are of particular interest in quantum information theory. These are the ones that are defined by locality constraints, i.e.\ by the local operations and classical communication (LOCC) paradigm. We focus our study on the special case of bipartite systems, where only two parties \(A\) and \(B\) are involved. The extension of these classes to multiple parties is mostly straightforward (see e.g.\ \cite{Matthews, Chitambar}). 

The most restrictive constraint that still has operational significance is given by the set LO\((A:B)\). Here, \(A\) and \(B\) are only allowed to perform a measurement on their part of the system and no classical communication is permitted during the measurement procedure. A general POVM element of this class is of product form, i.e.\ it may be written as
\begin{equation}
   M_{AB}^{(x,y)} = M_A^x \otimes M_B^y \, ,
\end{equation}
where \( \{ M_A^x\}_{x} \) and \( \{ M_B^y \}_{y} \) are POVMs on the \(A\) and \(B\) system, respectively. We additionally consider a subset of this set denoted as P-LO\((A:B)\), where we only allow for projective measurements. 

If we allow a single round of classical communication from \(A\) to \(B\) to aid the measurement process, we obtain the set LOCC\(_1(A\to B)\). Here, \(A\) measures her system, communicates the outcome to \(B\), and, conditional on this outcome, \(B\) performs his measurement. Without loss of generality, we can write a POVM element of this class in the following form 
\begin{equation}
   M_{AB}^{(x,y)} = M_A^x \otimes M_B^{y|x} \, ,
\end{equation}
where \( \{ M_A^x\}_{x } \) and \( \{ M_B^{y|x} \}_{y } \) are POVMs.\footnote{See also \cite[Lemma 4]{Brandao2} for an equivalent definition.} As above, we can also define a projective version of this set, denoted P-LOCC\(_1(A:B)\), where we allow only for PVMs to be performed.

Now, allowing for ever more rounds of classical communication between the parties would yield ever bigger sets of the LOCC type \cite{Chitambar}. Generally, the set LOCC\((A:B)\) is defined as all POVMs that can be implemented by local operations and arbitrary amounts of classical communication. This set, although operationally well-defined, is mathematically difficult to characterize (cf.\ also \cite{Chitambar}). Therefore, one often looks at relaxations that admit a simpler mathematical structure. 

In this manuscript, we consider two such relaxations, namely the sets SEP\((A:B)\) and PPT\((A:B)\). The set SEP\((A:B)\) consists of POVMs whose elements are separable operators. Without loss of generality, such an element is given by 
\begin{equation}
    M_{AB}^z = M_A^{z} \otimes M_B^{z} \, ,
\end{equation}
where \( M_A^{z} \) and \( M_B^{z} \) are positive semi-definite operators. Furthermore, the set PPT\((A:B)\) consists of POVMs whose elements have positive partial transpose (PPT), i.e.\ they satisfy 
\begin{equation}
    \left( M_{AB}^z \right)^\Gamma \geq 0 \, ,
\end{equation}
where \((\cdot)^\Gamma\) denotes the transposition w.r.t.\ the \(B\)-system.\footnote{Note that since the partial transpositions w.r.t.\ the \(A\) and \(B\) system are related to each other via the full transposition, this also implies positivity of the partial transpose on \(A\).} 

The following inclusions are trivial
\begin{equation}
    \text{LO} \subseteq \text{LOCC}_1 \subseteq \text{LOCC} \subseteq \text{SEP} \subseteq \text{PPT} \subseteq \text{ALL} \, ,
\end{equation}
and all of them are known to be strict if the dimension of the system is large enough \cite{Chitambar}. In the following, we refer to the induced measured R\'enyi divergences of the LOCC type as locally-measured. By the above inclusions, we also immediately get the chain of inequalities
\begin{equation}
    D_\alpha^\text{LO} \leq D_\alpha^\text{LOCC\(_1\)} \leq D_\alpha^\text{LOCC} \leq  D_\alpha^\text{SEP} \leq D_\alpha^\text{PPT} \leq D_\alpha^\text{ALL} \, .
\end{equation}


\subsection{Connection to Quantum R\'enyi Divergences}

Before we continue, let us point out known connections of the measured R\'enyi divergence to other quantum generalizations studied in the literature. The two most common extensions are the sandwiched R\'enyi divergence \cite{Müller-Lennert, Wilde} defined as
\begin{equation}
    \Tilde{D}_\alpha(\rho\|\sigma) := \left\{ \begin{array}{ll}
        \frac{1}{\alpha-1} \log\Tr[\left(\sigma^\frac{1-\alpha}{2\alpha}\rho\sigma^\frac{1-\alpha}{2\alpha}\right)^\alpha] & \text{if} \; (\alpha \in (0,1) \wedge \rho\not\perp\sigma)\vee \rho\ll\sigma \\
        +\infty & \text{else}
    \end{array} 
    \right.
\end{equation}
and the one due to Petz \cite{Petz} defined as 
\begin{equation}
    \bar{D}_\alpha(\rho\|\sigma) := \left\{ \begin{array}{ll}
        \frac{1}{\alpha-1} \log\Tr[\rho^\alpha\sigma^{1-\alpha}] & \text{if} \; (\alpha \in (0,1) \wedge \rho\not\perp\sigma)\vee\rho\ll\sigma  \\
        +\infty & \text{else}
    \end{array} \right. \, .
\end{equation}
In the limit \(\alpha\to1\), both of these converge to the Umegaki relative entropy \cite{Umegaki} given by
\begin{equation}
    D(\rho\|\sigma) := \left\{ \begin{array}{ll}
        \Tr[ \rho (\log\rho-\log\sigma) ] & \text{if} \, \rho\ll\sigma  \\
        +\infty & \text{else}
    \end{array} \right.
\end{equation}
and for \(\alpha\to\infty\) we obtain the quantum max-divergence \cite{Datta, Renner}
\begin{equation}\label{Eq:Quantum_Max_Divergence}
    D_{\max} \left(\rho \middle\| \sigma \right) := \log \inf_{\lambda >0} \bigg\{ \lambda \; \bigg| \; \lambda\sigma-\rho\geq0 \bigg\} \, .
\end{equation}

We can then write the measured R\'enyi divergence equivalently as 
\begin{align}
    D_\alpha^\mathcal{M}(\rho\|\sigma) &= \sup_{M\in\mathcal{M}}\tilde{D}_\alpha \left(M(\rho)\middle\|M(\sigma)\right) & \text{or} && D_\alpha^\mathcal{M}(\rho\|\sigma) &= \sup_{M\in\mathcal{M}}\bar{D}_\alpha \left(M(\rho)\middle\|M(\sigma)\right) \, ,
\end{align}
where the maximization is in terms of the measurement maps as defined in Eq.\ \eqref{Eq:Measurement_Map}. To verify this, note that the operators commute after application of the measurement map and both extensions reduce to the classical quantity on commuting states. This gives a top-down approach to define the measured R\'enyi divergence that complements our bottom-up approach of lifting the classical quantity.

The sandwiched R\'enyi divergence is known to satisfy a data-processing inequality under measurements for \(\alpha\geq1/2\) \cite[Appendix A]{Mosonyi},\footnote{It is not known whether this holds for \(\alpha\in(0,1/2)\).} i.e.\ we have
\begin{equation}
    \tilde{D}_\alpha(M(\rho)\|M(\sigma)) \leq \tilde{D}_\alpha(\rho\|\sigma) \, ,
\end{equation}
which immediately gives us
\begin{equation}\label{Eq:Comparison_Quantum}
    D_\alpha^\mathcal{M}(\rho\|\sigma) \leq \tilde{D}_\alpha(\rho\|\sigma) \; \; \text{for} \; \; \alpha \geq \frac{1}{2} \, .
\end{equation}
Moreover, the Petz divergence is known to be monotone under measurements for all \(\alpha>0\) \cite[Appendix A]{Mosonyi}, i.e.\ we get the analogous inequality \(D_\alpha^\mathcal{M} \leq \bar{D}_\alpha\) for all orders \(\alpha>0\). Let us emphasize that the inequality in Eq.\ \eqref{Eq:Comparison_Quantum} is in general strict even if no restrictions on the measurements are made. In fact, we know from \cite[Theorem 6]{Berta1} that for \(\alpha \in (1/2,\infty) \), the inequality is strict if \([\rho,\sigma] \not = 0\) and \(\tilde{D}_\alpha(\rho\|\sigma)<+\infty\). This differs from the measured norms of \cite{Matthews}, defined via the trace norm \(\norm{X}_1 := \tr|X|\) as 
\begin{equation}
    \norm{\rho}_\mathcal{M} := \sup_{M\in\mathcal{M}} \norm{M(\rho)}_1 \, ,
\end{equation}
where it is known that \( \norm{\cdot}_\text{ALL}=\norm{\cdot}_1\) holds. There are, however, two noteworthy exceptions to this rule. First, the measured R\'enyi divergence at order \(\alpha = 1/2\) is connected to the measured fidelity via the relationship
\begin{equation}\label{fidelity_connect}
    D_{1/2}^{\mathcal{M}} \left(\rho \middle\| \sigma \right)=- 2 \log F^\mathcal{M}(\rho,\sigma)\, ,
\end{equation}
where the measured fidelity is defined as 
\begin{equation}
    F^\mathcal{M}(\rho,\sigma) := \inf_{M \in \mathcal{M}} Q_{1/2} \left( \mu_\rho^M \middle\| \mu_\sigma^M \right) \,.
\end{equation}
It is known that \(F=F^\text{ALL}\) holds for the quantum fidelity \(F(\rho,\sigma) := \tilde{Q}_{1/2}(\rho\|\sigma)\) \cite[Appendix A]{Mosonyi}, which in turn implies \(D_{1/2}^\text{ALL} = \Tilde{D}_{1/2} \). Second, the quantum max-divergence is known to be achievable by a measurement as well \cite[Appendix A]{Mosonyi}, i.e.\ we also have the equality \(D_\infty^\text{ALL}=D_{\max}\).


\subsection{Related Concepts}

Lastly, let us make some remarks about other connections to previously-studied quantities in the literature. First, the definition of measured norms by Matthews et.\ al.\ in \cite{Matthews} motivated our definition of the measured R\'enyi divergence. At \(\alpha=1\), our definition recovers the measured relative entropy \(D^{\mathcal{M}}\) introduced by Piani in \cite{Piani} and further studied by Berta and Tomamichel in \cite{Berta2}. Here, the case of unrestricted measurements, that is the quantity \(D_1^{\text{ALL}}\), corresponds to the most-well-known notion of measured relative entropy, originally studied by Donald \cite{Donald} as well as Hiai and Petz \cite{Hiai}. The quantity \(D_\alpha^{\text{ALL}}\) for \(\alpha>0\) was investigated by Berta et.\ al.\ in \cite{Berta1}. Mosonyi and Hiai studied its restriction to binary measurements in \cite{Mosonyi2}. If the measurement set \(\mathcal{M}\) consists only of a single POVM, we can recover the generalized \(\alpha\)-observational entropies of Sinha and Aravinda \cite{Sinha}. Moreover, note that \(\alpha=\infty\) gives the measured max-divergence \( D_{\max}^{\mathcal{M}}\). We show in Sec.\ \ref{Sec:Measured_Max_Divergence} that these are connected to the cone-restricted max-divergence of George and Chitambar \cite{George}.


\section{Mathematical Properties}\label{Sec:Properties}

In this section, we investigate the mathematical properties of the measured R\'enyi divergences. Although the measured R\'enyi divergence is not a quantum divergence in the axiomatic sense defined by R\'enyi \cite{Renyi} (cf.\ also \cite[Chapter 4]{Tomamichel}), we show in the following that adapted versions of his axioms do hold. Our proofs are based on properties of the classical R\'enyi divergence and for an overview of them we refer to \cite[Chapter 4]{Tomamichel} (see also \cite{Csiszar}).


\subsection{Positivity, Monotonicity, Convexity and Continuous Extension}

Let us start with some basic properties that the measured R\'enyi divergences share with their classical counterpart and that hold independent of the underlying measurement set \(\mathcal{M}\). 
\begin{lemma}\label{Lem:General_Properties}
Let \(\rho, \sigma \in \mathcal{S}\), \( \mathcal{M} \subseteq \text{ALL}\) and \(\alpha > 0\). The following hold:
\begin{enumerate}
    \item \(D_{\alpha}^{\mathcal{M}} \left(\rho \middle\| \sigma \right) \geq 0 \),
    \item \(D_{\alpha}^{\mathcal{M}} \left(\rho \middle\| \sigma \right)\) is a non-decreasing function in \(\alpha\),
    \item \(D_1^\mathcal{M}\left(\rho \middle\| \sigma \right) = \sup_{\alpha\in(0,1)} D_\alpha^\mathcal{M}\left(\rho \middle\| \sigma \right)\) and \(D_\infty^\mathcal{M}\left(\rho \middle\| \sigma \right) = \sup_{\alpha\in(1,\infty)} D_\alpha^\mathcal{M}\left(\rho \middle\| \sigma \right)\),
    \item \(Q_{\alpha}^{\mathcal{M}} \left(\rho \middle\| \sigma \right)\) is jointly concave in \((\rho,\sigma)\) for \(\alpha < 1\) and jointly convex for \(\alpha>1\) and therefore \(D_{\alpha}^{\mathcal{M}} \left(\rho \middle\| \sigma \right)\) is jointly convex for \(\alpha \leq 1\) and jointly quasi-convex for \(\alpha>1\).
\end{enumerate}
\end{lemma}

\begin{proof} These results follow by lifting the corresponding classical properties (see  \cite[Chapter 4]{Tomamichel}).

(1) This follows immediately from the positivity of the classical quantity. 

(2) 
By monotonicity of the classical quantity, we have for all \(\alpha, \beta\) with \(0 < \alpha \leq \beta\) that
\begin{align}\label{mon_full}
    D_\alpha \left(\mu_\rho^M\middle\| \mu_\sigma^M\right) & \leq D_\beta \left(\mu_\rho^M\middle\|\mu_\sigma^M\right)
\end{align}
for an arbitrary measurement \(M\in\mathcal{M}\) and it then directly follows that \( D_\alpha^{\mathcal{M}} \left( \rho \middle\| \sigma \right) \leq D_\beta^{\mathcal{M}} \left( \rho \middle\| \sigma \right)\). 


(3) 
Note that we can write
\begin{equation}\label{Eq:D1_Supremum_Def}
    D_1^\mathcal{M} (\rho\|\sigma) = \sup_{M \in \mathcal{M}} D_1 \left(\mu_\rho^M\middle\|\mu_\sigma^M \right) = \sup_{M \in \mathcal{M}} \sup_{\alpha \in (0,1)} D_\alpha \left(\mu_\rho^M\middle\|\mu_\sigma^M \right) = \sup_{\alpha \in (0,1)} D_\alpha^\mathcal{M}(\rho\|\sigma) \, ,
\end{equation}
where the second equality follows by the classical property\footnote{Due to monotonicity and continuity in \(\alpha\), we may write the limit \(\alpha\to1\) as a supremum.} and the last is an interchange of suprema. 

The statement for \(\alpha = \infty\) follows by an analogous argument.

(4) Let us first consider the case \(\alpha\in(0,1)\). 
Note that by linearity of the trace, we have
\begin{equation}
    \mu_{\lambda \rho_1 + (1-\lambda) \rho_2}^M = \lambda \mu_{\rho_1}^M + (1-\lambda) \mu_{\rho_2}^M \, ,
\end{equation}
where the addition is understood element-wise. Joint concavity of \( Q_\alpha^\mathcal{M} \) is then a result of
\begin{align}
     Q_\alpha^\mathcal{M} \big( \lambda\rho_1+(1-\lambda)\rho_2 \big\| \lambda\sigma_1+(1-\lambda)\sigma_2 \big)
     &= \inf_{M \in \mathcal{M}} Q_\alpha \left( \lambda \mu_{\rho_1}^M+(1-\lambda)\mu_{\rho_2}^M\middle\|\lambda \mu_{\sigma_1}^M+(1-\lambda)\mu_{\sigma_2}^M\right) \\
     &\geq \inf_{M \in \mathcal{M}} \lambda Q_\alpha \left( \mu_{\rho_1}^M\middle\|\mu_{\sigma_1}^M\right) + (1- \lambda) Q_\alpha \left( \mu_{\rho_2}^M\middle\|\mu_{\sigma_2}^M\right) \\
     &\geq \lambda Q_\alpha^\mathcal{M} \left(\rho_1\middle\|\sigma_1\right) + (1- \lambda)  Q_\alpha^\mathcal{M} \left(\rho_2\middle\|\sigma_2\right) \, ,
\end{align}
where the first inequality stems from joint concavity of the classical quantity and the second one from the superadditivity of the infimum.


Since \(f(t):= \frac{1}{\alpha-1}\log(t)\) is a non-increasing and convex function, it follows that \(D_\alpha^\mathcal{M}\) is jointly convex in \((\rho,\sigma)\) for \(\alpha\in(0,1)\). Finally, \(D_1^\mathcal{M}\) is jointly convex as the supremum of a family of jointly convex functionals (cf.\ Property 3).

For \(\alpha\in(1,\infty)\), we can similarly use joint convexity of the classical \(Q_\alpha\) and the subadditivity of the supremum to show joint convexity of \( Q_\alpha^\mathcal{M} \). Since \(f(t):= \frac{1}{\alpha-1}\log(t) \) is non-decreasing and quasi-convex in this parameter range, it follows that \(D_\alpha^\mathcal{M}\) is jointly quasi-convex. Further, \(D_\infty^\mathcal{M}\) as the supremum of a family of jointly quasi-convex functions is jointly quasi-convex.
\end{proof}

Property 1 of Lemma \ref{Lem:General_Properties} shows that the measured R\'enyi divergence is indeed an information divergence in the information-theoretic sense. In the rest of the manuscript, we will focus for the most part on sets \(\mathcal{M}\) that are informationally complete \cite{Prugovecki}. These sets contain at least one measurement \(M\) such that for \( \rho\not=\sigma \), we have that \( \mu_{\rho}^M \not= \mu_{\sigma}^M\). Note, for instance, that since the set LO\((A:B)\) is informationally complete its supersets described in Sec.\ \ref{Sec:Local_Measurement_Sets} are as well. This allows us to strengthen Property 1 to \(D_{\alpha}^{\mathcal{M}}(\rho\|\sigma) > 0\) if \(\rho \not = \sigma\) for \(\rho,\sigma \in \mathcal{S}\), i.e.\ the measured R\'enyi divergences are faithful. We are then justified to interpret them as a measure of distinguishability for quantum states.

Moreover, the convexity properties of the measured R\'enyi divergence enable us to derive an equivalent characterization in terms of a continuous extension based on states that have full support. 
\begin{lemma}\label{Lem:Continous_Extension}
    Let \(\rho,\sigma\in\mathcal{S}\), \(\mathcal{M} \subseteq \text{ALL}\) and \(\alpha>0\). Then, we can express the measured R\'enyi divergence equivalently as
    \begin{equation}\label{Eq:Continous_Extension}
        D_\alpha^\mathcal{M} \left(\rho\|\sigma\right) = \sup_{\epsilon\in(0,1]} D_\alpha^\mathcal{M} \big((1-\epsilon)\rho+\epsilon\pi\big\|(1-\epsilon)\sigma+\epsilon\pi\big) \, ,
    \end{equation}
    where \(\pi:= 1_\dutchcal{H}/d\) is the completely mixed state.
\end{lemma}

\begin{proof}

In the following, we denote the objective function on the right-hand side of Eq.\ \eqref{Eq:Continous_Extension} as \(f_\alpha(\epsilon) := D_\alpha^\mathcal{M} \left(\rho_\epsilon\middle\|\sigma_\epsilon\right)\)
using the shorthand \(\rho_\epsilon := (1-\epsilon) \rho + \epsilon \pi \). Note that \(\rho_\epsilon\) has full support on \(\dutchcal{H}\) for \(\epsilon\in(0,1]\) independent of \(\rho\). It is also straightforward to verify that 
\begin{equation}
    \rho_{\lambda \epsilon_1 + (1-\lambda) \epsilon_2} = \lambda \rho_{\epsilon_1} + (1-\lambda) \rho_{\epsilon_2} \, .
\end{equation}

Let \(\alpha\in(0,1]\) and \(\epsilon_1,\epsilon_2 \in (0,1)\). Then, joint convexity of \(D_\alpha^\mathcal{M}\) implies
\begin{align}
    f_\alpha\left(\lambda\epsilon_1+(1-\lambda)\epsilon_2\right) &=  D_\alpha^\mathcal{M} \left(\lambda \rho_{\epsilon_1} + (1-\lambda) \rho_{\epsilon_2}\middle\|\lambda \sigma_{\epsilon_1} + (1-\lambda) \sigma_{\epsilon_2}\right) \\
    &\leq \lambda D_\alpha^\mathcal{M} \left( \rho_{\epsilon_1}\middle\|\sigma_{\epsilon_1}\right) + (1-\lambda) D_\alpha^\mathcal{M} \left( \rho_{\epsilon_2}\middle\|\sigma_{\epsilon_2}\right) = \lambda f_\alpha(\epsilon_1) + (1-\lambda) f_\alpha(\epsilon_2) \, .
\end{align}
Therefore, \(f_\alpha(\epsilon)\) is convex which implies it is continuous on the interval \((0,1)\). For \(\epsilon\in(0,1)\), we have \(f_\alpha(\epsilon)\geq f_\alpha(1)=0\), which further implies that \(f_\alpha(\epsilon)\) is monotone decreasing in \(\epsilon\). The supremum thus constitutes a limit \(\epsilon\to0\) in case of convergence. 

For \(\alpha\in(1,\infty)\), we have by Lemma \ref{Lem:General_Properties} that 
\begin{equation}
    Q_\alpha^\mathcal{M} \left(\lambda \rho_{\epsilon_1} + (1-\lambda) \rho_{\epsilon_2}\middle\|\lambda \sigma_{\epsilon_1} + (1-\lambda) \sigma_{\epsilon_2}\right) \leq \lambda Q_\alpha^\mathcal{M} \left( \rho_{\epsilon_1}\middle\|\sigma_{\epsilon_1}\right) + (1-\lambda) Q_\alpha^\mathcal{M} \left( \rho_{\epsilon_2}\middle\|\sigma_{\epsilon_2}\right)
\end{equation}
and thus \( g_\alpha(\epsilon):=\exp((\alpha-1)f_\alpha(\epsilon))\) is convex on \(\epsilon\in(0,1)\). Additionally, we have \(g_\alpha(\epsilon) \geq g_\alpha(1)=1\), which implies \(g_\alpha(\epsilon)\) is monotone decreasing on \((0,1)\). The composition with a continuous non-decreasing function preservers monotonicity and continuity. Thus, we can conclude that \(f_\alpha(\epsilon)\) is continuous and monotone decreasing  in \(\epsilon\) on \((0,1)\).

Finally, we have by Lemma \ref{Lem:General_Properties} that 
\begin{equation}
    D_\infty^\mathcal{M} \left(\rho\middle\|\sigma\right) = \sup_{\alpha\in(1,\infty)} D_\alpha^\mathcal{M} \left(\rho\middle\|\sigma\right) = \sup_{\alpha\in(1,\infty)} \sup_{\epsilon\in(0,1]} D_\alpha^\mathcal{M} \left(\rho_\epsilon\middle\|\sigma_\epsilon\right) = \sup_{\epsilon\in(0,1]} D_\infty^\mathcal{M} \left(\rho_\epsilon\middle\|\sigma_\epsilon\right) \, .
\end{equation}
\end{proof}

We remark that Lemma \ref{Lem:Continous_Extension} allows us to make the assumption that the states have full support and then lift any result to the general case via Eq.\ \eqref{Eq:Continous_Extension}.


\subsection{Superadditivity and Regularization}

Next, we discuss the additivity properties of the measured R\'enyi divergence on independent states. For this, we denote the family of i.i.d.\ states \(\left(\{\rho\}, \{\rho^{\otimes 2}\}, ...\right)\) as \(\boldsymbol{\rho}\) and point out that \(\rho^{\otimes n}\) is a state on the space \(\dutchcal{H}^{\otimes n}\). In order to discuss additivity questions then, we have to introduce a family of measurements \(\boldsymbol{\mathcal{M}} := (\mathcal{M}_1, \mathcal{M}_2, ...)\), where \(\mathcal{M}_n\) is a set of measurements on \(\dutchcal{H}^{\otimes n}\). Further, we assume that this family of measurement satisfies some mild regularity condition. Namely, for any \(k,l\) and any \(M_k \in \mathcal{M}_k, M_l \in \mathcal{M}_l\), we have \(M_k \otimes M_l \in \mathcal{M}_{k+l}\).
Note that this is satisfied by the locality-constrained measurement sets defined in Sec.\ \ref{Sec:Local_Measurement_Sets}, where we define e.g.\ \( \text{LO}_n(A:B) = \text{LO}(A_1 ... A_n : B_1 ... B_n)\).

Given this assumption, we can show a general superadditivity result which in turn allows us to define the regularized measured R\'enyi divergence of order \(\alpha\).
\begin{lemma}\label{Lem:Regularization}
    Let \(\boldsymbol{\mathcal{M}}\) be a family of measurements and \(\boldsymbol{\rho},\boldsymbol{\sigma}\) collections of i.i.d.\ states as defined above. The measured R\'enyi divergence of order \(\alpha>0\) is super-additive on i.i.d.\ states, i.e.\
    \begin{equation}
        D_\alpha^{\mathcal{M}_{k+l}}\left(\rho^{\otimes k+l}\middle\|\sigma^{\otimes k+l}\right) \geq D_\alpha^{\mathcal{M}_{k}}\left(\rho^{\otimes k}\middle\|\sigma^{\otimes k}\right) + D_\alpha^{\mathcal{M}_{l}}\left(\rho^{\otimes l}\middle\|\sigma^{\otimes l}\right) \, .
    \end{equation}
    Moreover, the regularized measured R\'enyi divergence is well-defined and given by 
    \begin{equation}\label{Eq:Regularization}
        D_\alpha^{\boldsymbol{\mathcal{M}}}(\boldsymbol{\rho}\|\boldsymbol{\sigma}) := \lim_{n\to\infty} \frac{1}{n} D_\alpha^{\mathcal{M}_n}\left(\rho^{\otimes n}\middle\|\sigma^{\otimes n}\right) = \sup_{n\in\mathbbm{N}} \frac{1}{n} D_\alpha^{\mathcal{M}_n}\left(\rho^{\otimes n}\middle\|\sigma^{\otimes n}\right) \, .
    \end{equation}
\end{lemma}

\begin{proof} 

First, note that a product measurement on product states induces a measure that satisfies
\begin{equation}
    \mu_{\rho_1 \otimes \rho_2}^{M_1 \otimes M_2} (x,y) = \tr[ (\rho_1 \otimes \rho_2)(M_1^x \otimes M_2^y)] = \tr[ \rho_1 M_1^z] \tr[ \rho_2 M_2^z] = \mu_{\rho_1}^{M_1}(x) \mu_{\rho_2}^{M_2}(y) \, .
\end{equation}

With this, we can show that
\begin{align}
    \sup_{M \in \mathcal{M}_{k+l}} D_\alpha \left(\mu_{\rho^{\otimes k+l}}^{M}\middle\|\mu_{\sigma^{\otimes k+l}}^M\right) &\geq \sup_{M_k \otimes M_l \in \mathcal{M}_{k+l}} D_\alpha \left(\mu_{\rho^{\otimes k+l}}^{M_k \otimes M_l}\middle\|\mu_{\sigma^{\otimes k+l}}^{M_k \otimes M_l} \right) \\
    &= \sup_{M_k \in \mathcal{M}_{k}} D_\alpha \left(\mu_{\rho^{\otimes k}}^{M_k}\middle\|\mu_{\sigma^{\otimes k}}^{M_k} \right) + \sup_{M_l \in \mathcal{M}_{l}} D_\alpha \left(\mu_{\rho^{\otimes l}}^{M_l}\middle\|\mu_{\sigma^{\otimes l}}^{M_l} \right) \, ,
\end{align}
where the inequality step follows by restricting to product measurements w.r.t.\ to a cut of the system in two blocks of \(k\) and \(l\) systems, respectively. The result then follows by the additivity of the classical quantity on independent measures \cite[Chapter 4]{Tomamichel}. Eq.\ \eqref{Eq:Regularization} then follows by applying Fekete's lemma.    
\end{proof}

Notice that the above proof can be adapted straightforwardly to show that 
\begin{equation}%
    D_\alpha^{\mathcal{M}_2} \left(\rho\otimes \tau\middle\|\sigma\otimes\omega\right) \geq D_\alpha^{\mathcal{M}_1} \left(\rho\middle\|\sigma\right) + D_\alpha^{\mathcal{M}_1} \left(\tau\middle\|\omega\right) \, ,
\end{equation}
i.e.\ superadditivity holds also on general product states.

In the case \(\mathcal{M}=\text{ALL}\), the regularized measured R\'enyi divergence becomes single-letter and is equal to the sandwiched R\'enyi divergence, i.e.\
\begin{equation}\label{Eq:Connection_Regularized_Sandwiched}
    D_\alpha^{\mathbf{ALL}}(\boldsymbol{\rho}\|\boldsymbol{\sigma}) = \tilde{D}_\alpha(\rho\|\sigma) \; \; \text{for} \; \; \alpha \geq \frac{1}{2} \, ,
\end{equation}
which is known in the literature as asymptotic achievability under measurements (cf.\ \cite[Appendix A]{Mosonyi}).\footnote{It is not known if this property holds for \(\alpha\in(0,1/2)\).} Note that this result implies that the unrestricted measured R\'enyi divergence is strictly superadditive on states for which the inequality in Eq.\ \eqref{Eq:Comparison_Quantum} is strict, i.e.\ for states with \([\rho,\sigma]\not=0\) and \(\tilde{D}_\alpha(\rho\|\sigma)<+\infty\). In the general case, however, it is unclear how to evaluate the regularized quantity. We show in Sec.\ \ref{Sec:DataHiding}, that for some important examples of states they become single-letter and can be computed explicitly. 

Lastly, note that the regularized quantity for the order \(\alpha=1\) already has found an operational interpretation in the context of quantum hypothesis testing under restricted measurements as the optimal Stein's exponent (see \cite[Theorem 16]{Brandao} and Sec.\ \ref{Sec:Hypothesis_Testing} for more details). Moreover, it was shown in \cite{Mosonyi} that the sandwiched Rényi divergence has an operational interpretation for \(\alpha > 1\) in the strong converse problem of unrestricted quantum hypothesis testing.


\subsection{Data-Processing Inequality}

An important property of the classical divergence is its data-processing inequality (DPI), i.e.\ the classical R\'enyi divergence is contractive under the application of classical channels. Similarly, the measured R\'enyi divergence satisfies a DPI under channels compatible with the respective measurement set. For this, we call a quantum channel \(\mathcal{G}\) compatible with a measurement set \(\mathcal{M}\) if it satisfies \(\mathcal{G}^\dagger(M) \in \mathcal{M} \) for all \(M \in \mathcal{M}\), where \(\mathcal{G}^\dagger\) is applied to the individual POVM elements.
\begin{lemma}\label{Lem:Modified_DPI}
    Let \(\rho,\sigma \in \mathcal{S}\), \(\mathcal{M} \subseteq \text{ALL}\) and \(\alpha>0\). Then, \(D_{\alpha}^{\mathcal{M}} \) is monotone decreasing under channels compatible with \(\mathcal{M}\) as defined above, i.e.\ we have 
    \begin{equation}
        D_{\alpha}^{\mathcal{M}} \left( \rho\middle\|\sigma \right) \geq D_{\alpha}^{\mathcal{M}} \left( \mathcal{G}(\rho)\middle\|\mathcal{G}(\sigma) \right) \, . 
    \end{equation} 
\end{lemma}

\begin{proof} 

The proof adapts the argument of \cite[Lemma 2]{Berta2}. Observe that by the definition of the adjoint of a quantum channel, we have
\begin{align}
    \mu_{\mathcal{G}(\rho)}^M (z) = \tr[\mathcal{G}(\rho) M^z ] = \tr[\rho \mathcal{G}^\dagger(M^z) ] = \mu_{\rho}^{\mathcal{G}^\dagger(M)} (z) \, .
\end{align}

The monotonicity then directly follows from
\begin{align}
    D_{\alpha}^{\mathcal{M}} \left( \mathcal{G}(\rho) \middle\| \mathcal{G}(\sigma) \right) = \sup_{M \in \mathcal{M}} D_\alpha \left( \mu_{\mathcal{G}(\rho)}^M\middle\|\mu_{\mathcal{G}(\sigma)}^M\right) = \sup_{M \in \mathcal{M}} D_\alpha \left( \mu_{\rho}^{\mathcal{G}^\dagger(M)}\middle\|\mu_{\sigma}^{\mathcal{G}^\dagger(M)}\right) \leq D_{\alpha}^{\mathcal{M}} \left( \rho\middle\|\sigma \right) \, ,
\end{align}
where the inequality step results from a superset argument using the compatibility assumption.
\end{proof}

Note that all quantum channels are compatible with the set \(\mathcal{M}=\text{ALL}\) and thus \(D_\alpha^\text{ALL}\) is contractive under general channels. Moreover, Lemma \ref{Lem:Modified_DPI} implies for the sets LOCC, SEP and PPT a DPI for the respective measured R\'enyi divergence under applications of LOCC channels. For the sets LOCC\(_1(A\to B)\), we can use the observation from \cite[Lemma 2]{Berta2} to conclude that \( D_\alpha^{\text{LOCC}_1} \) is monotone under LOCC\(_1(A\to B)\) channels, i.e.\ local operations supported by one-way classical communication from \(A\) to \(B\). Similarly, \(D_\alpha^{\text{LO}}\) is contractive under local channels.


\subsection{Unitary and Isometric Invariance}\label{Sec:Isometric_Invariance}

An immediate consequence of Lemma \ref{Lem:Modified_DPI}, is that \(D_\alpha^\text{ALL}\) is unitarily invariant, i.e.\
\begin{equation}
    D_\alpha^\text{ALL} \left(U\rho U^\dagger\middle\|U\sigma U^\dagger\right) = D_\alpha^\text{ALL} \left(\rho\middle\|\sigma\right) \; \; \text{for} \; \;  U \in \mathcal{U} \, .
\end{equation}
To see this, note that the set of all measurements is invariant under unitary conjugation. A general set \(\mathcal{M}\), however, does not satisfy this and as a result \(D_\alpha^\mathcal{M}\) is in general not unitarily invariant. 

Nevertheless, note that the locality-constrained measurement sets of Sec.\ \ref{Sec:Local_Measurement_Sets} are closed under local unitary conjugation, i.e.\ operators of the form \(U_A \otimes U_B\) with \(U_A \in \mathcal{U}_A\) and \(U_B \in \mathcal{U}_B\). Therefore, the locally-measured R\'enyi divergences are invariant under local unitaries. Moreover, we have by the same argument the stronger invariance under local isometries for the locally-measured R\'enyi divergences.


\subsection{Boundedness}\label{Sec:Boundedness}

It is straightforward to verify that the measured R\'enyi divergence is bounded whenever the quantum divergence is bounded. This follows immediately from the observation that for all \(\alpha>0\), we have
\begin{equation}
    D^{\mathcal{M}}_\alpha \left(\rho \middle\|\sigma \right) \leq \bar{D}_\alpha \left(\rho \middle\|\sigma \right)
\end{equation}
by the monotonicity under measurements of the latter. In the case \(\mathcal{M} = \text{ALL}\), the reverse statement also holds, i.e.\ \(\rho \perp \sigma \) or \(\rho\not\ll\sigma\) for \(\alpha\in(1,\infty)\) implies that \(D^\text{ALL}_\alpha(\rho\|\sigma) = +\infty\). To verify this, consider first the case \(\rho\perp\sigma\). We can define a projective 2-outcome measurement using \(P_{\text{supp}(\rho)}\) and its orthogonal complement \(P_{\text{supp}(\rho)}^\perp := 1_\dutchcal{H} - P_{\text{supp}(\rho)}\), where \(P_{\text{supp}(\rho)}\) denotes the projector on the support of \(\rho\). Observe that the resulting measures 
\begin{align}
    \mu_\rho^P &= \{ 1, 0 \} && \text{and} & \mu_\sigma^P &= \{ 0, 1 \}
\end{align}
are orthogonal and thus \(D^\text{ALL}_\alpha(\rho\|\sigma) = +\infty\). Similarly, if \(\alpha\in(1,\infty)\) and \(\rho\not\ll\sigma\) we have by definition \(\text{supp}(\rho) \cap \ker(\sigma)\not=\emptyset\) and by considering a 2-outcome PVM based on a projector on that subspace, we can reach the same conclusion as before.

Let us emphasize, however, that these conclusions do not hold for a general set \(\mathcal{M}\), as the PVMs we defined above are not necessarily part of the set \(\mathcal{M}\) under consideration. This highlights that the distinguishing power of a restricted measurement set may be significantly diminished compared to the unrestricted case. The prototypical example is the phenomenon of quantum data hiding (cf.\ \cite{Terhal, Eggeling, DiVincenzo, DiVincenzo2}), which asserts that there exist quantum states that can be perfectly distinguished if all measurements are available but are almost indistinguishable using only local measurements and classical communication. In Sec.\ \ref{Sec:DataHiding}, we show that for such states the locally-measured R\'enyi divergences are bounded.


\subsection{Remarks on Optimal Measurement}\label{Sec:Optimal_Measurements}

We can also make some general observations on the optimal measurement in Eq.\ \eqref{Eq:Measured_Renyi_Divergence}. First, note that for a given measurement coarse-graining can only reduce the divergence due to the classical data-processing inequality. This implies for the PVM classes P-LO and P-LOCC\(_1\) that we may restrict the optimization w.l.o.g.\ to rank-1 projectors since any PVM can be fine-grained into a rank-1 PVM. 

We are able to say more if we assume additional mathematical structure in the measurement set. For this, we define the convex combination of two POVMs \(M_1\) and \(M_2\) over \(\mathcal{Z}\) element-wise as 
\begin{equation}
    M^z = \lambda M_1^z + (1-\lambda) M_2^z \, . \footnote{This definition can be extended to two POVMs over different alphabets \(\mathcal{Z}_1\) and \(\mathcal{Z}_2\) by mapping them on POVMs over a common alphabet \(\mathcal{Z}\) with \(|\mathcal{Z}| = \max \{|\mathcal{Z}_1|,|\mathcal{Z}_2|\}\) extending the smaller set by zero elements.}
\end{equation}
By linearity of the trace, we then have for the resulting measures
\begin{equation}
    \mu_{\rho}^{\lambda M_1 + (1-\lambda)M_2} = \lambda \mu_\rho^{M_1} + (1-\lambda) \mu_\rho^{M_2} \, ,
\end{equation}
where the addition is understood element-wise. Now, if the set \(\mathcal{M}\) is convex then \( D_\alpha \left(\mu_\rho^M\middle\|\mu_\sigma^M\right) \) as a function of the measurement is convex for \(\alpha \in (0,1]\) and quasi-convex for \(\alpha \in (1,\infty]\). This is a direct result of the joint (quasi-)convexity of the classical quantity. Thus, given a convex measurement set, the optimization can be restricted to extremal elements w.r.t.\ to that set.


\subsection{Pinsker Inequality}

Let us close out this section, by mentioning that the measured R\'enyi divergences for orders \(\alpha \in (0,1]\) are connected to the measured Schatten norms of \cite{Matthews} via a Pinsker-type inequality
\begin{equation}\label{Eq:Pinsker_Inequality}
    D_\alpha^\mathcal{M} \left( \rho \middle\| \sigma \right) \geq \frac{\alpha}{2} \norm{\rho - \sigma}_\mathcal{M}^2 \, .
\end{equation}
This follows directly by lifting the corresponding classical result from \cite[Corollary 6]{Gilardoni},
\begin{equation}
    D_\alpha(\mu\|\nu) \geq \frac{\alpha}{2} TV^2(\mu,\nu) \, ,
\end{equation}
where \(TV(\mu,\nu) := \sum_{z\in\mathcal{Z}} |\mu(z) - \nu(z)|\) is the total variational distance. 


\section{Variational Characterizations}\label{Sec:Variational}

Variational characterizations of divergences offer a versatile tool in various applications of quantum information theory. Building on the work of Berta and Tomamichel in \cite[Lemma 3]{Berta2}, who established a variational bound on the measured relative entropy in order to prove entanglement monogamy inequalities, we extend their results to the measured R\'enyi divergences. This may also be seen as a generalization of \cite[Lemma 3]{Berta1}, which derived an exact variational characterization in the case \(\mathcal{M}=\text{ALL}\). We then treat the locality-constrained measurement sets in detail and derive tight variational characterizations for the classes LO and LOCC\(_1\). We further take a closer look at the measured max-divergence and derive a dual variational characterization. Lastly, we additionally discuss an infinite-dimensional extension of our main result on variational characterizations in App.\ \ref{App:Infinite-Dimensional_Version}.


\subsection{Generic Upper Bound}\label{Sec:Variational_Upper_Bound}

In the following, we will prove the main technical result, which is a generic upper bound on \(D_\alpha^\mathcal{M}\) in terms of a variational formula. For this, let us first define the cone associated to the set \(\mathcal{M}\) as
\begin{align}\label{Def:Measurement_Cone}
    C_\mathcal{M} := \bigcup_{M \in \mathcal{M}} \text{cone}\left( M \right) \, ,
\end{align}
where \(\text{cone}\left( M \right)\) denotes the conical hull of the elements of the POVM \(M\), i.e., 
\begin{equation}
    \text{cone}(M) := \left\{ \sum_{z \in \mathcal{Z}} \lambda_z M^z : \lambda_z \geq 0 \right\} \, .
\end{equation}

We can then summarize our result compactly as 
\begin{equation}\label{Eq:Compact_Variational}
    D_\alpha^\mathcal{M} \left( \rho \middle\| \sigma \right) \leq V_\alpha^\mathcal{M} (\rho,\sigma) \, ,
\end{equation}
where \(V_\alpha^\mathcal{M} (\rho,\sigma)\) denotes the variational bound given by
\begin{equation}
    V_\alpha^\mathcal{M} (\rho,\sigma) := \sup_{\omega > 0} \bigg\{ \nu_\alpha \left( (\rho,\sigma); \omega \right) \; \bigg| \; \omega \in C_\mathcal{M} \bigg\} = \sup_{\omega > 0} \bigg\{ \eta_\alpha \left( (\rho,\sigma); \omega \right) \; \bigg| \; \omega \in C_\mathcal{M} \bigg\} \, .
\end{equation}
Here, we introduced some compact notation for the two types of objective functions we get. The first class of objective function is given by 
\begin{equation}\label{Eq:Additive_Expression}
    \nu_\alpha \big( (\rho,\sigma); \omega \big) := \left\{ \begin{array}{ll}
        \frac{1}{\alpha-1} \log\left(\alpha \tr[\rho \omega] + (1-\alpha) \tr[\sigma \omega^{\frac{\alpha}{\alpha-1}}] \right) & \text{for} \; \alpha \in (0,1/2) \vspace{0.3em} \\
        \frac{1}{\alpha-1} \log\left( \alpha \tr[\rho \omega^{\frac{\alpha-1}{\alpha}}] + (1-\alpha) \tr[ \sigma \omega] \right) & \text{for} \; \alpha \in [1/2,1) \cup (1,\infty) \vspace{0.3em} \\
        \tr[ \rho \log \omega] + 1 - \tr[ \sigma \omega] & \text{for} \; \alpha = 1 \vspace{0.3em} \\ 
        \log \tr[\rho \omega]  + 1 - \tr[ \sigma \omega] & \text{for} \; \alpha = \infty \\ 
    \end{array} \right. 
\end{equation}
and the second is defined as
\begin{equation}\label{Eq:Multiplicative_Expression}
    \eta_\alpha \big( (\rho,\sigma); \omega \big) := \left\{ \begin{array}{ll}
        \frac{1}{\alpha-1} \log\left( \tr[\rho \omega]^\alpha \tr[\sigma \omega^{\frac{\alpha}{\alpha-1}}]^{1-\alpha} \right) & \text{for} \; \alpha \in (0,1/2)  \vspace{0.3em} \\
        \frac{1}{\alpha-1} \log\left( \tr[\rho \omega^{\frac{\alpha-1}{\alpha}}]^\alpha \tr[ \sigma \omega]^{1-\alpha} \right) & \text{for} \; \alpha \in [1/2,1) \cup (1,\infty) \vspace{0.3em} \\
        \tr[ \rho \log \omega] - \log \tr[ \sigma \omega] & \text{for} \; \alpha = 1 \vspace{0.3em} \\
        \log \tr[\rho \omega] - \log \tr[ \sigma \omega] & \text{for} \; \alpha = \infty
    \end{array} \right. \, .
\end{equation}

We remark that the objective function \( \nu_\alpha\) is concave in \(\omega\). Furthermore, the second class of objective function \( \eta_\alpha\) exhibits a scaling invariance in \(\omega\). The variational characterizations thus have a particularly appealing form. Moreover, the variational program corresponding to \(\eta_\alpha\) can be interpreted as a generalization of Alberti's theorem for the fidelity \cite{Alberti} to general orders \(\alpha\) and sets \(\mathcal{M}\). Utilizing the connection of the measured R\'enyi divergence at \(\alpha = 1/2\) to the measured fidelity from Eq.\ \eqref{fidelity_connect}, we additionally get a direct generalization of this theorem as
\begin{equation}
    F^\mathcal{M}(\rho,\sigma) \geq \inf_{\stackrel{\omega > 0}{\omega \in C_\mathcal{M}}} \sqrt{\tr[\rho \omega^{-1}] \tr[\sigma \omega] } \, .
\end{equation}
The following subsections are devoted to proving this result (see also App.\ \ref{App:Infinite-Dimensional_Version} for a treatment of the infinite-dimensional case). In Sec.\ \ref{Sec:Discussion_Proof}, we additionally give some insights on sufficient conditions for when the variational bound is tight.


\subsubsection{Proof for General Orders}

Our main theorem gives variational bounds on \(Q_\alpha^\mathcal{M}\). Recalling Eq.\ \eqref{Def:Measured_Q_Connection}, this then immediately gives the claimed upper bound on \(D_{\alpha}^{\mathcal{M}}\) for \(\alpha \in (0,1) \cup (1,\infty)\). 
\begin{theorem}\label{Thm:Variational_Q}
For \(\rho, \sigma \in \mathcal{S} \) and nonempty \(\mathcal{M} \subseteq \text{ALL}\), we have
\begin{enumerate}
    \item for \(\alpha \in \left(0,1/2\right)\)
    \begin{align}
        Q_{\alpha}^{\mathcal{M}} \left(\rho \middle\| \sigma \right) \geq& \inf_{\stackrel{\omega > 0}{\omega \in C_\mathcal{M}}} \alpha \tr[ \rho \omega ] + (1-\alpha) \tr[\sigma \omega^{\frac{\alpha}{\alpha-1}}] = \inf_{\stackrel{\omega > 0}{\omega \in C_\mathcal{M}}} \tr[ \rho \omega ]^{\alpha} \tr[\sigma \omega^{\frac{\alpha}{\alpha-1}}]^{1-\alpha} \, ,
    \end{align} 

    \item for \(\alpha \in \left[1/2,1\right)\)
    \begin{align}
        Q_{\alpha}^{\mathcal{M}}\left(\rho \middle\| \sigma \right) \geq& \inf_{\stackrel{\omega > 0}{\omega \in C_\mathcal{M}}} \alpha \tr[ \rho \omega^{\frac{\alpha-1}{\alpha}}] + (1-\alpha) \tr[\sigma \omega] = \inf_{\stackrel{\omega > 0}{\omega \in C_\mathcal{M}}} \tr[ \rho \omega^{\frac{\alpha-1}{\alpha}}]^{\alpha} \tr[\sigma \omega]^{1-\alpha} \, ,
    \end{align}

    \item for \(\alpha \in (1,\infty)\)
    \begin{align}
        Q_{\alpha}^{\mathcal{M}} \left( \rho \middle\| \sigma \right) \leq& \sup_{\stackrel{\omega > 0}{\omega \in C_\mathcal{M}}} \alpha \tr[ \rho \omega^{\frac{\alpha-1}{\alpha}}] + (1-\alpha) \tr[\sigma \omega] = \sup_{\stackrel{\omega > 0}{\omega \in C_\mathcal{M}}} \tr[ \rho \omega^{\frac{\alpha-1}{\alpha}}]^{\alpha} \tr[\sigma \omega]^{1-\alpha} \, .
    \end{align}
\end{enumerate}

\end{theorem}

\begin{proof}

The first part of the proof adapts the proof idea of \cite[Lemma 3]{Berta2}.

\underline{1) Case of \(\rho,\sigma\) with full support}

Let us assume first that both \(\rho\) and \(\sigma\) have full support. We then have for an arbitrary \(M\in\mathcal{M}\) that \( \mu_\rho^M(z), \mu_\sigma^M(z) > 0 \) for all \(z\in\mathcal{Z}\) and thus \(\mu_\rho^M(z) / \mu_\sigma^M(z) \in (0,+\infty)\). Further, let us define the shorthand symbol \( Q_\alpha^M \left(\rho \middle\| \sigma \right) := Q_{\alpha} \left(\mu_\rho^M\middle\|\mu_\sigma^M\right) \) for a fixed \(M \in \mathcal{M}\). 

Let \( \alpha \in (0,1/2) \). Here, we can bound \(Q_\alpha^M \) for any \(M \in \mathcal{M}\) as 
\begin{align}
    Q_\alpha^M \left(\rho \middle\| \sigma \right) = \sum_z \mu_\sigma^M(z) \left( \frac{\mu_\rho^M(z)}{\mu_\sigma^M(z)} \right)^{\alpha} &= \tr \left[ \sigma \sum_z \sqrt{M^z} \left( \frac{\mu_\rho^M(z)}{\mu_\sigma^M(z)} \right)^{\alpha} \sqrt{M^z}\right] \\
    &\geq \tr \left[ \sigma \left(  \sum_z \left( \frac{\mu_\rho^M(z)}{\mu_\sigma^M(z)} \right)^{\alpha-1} M^z  \right)^{\frac{\alpha}{\alpha-1} } \right] \, ,
\end{align}
where the inequality follows from an application of the operator Jensen inequality \cite[Theorem 2.1]{Hansen} using the operator convexity of \(f(t) := t^\frac{\alpha}{\alpha-1}\) for \(\alpha \in (0,1/2)\) \cite[Section V]{Bhatia}. For any \(M \in \mathcal{M}\), we further have by definition of \(C_\mathcal{M}\) that
\begin{equation}
    \omega_M := \sum_z \left( \frac{\mu_\rho^M(z)}{\mu_\sigma^M(z)} \right)^{\alpha-1} M^z \in C_\mathcal{M} \, .
\end{equation}
Observe that \(\omega_M > 0\)\footnote{This follows from the completeness relation of POVMs and the strict positivity of all coefficients.} and \( \tr[ \rho \omega_M] = Q_{\alpha}^{M} \left(\rho \middle\| \sigma \right) \). With this, we obtain the desired lower bound on \(Q_\alpha^\mathcal{M}\) as follows
\begin{align}
    Q_{\alpha}^{\mathcal{M}} \left(\rho \middle\| \sigma \right) = \inf_{M \in \mathcal{M}} \alpha Q_{\alpha}^{M}\left(\rho \middle\| \sigma \right) + (1-\alpha) Q_{\alpha}^{M}\left(\rho \middle\| \sigma \right) &\geq \inf_{M \in \mathcal{M}} \alpha \tr[ \rho \omega_M] + (1-\alpha) \tr[ \sigma \omega_M^\frac{\alpha}{\alpha-1} ] \\
    &\geq \inf_{\stackrel{\omega > 0}{\omega \in C_\mathcal{M}}} \alpha \tr[ \rho \omega] + (1-\alpha) \tr[ \sigma \omega^\frac{\alpha}{\alpha-1} ] \, ,
\end{align}  
where the final inequality step follows by a superset argument. 

We can apply the same proof strategy to show the claim for the remaining cases of \(\alpha\). For \(\alpha \in [1/2,1)\), we can bound \(Q_\alpha^M\) for any \(M \in \mathcal{M}\) as 
\begin{align}
    Q_\alpha^M \left(\rho \middle\| \sigma \right) = \sum_z \mu_\rho^M(z) \left( \frac{\mu_\rho^M(z)}{\mu_\sigma^M(z)} \right)^{\alpha-1} &= \tr \left[ \rho \sum_z \sqrt{M^z} \left( \frac{\mu_\rho^M(z)}{\mu_\sigma^M(z)} \right)^{\alpha-1} \sqrt{M^z}\right] \\
    &\geq \tr \left[ \rho \left(  \sum_z \left( \frac{\mu_\rho^M(z)}{\mu_\sigma^M(z)} \right)^{\alpha} M^z  \right)^{\frac{\alpha-1}{\alpha} } \right] \label{Eq:Alpha_Bound}\, ,
\end{align}
where the last line follows from the operator convexity of \(f(t) := t^\frac{\alpha-1}{\alpha}\) for \(\alpha \in [1/2,1)\) \cite[Section V]{Bhatia}. For any \(M \in \mathcal{M}\), we further have that
\begin{equation}
    \omega_M := \sum_z \left( \frac{\mu_\rho^M(z)}{\mu_\sigma^M(z)} \right)^{\alpha} M^z \in C_\mathcal{M}
\end{equation}
with \(\omega_M>0\) and \( \tr[ \sigma \omega_M] = Q_{\alpha}^{M}\left(\rho \middle\| \sigma \right)\). This then allows us to lower bound \(Q_\alpha^\mathcal{M}\) as
\begin{align}
    Q_{\alpha}^{\mathcal{M}} \left(\rho \middle\| \sigma \right) &\geq \inf_{M \in \mathcal{M}} \alpha \tr[ \rho \omega_M^\frac{\alpha-1}{\alpha}] + (1-\alpha) \tr[ \sigma \omega_M ] \geq \inf_{\stackrel{\omega > 0}{\omega \in C_\mathcal{M}}} \alpha \tr[ \rho \omega^\frac{\alpha-1}{\alpha}] + (1-\alpha) \tr[ \sigma \omega ] \, . 
\end{align}  

Finally, for \(\alpha \in (1,\infty) \) the function \( f(t) := t^\frac{\alpha-1}{\alpha} \) is operator concave \cite[Section V]{Bhatia} and Eq.\ \eqref{Eq:Alpha_Bound} becomes an upper bound on \(Q_\alpha^M\) in this parameter range. Defining \(\omega_M\) as for \(\alpha \in [1/2,1) \) then leads to the desired upper bound
\begin{align}
    Q_{\alpha}^{\mathcal{M}} \left(\rho \middle\| \sigma \right) \leq \sup_{\stackrel{\omega > 0}{\omega \in C_\mathcal{M}}}  \alpha \tr[ \rho \omega^\frac{\alpha-1}{\alpha}] + (1-\alpha) \tr[ \sigma \omega ] \, .
\end{align}  

\underline{2) Extension to general case}

Next, we employ Lemma \ref{Lem:Continous_Extension} to lift the obtained bounds to states that do not have full support. Accordingly, we consider full-support extensions \( \rho_\epsilon := (1-\epsilon) \rho + \epsilon \pi \) for \(\rho\in\mathcal{S}\). Additionally, let us introduce a shorthand symbol for the obtained objective function
\begin{equation}\label{Eq:Objective_Function}
    f_\alpha \big( (\rho,\sigma),\omega \big) :=  \left\{ \begin{array}{ll}
         \alpha \tr[\rho \omega] + (1-\alpha) \tr[\sigma \omega^{\frac{\alpha}{\alpha-1}}] & \text{for} \; \alpha \in (0,1/2) \\
         \alpha \tr[\rho \omega^{\frac{\alpha-1}{\alpha}}] + (1-\alpha) \tr[ \sigma \omega] & \text{for} \; \alpha \in [1/2,1) \cup (1,\infty)
    \end{array} \right. \, .
\end{equation}
Note that \(f_\alpha\) is a linear function in the pair \((\rho,\sigma)\). For \(\alpha \in (0,1)\), we define
\begin{equation}
    g_\alpha(\epsilon) := \inf_{\stackrel{\omega > 0}{\omega \in C_\mathcal{M}}} f_\alpha \big( (\rho_\epsilon, \sigma_\epsilon), \omega \big)
\end{equation}
and observe that for \(\epsilon_1,\epsilon_2 \in (0,1)\) we have
\begin{align}
    g_\alpha(\lambda \epsilon_1 + (1-\lambda) \epsilon_2) &= \inf_{\stackrel{\omega > 0}{\omega \in C_\mathcal{M}}} \lambda f_\alpha \big( (\rho_{\epsilon_1}, \sigma_{\epsilon_1}), \omega \big) + (1-\lambda) f_\alpha \big( (\rho_{\epsilon_2}, \sigma_{\epsilon_2}), \omega \big)  \\
    &\geq  \lambda \inf_{\stackrel{\omega > 0}{\omega \in C_\mathcal{M}}} f_\alpha \big( (\rho_{\epsilon_1}, \sigma_{\epsilon_1}), \omega \big) + (1-\lambda) \inf_{\stackrel{\omega > 0}{\omega \in C_\mathcal{M}}} f_\alpha \big( (\rho_{\epsilon_2}, \sigma_{\epsilon_2}), \omega \big) \label{Eq:Extension_Bpund} \, ,
\end{align}
where we used the superadditivity of the infimum in the inequality step. Therefore, \(g_\alpha(\epsilon)\) is a concave function in \(\epsilon\) on the interval \((0,1)\).

Note that at \(\epsilon=1\), we have \(\rho_1 = \sigma_1 = \pi\). Further, since \(\omega > 0\) we can parameterize the optimization variable by the spectral theorem as \(\omega = \sum_{i=1}^{d} \lambda_i P_i\), where \(\lambda_i > 0\) are the eigenvalues of \(\omega\) corresponding to the eigenvectors given by rank-1 projectors \(P_i\). With this, we have for an arbitrary \(\rho\in\mathcal{S}\) that
\begin{align}
    f_\alpha \big( (\rho,\rho), \omega \big) = \left\{ \begin{array}{ll}
         \sum_{i=1}^d \alpha \lambda_i \tr[\rho P_i] + (1-\alpha) \lambda_i^\frac{\alpha}{\alpha-1} \tr[\rho P_i] & \text{for} \; \alpha \in (0,1/2) \\
         \sum_{i=1}^d \alpha \lambda_i^\frac{\alpha-1}{\alpha} \tr[\rho P_i] + (1-\alpha) \lambda_i \tr[\rho P_i] & \text{for} \; \alpha \in [1/2,1)
    \end{array} \right. \, .
\end{align}
Minimizing both expressions w.r.t.\ \(\lambda_i\) yields the optimizer \(\omega^\star = 1_\dutchcal{H}\) with the corresponding optimal value \(f \big( (\rho,\rho), 1_\dutchcal{H} \big) = \tr[\rho] = 1 \). Note that \(1_\dutchcal{H} \in C_\mathcal{M}\) holds for all nonempty sets \(\mathcal{M}\) and thus \(g_\alpha(1)=1\). Since \(g_\alpha(\epsilon)\) is upper-bounded by \(1\),\footnote{For any \(\rho,\sigma\in\mathcal{S}\), we can pick \(\omega^\star = 1_\dutchcal{H}\) to obtain the upper bound \(f_\alpha ((\rho,\sigma),1_\dutchcal{H}) = \alpha \tr[\rho] + (1-\alpha)\tr[\sigma] = 1\).} we conclude that \(g_\alpha(\epsilon)\) is a monotone increasing function in \(\epsilon\). 

With this, we obtain
\begin{equation}
    Q_{\alpha}^{\mathcal{M}}(\rho \| \sigma) = \inf_{\epsilon \in (0,1]} Q_{\alpha}^{\mathcal{M}}(\rho_\epsilon \| \sigma_\epsilon ) \geq \inf_{\epsilon \in (0,1]} \inf_{\stackrel{\omega > 0}{\omega \in C_\mathcal{M}}} f_\alpha \big( ( \rho_\epsilon, \sigma_\epsilon ), \omega \big) = \inf_{\stackrel{\omega > 0}{\omega \in C_\mathcal{M}}} f_\alpha \big( (\rho, \sigma), \omega \big) \, ,
\end{equation}
where the inequality step applies the bound for states with full support. 

For \(\alpha \in (1,\infty)\), we similarly define 
\begin{equation}
    g_\alpha(\epsilon) := \sup_{\stackrel{\omega > 0}{\omega \in C_\mathcal{M}}} f_\alpha \big( (\rho_\epsilon, \sigma_\epsilon), \omega \big) \, .
\end{equation}
By the subadditivity of the supremum, we can conclude as above that \(g_\alpha(\epsilon)\) is a convex function on \((0,1)\). An analogous argument to the one above then yields that \(g_\alpha(1) = 1 \) provides a lower bound on \(g_\alpha(\epsilon)\). Thus, \(g_\alpha(\epsilon)\) is monotone decreasing in \(\epsilon\). The desired claim then follows by
\begin{equation}
    Q_{\alpha}^{\mathcal{M}}(\rho \| \sigma) = \sup_{\epsilon \in (0,1]} Q_{\alpha}^{\mathcal{M}}(\rho_\epsilon \| \sigma_\epsilon ) \leq \sup_{\epsilon \in (0,1]} \sup_{\stackrel{\omega > 0}{\omega \in C_\mathcal{M}}} f_\alpha \big( ( \rho_\epsilon, \sigma_\epsilon ), \omega \big) = \sup_{\stackrel{\omega > 0}{\omega \in C_\mathcal{M}}} f_\alpha \big( (\rho, \sigma), \omega \big) \, .
\end{equation}

\underline{3) Equivalence of expressions}

It remains to show the equivalence of the given variational expressions. For this, let us note that for each positive \(\omega \in C_\mathcal{M}\) there exists a \(\lambda > 0 \) such that \(\lambda \omega \in C_\mathcal{M} \cap \mathcal{S}\) by the cone property. We can thus reformulate the optimization problem for \(\alpha \in (0,1)\) as follows 
\begin{equation}
     \inf_{\stackrel{\omega > 0}{\omega \in C_\mathcal{M}}} f_\alpha \big( (\rho,\sigma), \omega) = \inf_{\lambda > 0} \inf_{\stackrel{\omega > 0}{\omega \in C_\mathcal{M} \, \cap \, \mathcal{S}} } f_\alpha \big( (\rho,\sigma), \lambda \omega)
\end{equation}
and for \(\alpha \in (1,\infty)\) we have the same expression with suprema instead of infima. 

We now carry out the optimization over the scaling parameter \(\lambda\) analytically. For \(\alpha \in (0,1/2)\), we obtain the global optimizer \( \lambda^\star := \tr[\rho \omega]^{\alpha-1} \tr[\sigma \omega^{\frac{\alpha}{\alpha-1}}]^{1-\alpha}  >0 \). The remaining variational expression is given by 
\begin{equation}
    \inf_{\stackrel{\omega > 0}{\omega \in C_\mathcal{M} \cap \mathcal{S}} } \tr[\rho \omega]^\alpha \tr[\sigma \omega^{\frac{\alpha}{\alpha-1}}]^{1-\alpha} \, .
\end{equation}
For \(\alpha \in [1/2,1) \cup (1,\infty)\), the optimization over \(\lambda\) yields the optimizer \( \lambda^\star :=  \tr[\rho \omega^\frac{\alpha-1}{\alpha}]^\alpha \tr[\sigma \omega]^{-\alpha} >0 \). The remaining variational expression is then given for \(\alpha \in [1/2,1)\) by 
\begin{equation}
    \inf_{\stackrel{\omega > 0}{\omega \in C_\mathcal{M} \cap \mathcal{S}}} \tr[\rho \omega^{\frac{\alpha-1}{\alpha}}]^\alpha \tr[\sigma \omega]^{1-\alpha}
\end{equation}
and for \(\alpha \in (1,\infty)\) by the same expression but with the infimum replaced by a supremum.

The resulting objective functions are invariant under scaling of \(\omega\), so we may relax the constraint and optimize over the whole cone \(C_\mathcal{M}\) without changing the optimal value. This completes the proof.
\end{proof}


\subsubsection{Proofs for Extended Orders}

We complete our proof of the claim in Eq.\ \eqref{Eq:Compact_Variational} via two corollaries of Theorem \ref{Thm:Variational_Q} that extend the result to the orders \(\alpha=1\) and \(\alpha=\infty\). At \(\alpha=1\), we reproduce the bound obtained in \cite[Lemma 3]{Berta2} for the measured relative entropy. We also give an equivalent characterization that was not shown there. 
\begin{corollary}\label{Thm:Variational_Umegaki}
For \(\rho,\sigma \in \mathcal{S}\) and nonempty \(\mathcal{M} \subseteq \text{ALL}\), we have
\begin{align}
    D^{\mathcal{M}} \left(\rho \middle\| \sigma \right) \leq  \sup_{\stackrel{\omega > 0}{\omega \in C_\mathcal{M}}} \tr[\rho \log \omega] + 1 - \tr[\sigma \omega] = \sup_{\stackrel{\omega > 0}{\omega \in C_\mathcal{M}}} \tr[\rho \log \omega] - \log \tr[\sigma \omega] \, .
\end{align}
\end{corollary}
\begin{proof}
    Slightly modifying the proof for Property 3 in Lemma \ref{Lem:General_Properties} allows us to write
    \begin{align}
        D^\mathcal{M} \left( \rho \middle\| \sigma \right) = \sup_{\alpha \in [1/2,1)} D_\alpha^\mathcal{M} \left( \rho \middle\| \sigma \right) &\leq \sup_{\alpha \in [1/2,1) } \sup_{\stackrel{\omega > 0}{\omega \in C_\mathcal{M}}} \frac{1}{\alpha -1 } \log \left( \tr[\rho \omega^{\frac{\alpha - 1}{\alpha}}]^\alpha \tr[\sigma \omega]^{1-\alpha} \right) \\
        &= \sup_{\alpha \in [1/2,1) } \sup_{\stackrel{\omega > 0}{\omega \in C_\mathcal{M}}} \frac{\alpha}{\alpha -1} \log \tr[ \rho \omega^{\frac{\alpha-1}{\alpha}}] - \log \tr[\sigma \omega ] \, ,
    \end{align}
    where we used Theorem \ref{Thm:Variational_Q} for the inequality step. By the spectral theorem, we can decompose the optimization variable as \(\omega = \sum_{i=1}^d \lambda_i P_i\). Then, for \(\alpha \in [1/2,1)\) we have
    \begin{align}
        \frac{\alpha}{\alpha -1} \log \tr[ \rho \omega^{\frac{\alpha-1}{\alpha}}] &= \frac{\alpha}{\alpha -1} \log( \sum_{i=1}^d \lambda_i^{\frac{\alpha-1}{\alpha}} \tr[\rho P_i] ) \\
        &\leq \frac{\alpha}{\alpha -1} \sum_{i=1}^d \log( \lambda_i^{\frac{\alpha-1}{\alpha}}) \tr[\rho P_i] = \tr[ \rho \log\omega] \, ,
    \end{align}
where the inequality step follows from the convexity of \(f(t):= \frac{\alpha}{\alpha-1}\log(t)\) using the Jensen inequality and \(\sum_i \tr[\rho P_i] = \tr[\rho] = 1\). With this, we conclude 
\begin{equation}
    D^{\mathcal{M}} \left(\rho \middle\| \sigma \right) \leq 
    \sup_{\stackrel{\omega > 0}{\omega \in C_\mathcal{M}}} \tr[ \rho \log \omega ] - \log \tr[\sigma \omega ] \, .
\end{equation}

The equivalence of the two variational expressions can be shown similar as it was done in Theorem \ref{Thm:Variational_Q}. We use the same parameterization for the optimization variable to obtain
\begin{equation}
    \sup_{\stackrel{\omega > 0}{\omega \in C_\mathcal{M}}} \tr[\rho \log \omega ] +1 - \tr[\sigma \omega] = \sup_{\lambda > 0} \sup_{\stackrel{\omega > 0}{\omega \in C_\mathcal{M} \, \cap \, \mathcal{S}} }\tr[\rho \log( \lambda \omega) ] +1 - \tr[\sigma (\lambda \omega)]
\end{equation} 
A straightforward calculation shows that the expression is maximized for \(\lambda^\star = \frac{1}{\tr[\sigma \omega]} > 0\). This results in the remaining variational expression
\begin{equation}
    \sup_{\stackrel{\omega > 0}{\omega \in C_\mathcal{M} \cap \mathcal{S}} } \tr[\rho \log \omega ] - \log \tr[\sigma \omega] \, .
\end{equation}
Since the new objective function is scaling invariant, we can optimize over the whole cone \(C_\mathcal{M}\) without changing the optimal value.
\end{proof}

At \(\alpha = \infty\), we similarly obtain the following characterization of the measured max-divergence.
\begin{corollary}\label{Thm:Variational_Max}
    For \(\rho, \sigma \in \mathcal{S} \) and nonempty \(\mathcal{M} \subseteq \text{ALL}\), we have
    \begin{equation}
    D_{\max}^{\mathcal{M}} \left(\rho \middle\| \sigma \right) \leq \sup_{\stackrel{\omega > 0}{\omega \in C_\mathcal{M}}} \log \tr[\rho \omega] + 1 - \tr[\sigma \omega] = \sup_{\stackrel{\omega > 0}{\omega \in C_\mathcal{M}}} \log(\frac{\tr[\rho \omega]}{\tr[\sigma \omega]})\, .
\end{equation}

\end{corollary}

\begin{proof}

Using Property 3 of Lemma \ref{Lem:General_Properties}, we obtain 
\begin{align}
    D_{\max}^\mathcal{M} \left( \rho \middle\| \sigma \right) = \sup_{\alpha\in(1,\infty)} D_\alpha^\mathcal{M} \left( \rho \middle\| \sigma \right) &\leq \sup_{\alpha\in(1,\infty)} \sup_{\stackrel{\omega > 0}{\omega \in C_\mathcal{M}}} \log \left( \tr[\rho \omega^{\frac{\alpha - 1}{\alpha}}]^\alpha \tr[\sigma \omega]^{1-\alpha} \right) \\
    &= \sup_{\alpha\in(1,\infty)} \sup_{\stackrel{\omega > 0}{\omega \in C_\mathcal{M}}} \frac{\alpha}{\alpha -1} \log \tr[ \rho \omega^{\frac{\alpha-1}{\alpha}}] - \log \tr[\sigma \omega ] \, ,
\end{align}
where we used Theorem \ref{Thm:Variational_Q} for the inequality step. Using the decomposition \(\omega = \sum_{i=1}^d \lambda_i P_i\), we have for \(\alpha\in(1,\infty)\) that
    \begin{align}
        \frac{\alpha}{\alpha -1} \log \tr[ \rho \omega^{\frac{\alpha-1}{\alpha}}] &= \frac{\alpha}{\alpha -1} \log( \sum_{i=1}^d \lambda_i^{\frac{\alpha-1}{\alpha}} \tr[\rho P_i] ) \leq \log( \sum_{i=1}^d \lambda_i \tr[\rho P_i] ) = \log \tr[ \rho \omega] \, ,
    \end{align}
where the inequality step follows from the concavity of \(f(t):= t^\frac{\alpha-1}{\alpha}\) with the Jensen inequality. This allows us to conclude that 
\begin{align}
    D_{\max}^\mathcal{M} \left( \rho \middle\| \sigma \right) \leq \sup_{\stackrel{\omega > 0}{\omega \in C_\mathcal{M}}} \log( \frac{\tr[ \rho \omega]}{\tr[ \sigma \omega]} ) \, .
\end{align}

The equivalence of the two variational expressions can again be shown similar as it was done in Theorem \ref{Thm:Variational_Q}. We use the same parameterization of the optimization variable to obtain
\begin{equation}
    \sup_{\stackrel{\omega > 0}{\omega \in C_\mathcal{M}}} \log \tr[\rho \omega ] + 1 - \tr[\sigma \omega] = \sup_{\lambda > 0} \sup_{\stackrel{\omega > 0}{\omega \in C_\mathcal{M} \, \cap \, \mathcal{S}} } \log \tr[\rho (\lambda \omega) ] +1 - \tr[\sigma (\lambda \omega)]
\end{equation} 
A straightforward calculation shows that the expression is maximized for \(\lambda^\star = \frac{1}{\tr[\sigma \omega]} > 0\). This results in the remaining variational expression
\begin{equation}
    \sup_{\stackrel{\omega > 0}{\omega \in C_\mathcal{M} \, \cap \, \mathcal{S}} } \log \tr[\rho \omega ] - \log \tr[\sigma \omega] = \sup_{\stackrel{\omega > 0}{\omega \in C_\mathcal{M} \, \cap \, \mathcal{S}} } \log( \frac{\tr[ \rho \omega]}{\tr[ \sigma \omega]} ) \, .
\end{equation}
Since the new objective function is scaling invariant, we can optimize over the whole cone \(C_\mathcal{M}\) without changing the optimal value.
\end{proof}


\subsubsection{Discussion of Proof}\label{Sec:Discussion_Proof}

Let us conclude this section with some observations about sufficient conditions for when the given variational characterization is exact. Note that the proof of Theorem \ref{Thm:Variational_Q} essentially boils down to two inequality steps. The first is an application of the operator Jensen inequality. A sufficient condition for this step to become tight is when the cone \(C_\mathcal{M}\) is closed under power functions. The second step relates the cone \(C_\mathcal{M}\) to the measurement set \(\mathcal{M}\). A sufficient condition for equality is that the cone has the same mathematical structure as the measurement set. Note that both of these conditions are satisfied in the unrestricted case. Here, \(C_\text{ALL}\) is equal to the cone of positive semi-definite operators \(\mathcal{P}\). It was shown in \cite[Lemma 3]{Berta1} that in this case we do have equality in Eq.\ \eqref{Eq:Compact_Variational}.


\subsection{Locally-Measured R\'enyi Divergence}\label{Sec:Local_Variational_Formula}

Let us now take a closer look at the locality-constrained measurements sets introduced in Sec.\ \ref{Sec:Local_Measurement_Sets}. Here, we are able to derive an exact variational characterization for the two sets LO\((A:B)\) and LOCC\(_1(A\to B)\). These results generalize \cite[Lemma 4]{Berta2} for the measured relative entropy. We obtain these characterizations by embedding the states in a larger Hilbert space. We further provide an argument for why this is required and why the variational bound is not tight for \(\alpha \in (0,\infty)\) in general. In App.\ \ref{App:Variational_Bound_Tight}, we then give explicit counterexamples where the gap between the measured divergence and the variational characterization is strict.


\subsubsection{LO Measurements}

We start with the set of local measurements. Before we consider the general case, however, let us initially focus on the special case of projective local measurements. Here, the measurement cone is given by 
\begin{equation}
    C_{\text{P-LO}(A:B)} = \left\{ \sum_{x=1}^{d_A} \sum_{y=1}^{d_B} \lambda_{x,y} P_{A}^x \otimes P_{B}^y \, \middle| \, \lambda_{x,y} \geq 0 \wedge P_A^xP_A^{x'} = \delta_{x,x'} P_A^x \wedge P_B^yP_B^{y'} = \delta_{y,y'} P_A^y \right\} \, .
\end{equation}
Notice that this is the set of all operators that are classical in some basis on \(A \otimes B\). In this case, the variational bound is tight. 
\begin{lemma}\label{Lem:Projective_LO}
    Let \(\rho, \sigma \in \mathcal{S}_{AB}\) and \(\alpha>0\). With definitions as above, we have
    \begin{equation}
        D_\alpha^{\text{P-LO}(A:B)} \left( \rho \middle\| \sigma \right) = V_\alpha^{\text{P-LO}(A:B)} (\rho,\sigma) \, .
    \end{equation}
\end{lemma}
\begin{proof}
    We only need to prove that 
    \begin{equation}\label{Eq:Projective_Upper_Bound}
        \sup_{\omega > 0}\bigg\{ \nu_\alpha \big( (\rho,\sigma); \omega \big) \, \bigg| \, \omega \in C_{\text{P-LO}(A:B)} \bigg\} \leq D_\alpha^{\text{P-LO}(A:B)} \left( \rho \middle\| \sigma \right) \, .
    \end{equation}
    Let us illuminate the argument with the case \(\alpha \in (0,1/2)\). First, note that by definition any \( \omega \in C_{\text{P-LO}(A:B)} \) decomposes into
    \begin{equation}
        \omega = \sum_{x = 1}^{d_A} \sum_{y =1}^{d_B} \lambda_{x,y} P_{A}^x \otimes P_{B}^y \, .
    \end{equation}
    Since these are diagonal operators, we can rewrite the objective function as 
    \begin{equation}
        \exp \left( (\alpha-1) \nu_\alpha \big( (\rho, \sigma), \omega \big) \right) = \sum_{x=1}^{d_A} \sum_{y=1}^{d_B} \alpha \lambda_{x,y} \mu_\rho^P(x,y) + (1-\alpha) \lambda_{x,y}^{\frac{\alpha}{\alpha-1}} \mu_\sigma^P(x,y) \, .
    \end{equation}
    Here, we identified the measurement statistics of a local rank-1 PVM with elements \(P_{AB}^{(x,y)} = P_A^x \otimes P_B^y\). An analytical minimization of the expression w.r.t.\ the coefficients \(\lambda_{x,y}\) yields the optimizer
    \begin{equation}
        \lambda_{x,y}^\star = \left( \frac{ \mu_\rho^P(x,y)}{ \mu_\sigma^P(x,y)} \right)^{\alpha-1} \geq 0 \, .
    \end{equation}
    With this, we can see that the variational expression is upper bounded by
    \begin{equation}
        \sup_{\stackrel{\omega > 0}{\omega \in C_{\text{P-LO}(A:B)} } } \nu_\alpha \big( (\rho,\sigma); \omega \big) \leq \sup_{P = P_A \otimes P_B} D_\alpha \left( \mu_\rho^P \middle\| \mu_\sigma^P \right) = D_\alpha^{\text{P-LO}(A:B)} \left( \rho \middle\| \sigma \right) \, .
    \end{equation}

    An analogous argument can then be used to show the claim for the remaining orders \(\alpha\).
\end{proof} 

With this in place, we turn to the general case. On first look, the set \( C_\text{LO} \) consists of separable positive semi-definite operators of the form
\begin{equation}\label{Eq:Local_Cone}
     \omega = \sum_{x \in \mathcal{X}} \sum_{y \in \mathcal{Y}} \lambda_{x,y} M_A^x \otimes M_B^y \, ,
\end{equation}
where \(\lambda_{x,y} \geq 0\) and \( \{M_A^x\}_{x \in \mathcal{X}}\) and \(\{M_B^y\}_{y \in \mathcal{Y}}\) are POVMs over finite alphabets \(\mathcal{X}\) and \(\mathcal{Y}\), respectively. That is, the cone \( C_\text{LO} \) is clearly a subset of the cone of separable operators. It turns out, however, that \( C_\text{LO} \) is equal to the complete cone of separable operators. 

\begin{lemma}\label{Lem:Local_Cone}
     \(C_\text{LO}\) as defined in Eq.\ \eqref{Def:Measurement_Cone} is equal to the cone of separable operators.
\end{lemma}

\begin{proof}

To see this consider an arbitrary separable operator which we can write without loss of generality as \( \omega = \sum_{z \in \mathcal{Z}} X^z \otimes Y^z \) for some finite set of positive semi-definite operators \(\{X^z\}_z\) and \(\{Y^z\}_z\), respectively.  

Now, observe that we can always associate a two-outcome measurement to any positive semi-definite operator \(X\) using the elements \( M_X := X/\tr[X] \) and \(1_\dutchcal{H}-M_X\). If we do this for each individual \(X^z\) and re-scale all POVM elements by \(|\mathcal{Z}|\), we obtain a valid POVM on \(A\) with \(2 |\mathcal{Z}| \) outcomes since 
\begin{equation}
    \sum_{z \in \mathcal{Z} } \frac{X^z}{|\mathcal{Z}| \tr[X^z]} + \frac{1}{|\mathcal{Z}|}\left( 1_A - \frac{X^z}{\tr[X^z]} \right) = \sum_{z \in \mathcal{Z} } \frac{1_A}{|\mathcal{Z}|} = 1_A \, .
\end{equation}
Similarly, we define a valid POVM on \(B\) with \(2 |\mathcal{Z}| \) outcomes corresponding to the set of operators \(Y^z\). 

It is then easy to see that the separable operator \(\omega\) lies in the conic hull of the product POVM. For this, we take only the POVM elements \(X^z/( |\mathcal{Z}| \tr[X^z] )\otimes Y^z/( |\mathcal{Z}| \tr[Y^z]) \) and scale them by the positive coefficients \( \lambda_z =  |\mathcal{Z}|^2 \tr[X^z] \tr[X^z]\). Then, we have
\begin{equation}
    \sum_{z \in \mathcal{Z}} |\mathcal{Z}|^2 \tr[X^z] \tr[Y^z] \left( \frac{X^z}{ |\mathcal{Z}| \tr[X^z] } \otimes \frac{ Y^z}{ |\mathcal{Z}| \tr[Y^z] } \right) = \sum_{z \in \mathcal{Z}} X^z \otimes Y^z = \omega \, .
\end{equation}
\end{proof}

As a result, the cone \(C_\text{LO}\) is a strict superset to LO. This suggests that the inequality in the variational characterization of Eq.\ \eqref{Eq:Compact_Variational} is strict in the general case (see App.\ \ref{App:Variational_Bound_Tight} for examples). We can, however, adapt the proof idea of \cite[Lemma 4]{Berta2} to obtain an exact characterization. The main idea is to use Naimark's dilation theorem to map the problem of optimizing over general POVMs to one of optimizing over PVMs on a larger Hilbert space. For this, let us introduce the spaces \(A'\) isomorphic to \(A \otimes A\) and \(B'\) isomorphic to \(B \otimes B\). We then have the following result for the set LO that generalizes a corresponding result from \cite[Equation 47]{Berta2} for the measured relative entropy. 
\begin{proposition}\label{Prop:Exact_LO_Variational}
    Let \( \rho_{AB}, \sigma_{AB} \in \mathcal{S}_{AB}\) and \(\alpha>0\). Further, let \(A' \otimes B'\) as defined above and \(\rho_{A'B'}, \sigma_{A'B'}\) be local embeddings of \(\rho_{AB}\) and \(\sigma_{AB}\) in this space, respectively. We then have with notation as defined above that
    \begin{equation}
        D_\alpha^{\text{LO}(A:B)} \left( \rho_{AB} \middle\| \sigma_{AB} \right) = V_\alpha^{\text{P-LO}(A':B')} (\rho_{A'B'},\sigma_{A'B'})
    \end{equation}

    Furthermore, the optimal local measurements are (rank-1) POVMs with at most \(d_A^2\) and \(d_B^2\) outcomes on \(A\) and \(B\), respectively.
    
\end{proposition}
\begin{proof}

We adapt the proof idea of \cite[Lemma 4]{Berta2}. 

Note that the set LO\((A:B)\) is convex on the \(A\) system and thus we can conclude that the optimal measurement on \(A\) must be extremal (cf.\ Sec.\ \ref{Sec:Optimal_Measurements}). These extremal POVMs have at most \(d_A^2\) rank-1 elements by \cite[Theorem 2.21]{Holevo}. Moreover, Naimark's dilation theorem states that there exists a rank-1 projective measurement on \(A'\) that produces the same measurement statistics. Furthermore, since the set LO\((A:B)\) is convex on the \(B\) system as well, the same conclusions hold for the measurement on the \(B\) system. 

The local isometric invariance of the locally-measured R\'enyi divergence gives
\begin{equation}
    D_{\alpha}^{\text{LO}(A:B)}\left(\rho_{AB} \middle\| \sigma_{AB} \right) = D_{\alpha}^{\text{LO}(A':B')} \left(\rho_{A'B'} \middle\| \sigma_{A'B'} \right) 
\end{equation}
and by the above observation we can restrict w.l.o.g.\ to projective measurements in the latter optimization. The claim then follows by an application of Lemma \ref{Lem:Projective_LO}. 
\end{proof}


\subsubsection{LOCC Measurements}

We obtain a similar result for the class LOCC\(_1(A\to B)\). Note that Lemma \ref{Lem:Local_Cone} implies that the cone \(C_{\text{LOCC}_1}\) is also equal to the separable cone by a superset argument. (In fact, it implies this for the cone \(C_\text{LOCC}\) as well.) Recalling the sufficient conditions in Sec.\ \ref{Sec:Discussion_Proof}, the cone thus does not share the same structure as the measurement set. As for the set LO, this gives an intuition why the variational characterization in Eq.\ \eqref{Eq:Compact_Variational} is strict in the general case as we show in App.\ \ref{App:Variational_Bound_Tight}.

Nevertheless, if we restrict again the measurements to projective ones we obtain an exact characterization as in Lemma \ref{Lem:Projective_LO} in terms of an optimization over the cone
\begin{equation}
    C_{\text{P-LOCC}_1(A\to B)} = \left\{ \sum_{x=1}^{d_A} P_{A}^x \otimes \omega_B^x \, \middle| \, \omega_B^x \geq 0 \wedge P_A^x P_A^{x'} = \delta_{x,x'} P_A^x \right\} \, .
\end{equation}
Note that this is the set of operators that are classical-quantum in some basis on \(A\). This claim can be shown similar to Lemma \ref{Lem:Projective_LO} by introducing the eigendecomposition \( \omega_B^x = \sum_{y=1}^{d_B} \lambda_{y|x} P^{y|x} \) and then carrying out the optimization over \(\lambda_{y|x}\) analytically. 

Applying the Naimark's dilation argument from Proposition \ref{Prop:Exact_LO_Variational} then similarly yields an exact characterization for the class LOCC\(_1\) as well. The proof boils down to the observation that the optimal measurement is comprised of a (rank-1) POVM on \(A\) with at most \(d_A^2\) outcomes followed by a conditional projective measurement on \(B\). Therefore, we have with local embeddings of \(\rho_{A'B}\) and \(\sigma_{A'B}\) of \(\rho_{AB}\) and \(\sigma_{AB}\) in \(A' \otimes B\), respectively, that
\begin{equation}
    D_\alpha^{\text{LOCC}_1(A\to B)} \left( \rho_{AB} \middle\| \sigma_{AB} \right) = V_\alpha^{\text{P-LOCC}_1(A'\to B)} (\rho_{A'B},\sigma_{A'B}) \, .
\end{equation} 

For the general class LOCC, we do not know how to obtain an exact characterization since no good mathematical characterization of the set is available.


\subsubsection{SEP and PPT Measurements}

For the sets SEP\((A:B)\) and PPT\((A:B)\), note that \(C_\text{SEP}\) and \(C_\text{PPT}\) are simple and just consist of separable positive semi-definite operators and positive semi-definite operators with positive partial transpose, respectively. That is, the cones have the same mathematical structure as the respective measurement set. However, these sets are not closed under power functions and thus do not satisfy the second sufficient condition given in Sec.\ \ref{Sec:Discussion_Proof} (see App.\ \ref{App:Variational_Bound_Tight} for examples).


\subsection{Measured Max-Divergence}\label{Sec:Measured_Max_Divergence}

We conclude our investigation into variational characterizations of the measured R\'enyi divergence by taking a closer look at the order \(\alpha=\infty\), i.e.\ the max-divergence. Here, we show that the obtained variational characterization is in fact exact for the classes SEP and PPT. Moreover, the obtained variational upper bound can be re-written as a conic program in standard form under some mild assumptions on the measurement set. This enables us to derive a dual variational characterization that is crucial for our application to restricted hypothesis testing (cf.\ Sec.\ \ref{Sec:DataHiding} for details).


\subsubsection{SEP and PPT Measurements}

Recalling the sufficient conditions for equality, we outlined in Sec.\ \ref{Sec:Discussion_Proof}, we can see that both of these hold at \(\alpha=\infty\) for the sets SEP\((A:B)\) and PPT\((A:B)\). For this, note that the relevant power for \(\alpha\to\infty\) is the identity for which the closure property trivially holds. We are then able to show that the variational bound is tight.

\begin{proposition}\label{Lem:Exact_Variational_Max}
    For \(\rho, \sigma \in \mathcal{S}_{AB} \) and \(\mathcal{M} = \text{SEP} \) or PPT, we have
    \begin{equation}
        D_{\max}^{\mathcal{M}} \left(\rho \middle\| \sigma \right) = \sup_{\stackrel{\omega > 0}{\omega \in C_\mathcal{M}}} \log \tr[\rho \omega] + 1 - \tr[\sigma \omega] = \sup_{\stackrel{\omega>0}{\omega \in C_\mathcal{M}}} \log(\frac{\tr[\rho \omega]}{\tr[\sigma \omega]}) \, .
    \end{equation}
\end{proposition}

\begin{proof}
    Let \(\mathcal{M} = \text{SEP}\) or PPT. The proof then essentially boils down to the fact that \(C_\mathcal{M}\) has the same mathematical structure as the measurement set. 
    
    Namely, we can define for any \(\omega \in C_\mathcal{M}\) a 2-outcome POVM as \(\{M_\omega, 1_\dutchcal{H}-M_\omega\} \in \mathcal{M} \) via the element \( M_\omega = \omega/\tr[\omega] \). Thus, for any \(\omega \in C_\mathcal{M}\) we have
    \begin{align}
        \log(\frac{\tr[\rho \omega]}{\tr[\sigma \omega]} ) &= \log(\frac{\tr[\rho M_\omega]}{\tr[\sigma M_\omega]} ) \leq D_{\max}\left(\mu_\rho^{M_\omega} \middle\| \mu_\sigma^{M_\omega} \right) \leq D_{\max}^\mathcal{M} \left( \rho \middle\| \sigma \right) \, .
    \end{align}
\end{proof}


\subsubsection{Dual Variational Expression}

If the cone \(C_\mathcal{M}\) is a proper cone\footnote{A cone is called proper if it is closed, convex, pointed and has nonempty interior.} and the set \(\mathcal{M}\) contains more than just the trivial measurement \(M=1_\dutchcal{H}\), we can rewrite the variational expression for the measured max-divergence into a conic program in standard form (cf.\ Watrous' lecture notes \cite[Lecture 1]{Watrous1}). Proceeding from there, we obtain the dual program. 
\begin{proposition}\label{Prop:Dual_Cone_Program}
For \(\rho, \sigma \in \mathcal{S}\) and nontrivial \(\mathcal{M} \subseteq \text{ALL}\) such that \(C_\mathcal{M}\) is a proper cone, we have
\begin{align}
     D_{\max}^{\mathcal{M}} \left( \rho \middle\| \sigma \right) \leq \sup_{\stackrel{\omega>0}{\omega \in C_\mathcal{M}}} \log(\frac{\tr[\rho \omega]}{\tr[\sigma \omega]}) = \inf_{\stackrel{\lambda > 0 }{\lambda \sigma - \rho \in C_{\mathcal{M}}^{\dagger}}} \log \lambda \, ,
\end{align}
with the dual cone \(C_\mathcal{M}^\dagger\) given by
\begin{equation}
    C_{\mathcal{M}}^{\dagger} := \big\{ \gamma \in \mathcal{H} \big| \tr[\gamma \omega] \geq 0 \; \forall \omega \in C_\mathcal{M} \big\} \, .
\end{equation}
\end{proposition}

\begin{proof}
    By Corollary \ref{Thm:Variational_Max}, we have
    \begin{equation}
        \exp( D_\text{max}^\mathcal{M}(\rho \| \sigma) ) \leq \sup_{\stackrel{\omega>0}{\omega \in C_\mathcal{M}}} \bigg\{ \tr[\rho \omega] \, \bigg| \, \tr[\sigma \omega] = 1 \bigg\} \, ,
    \end{equation}
    where we introduced w.l.o.g.\ the constraint \( \tr[\sigma \omega] = 1 \) by using the scaling invariance of the objective function together with \( \tr[\sigma \omega] > 0 \) for \(\omega > 0\).

    Note that per our assumption \(C_\mathcal{M}\) is a proper cone. Moreover, we can define a linear map \(\phi : \mathcal{H} \to \mathbbm{R} \) via \(\phi(\gamma) = \tr[\sigma \gamma]\) for \(\gamma\in\mathcal{H}\). Its adjoint is then given by \( \phi^\dagger(\lambda) = \lambda \sigma \) for \(\lambda\in\mathbbm{R}\). With this, the primal problem is given in standard form (cf.\ \cite[Problem 1.1]{Watrous1}) and the dual problem is 
   \begin{equation}
        \inf_{\lambda \in \mathbbm{R} } \bigg\{ \lambda \, \bigg| \, \lambda \sigma - \rho \in C_\mathcal{M}^\dagger \bigg\} \, .
    \end{equation}

    Moreover, since by picking the primal feasible \(\omega^\star=1_\dutchcal{H}\) the primal is lower-bounded by 1, we can restrict the optimization in the dual w.l.o.g.\ to \(\lambda > 0\) due to weak duality.

    Lastly, we can show that strong duality holds using the version of Slater's theorem stated in \cite[Theorem 1.3]{Watrous1}. For this, assume first that the dual feasible set is empty,\footnote{By definition the dual optimal value then is \(+\infty\).} i.e.\ for each \(\lambda>0\) there exists at least one \(\omega \in C_\mathcal{M}\) such that \( \tr[(\lambda \rho - \sigma) \omega] < 0\) by the definition of the dual cone. This implies that for all \(\lambda>0\), we can find a \(\omega\in C_\mathcal{M}\) such that
    \begin{equation}
        \lambda < \frac{\tr[\rho \omega]}{\tr[\sigma \omega]} \, .
    \end{equation}
    Therefore, in this case the primal problem is unbounded as well. If the dual problem is feasible, Slaters theorem applies with \(1_\dutchcal{H} \in \text{relint}(C_\mathcal{M})\).
\end{proof}

Observe that in the case of unrestricted measurements, we have \(C_\text{ALL}=\mathcal{P}\) which is known to be self dual. The given dual program then yields the definition of the quantum max-divergence which shows again that the quantum max-divergence is indeed achievable by a measurement.

The cones \(C_\text{SEP}\) and \(C_\text{PPT}\) are known to be proper cones (cf.\ Watrous' lecture notes \cite[Lecture 14 \& 18]{Watrous2}). The dual cone of the separable operators is the cone of block-positive operators (cf.\ e.g.\ \cite{Skowronek}). Moreover, the dual cone of PPT operators admits an explicit characterization as we show in the following corollary. This characterization of the dual PPT cone is known in the literature (cf.\ e.g.\ \cite{Horodecki}), but we provide a proof here for completeness.
\begin{corollary}\label{Lem:Dual_PPT_Max}
  For \(\rho, \sigma \in \mathcal{S}_{AB}\), we have
\begin{align}
     D_{\max}^\text{PPT} \left(\rho \middle\| \sigma \right) = \log \sup_{\stackrel{\omega>0}{\omega^\Gamma \geq 0}}  \bigg\{  \tr[\rho \omega] \, \bigg| \,  \tr[\sigma \omega] = 1 \bigg\} = \log \inf_{\stackrel{\lambda >0}{X,Y \in \mathcal{P}}} \bigg\{  \lambda \, \bigg| \, \lambda \sigma - \rho =  X + Y^\Gamma \bigg\} \, .
\end{align}  
\end{corollary}

\begin{proof}
By Propositon \ref{Lem:Exact_Variational_Max}, we have that 
\begin{equation}
    \exp( D_\text{max}^\text{PPT}(\rho \| \sigma) ) = \sup_{\omega>0} \bigg\{ \tr[\rho \omega] \, \bigg| \, \tr[\sigma \omega] = 1 \wedge \omega^\Gamma \geq 0 \bigg\}  \, ,
\end{equation}
where we introduced w.l.o.g.\ the additional constraint \(\tr[\sigma \omega] = 1\) using the scaling invariance of the problem. Based on this expression, we proceed to find its dual using the Lagrangian formalism (cf.\ e.g.\ \cite[Chapter 5]{Boyd}). The Lagrangian \( \mathcal{L} \) of the problem is given by
\begin{align}
    \mathcal{L}(\omega, X, Y, \lambda) &= \tr[\rho \omega] + \tr[X \omega] + \tr[Y \omega^\Gamma] + \lambda \left( 1 - \tr[\sigma \omega] \right) = \lambda + \tr[ \left(\rho - \lambda \sigma + X + Y^\Gamma \right) \omega ] 
\end{align}
with the dual variables \(X,Y \in \mathcal{P}\) and \(\lambda \in \mathbbm{R}\). In order to obtain the dual program, this has to be maximized over all Hermitian operators \(\omega\) without any constraints. The result is \(+ \infty\) unless \( \rho - \lambda \sigma + X + Y^\Gamma = 0 \). In such a case, the maximum is simply given by \(\lambda\). Thus, we can conclude that the dual program is given by
\begin{align}
    \inf_{\stackrel{\lambda \in \mathbbm{R}}{X,Y \in \mathcal{P}}} \bigg\{ \lambda \, \bigg| \, \lambda \sigma - \rho = X + Y^\Gamma \bigg\} \, .
\end{align}

Analogous to the proof of Proposition \ref{Prop:Dual_Cone_Program}, we can restrict the optimization to \(\lambda>0\) due to weak duality and show that strong duality holds using the strictly primal feasible point \(\omega=1_\dutchcal{H}\). 
\end{proof}

Lastly, we remark that Proposition \ref{Prop:Dual_Cone_Program} shows that the measured max-divergence is upper-bounded by the cone-restricted max divergences \(D_{\max}^{C_\mathcal{M}}\) of George and Chitambar \cite{George}. Combining this with Proposition \ref{Lem:Exact_Variational_Max} shows that in the cases of SEP and PPT measurements the measured max-divergence is in fact equal to the cone-restricted one.

As a final remark, note that the quantum fidelity also possesses a dual variational characterization (cf.\ e.g.\ \cite[Lecture 8]{Watrous2}). The standard derivation of this, however, does rely on the closure under inverses of the positive semi-definite cone, which does not hold for the separable or PPT cone.



\section{Restricted Hypothesis Testing}\label{Sec:DataHiding} 

As our application of the locally-measured R\'enyi divergences, we use them to study the hypothesis testing problem under restricted measurements. For this, we develop in the following section an operational interpretation of the regularized locally-measured R\'enyi divergence in the strong converse regime. In Sec.\ \ref{Sec:Isotropic_States} and App.\ \ref{App:Werner_States}, we then show that the regularized locally-measured R\'enyi divergences become single-letter and can be evaluated analytically on examples of data hiding states. As argued in the introduction, these states are natural to consider since they are characterized by a reduced distinguishability under LOCC measurements. Moreover, we can use these states to provide explicit counterexamples for when the variational bounds of Sec.\ \ref{Sec:Local_Variational_Formula} are not tight (see App.\ \ref{App:Variational_Bound_Tight}). In the following, we always consider a bipartite system \(A \otimes B\) with equal local dimension \(d_A=d_B=d\).


\subsection{Locally-Measured Hypothesis Testing}\label{Sec:Hypothesis_Testing}

We study in the following the simple hypothesis testing problem under restricted measurements. Here, the task is to decide whether the null hypothesis \( \rho^{\otimes n}\) or its alternative \(\sigma^{\otimes n}\) is true based on the outcome of a quantum measurement on \(\dutchcal{H}^{\otimes n}\). In the hypothesis testing problem, it is sufficient to treat only binary measurements, so-called tests \(\{T_n(0),T_n(1)\}\), where the outcomes \(0\) and \(1\) indicate the acceptance of the null and alternative hypothesis, respectively. Since the test is uniquely determined by the element \(T_n = T_n(0)\), we will use with some abuse of notation \(T_n\) to denote the POVM element and its associated test. In the restricted measurement scenario, we are not allowed to perform all possible quantum tests but only the ones from a subset \(\mathcal{M}_n\). 

To each test \(T_n \in \mathcal{M}_n\), we can associate two types of error probabilities, namely
\begin{align}
    \alpha_n(T_n) &:=  1- \tr[\rho^{\otimes n} T_n] & \text{and} && \beta_n(T_n) &:= \tr[\sigma^{\otimes n} T_n] \, .
\end{align}
The error probability of the first kind \(\alpha_n(T_n)\) gives the likelihood to wrongly reject the null hypothesis, whereas the error probability of the second kind \(\beta_n(T_n)\) characterizes the probability of erroneously accepting the null hypothesis. In general there exists a trade-off between these error probabilities and different scenarios are studied in the literature.


\subsubsection{Stein's Lemma}\label{Sec:Steins_Lemma}

In the Stein scenario, we enforce a constant constraint on the error probability of the first kind, i.e.\ we require \(\alpha_n(T_n) \leq \varepsilon\) with \(\varepsilon\in(0,1)\). The task is then to minimize the error of the second kind among all tests \(T_n\in\mathcal{M}_n\) that satisfy the constraint, i.e.\ we are interested in
\begin{equation}
    \beta^{\mathcal{M}_n}_n(\varepsilon) := \inf_{T_n \in \mathcal{M}_n} \bigg\{ \beta_n(T_n) \; \bigg| \; \alpha_n(T_n) \leq \varepsilon \bigg\} \, .
\end{equation}
A Stein's lemma then makes a statement about the asymptotic behavior of this error probability as \(n\to\infty\), i.e.\ it characterizes the rate exponent
\begin{equation}
    \zeta_\text{Stein}^\mathcal{M}(\rho,\sigma;\varepsilon) := \lim_{n\to\infty} -\frac{1}{n}\log \beta^{\mathcal{M}_n}_n(\varepsilon)
\end{equation}
if the limit exists. The quantum Stein's lemma \cite{Hiai, Ogawa} states that under no restrictions on the measurements we have for all \(\varepsilon\in(0,1)\) that 
\begin{equation}
     \zeta_\text{Stein}^\text{ALL}(\rho,\sigma;\varepsilon) = D^{\mathbf{ALL}} \left(\boldsymbol{\rho}\middle\|\boldsymbol{\sigma}\right) = D(\rho,\sigma) \, .
\end{equation}
Moreover, the weak version of Stein's lemma for restricted measurements has been proven by Brandao et.\ al.\ as a special case of \cite[Theorem 16]{Brandao}. Namely, their result implies
\begin{equation}
    \lim_{\varepsilon\to0} \zeta_\text{Stein}^\mathcal{M}(\rho,\sigma;\varepsilon) = D^{\boldsymbol{\mathcal{M}}} \left(\boldsymbol{\rho}\middle\|\boldsymbol{\sigma}\right) \, .
\end{equation}
Here, we want to investigate the strong version of the Stein's lemma for restricted measurements. Observe that in the proof of \cite[Theorem 16]{Brandao}, it was shown
\begin{equation}
    \text{$\liminf_{n\to\infty} -\frac{1}{n}\log \beta^{\mathcal{M}_n}_n(\varepsilon) \geq D^{\boldsymbol{\mathcal{M}}} \left(\boldsymbol{\rho}\middle\|\boldsymbol{\sigma}\right)$ for all \(\varepsilon\in(0,1)\).}
\end{equation}
This can be regarded as the achievability part of Stein's lemma and for the optimality part, we make use of the measured R\'enyi divergences for \(\alpha>1\) to prove the following lemma.

\begin{lemma}\label{Lem:Stein_Converse}
    With definitions as above and \(\varepsilon\in(0,1)\), we have
    \begin{equation}
        \limsup_{n\to\infty} -\frac{1}{n}\log \beta^{\mathcal{M}_n}_n(\varepsilon) \leq \inf_{\alpha>1} D_\alpha^{\boldsymbol{\mathcal{M}}} \left(\boldsymbol{\rho}\middle\|\boldsymbol{\sigma}\right) \, .
    \end{equation}
\end{lemma}

\begin{proof}
    
Let \(\alpha>1\) and \(\varepsilon\in(0,1)\) and consider a test \(T_n \in \mathcal{M}_n\). We then have by a standard argument (cf.\ e.g.\ \cite{Nagaoka}) that
\begin{align}
    &D_\alpha^{\mathcal{M}_n} \left(\rho^{\otimes n}\middle\|\sigma^{\otimes n}\right)\\
    &\geq \frac{1}{\alpha-1} \log\left( \tr[\rho^{\otimes n}T_n]^\alpha \tr[\sigma^{\otimes n}T_n]^{1-\alpha} + \left(1-\tr[\rho^{\otimes n}T_n]\right)^\alpha \left(1-\tr[\sigma^{\otimes n} T_n]\right)^{1-\alpha} \right) \\
    &\geq \frac{1}{\alpha-1} \log\left( \tr[\rho^{\otimes n}T_n]^\alpha \tr[\sigma^{\otimes n}T_n]^{1-\alpha} \right) \\
    &= \frac{\alpha}{\alpha-1} \log(1-\alpha_n(T_n)) - \log \beta_n(T_n) \, ,
\end{align}
which after rewriting shows the following relationship between the two error probabilities
\begin{align}\label{Eq:Error_Relationship}
    -\frac{1}{n}\log \beta_n(T_n) \leq \frac{1}{n} D_\alpha^{\mathcal{M}_n} \left(\rho^{\otimes n}\middle\|\sigma^{\otimes n}\right) + \frac{1}{n} \frac{\alpha}{\alpha-1} \log(\frac{1}{1-\alpha_n(T_n)}) \, . 
\end{align}

Recall that in the Stein's scenario, we consider tests \(T_n\) that satisfy the constraint \(\alpha_n(T_n) \leq \varepsilon\). For such tests, the above bound turns into
\begin{align}
    -\frac{1}{n}\log \beta_n(T_n) \leq \frac{1}{n} D_\alpha^{\mathcal{M}_n} \left(\rho^{\otimes n}\middle\|\sigma^{\otimes n}\right) + \frac{1}{n} \frac{\alpha}{\alpha-1} \log(\frac{1}{1-\varepsilon}) \, . 
\end{align}
Next, taking the supremum over such tests on the left-hand side, we get
\begin{align}
    -\frac{1}{n}\log \beta^{\mathcal{M}_n}_n(\varepsilon) \leq \frac{1}{n} D_\alpha^{\mathcal{M}_n} \left(\rho^{\otimes n}\middle\|\sigma^{\otimes n}\right) + \frac{1}{n} \frac{\alpha}{\alpha-1} \log(\frac{1}{1-\varepsilon})
\end{align}
and then taking the limit superior over \(n\) on both sides, we get by Lemma \ref{Lem:Regularization} that
\begin{align}
    \limsup_{n\to\infty} -\frac{1}{n}\log \beta^{\mathcal{M}_n}_n(\varepsilon) \leq D_\alpha^{\boldsymbol{\mathcal{M}}} \left(\boldsymbol{\rho}\middle\|\boldsymbol{\sigma}\right) \, .
\end{align}
Finally, taking the infimum over \(\alpha > 1\) on the right-hand side completes the proof. 
\end{proof}

In Sec.\ \ref{Sec:Isotropic_States} and App.\ \ref{App:Werner_States}, we show that for examples of data hiding states the locally-measured R\'enyi divergences are additive. The regularized terms then become single letter, which in turn enables us to show that the optimality part becomes tight, i.e.\ for these examples we have 
\begin{equation}
    \zeta_\text{Stein}^\mathcal{M}(\rho,\sigma;\varepsilon) = D^\mathcal{M} \left(\rho\middle\|\sigma\right)
\end{equation}
and we can explicitly compute the Stein's exponent for the restricted hypothesis testing problem (see Sec.\ \ref{Sec:Hypothesis_Testing_Isotropic} and App.\ \ref{Sec:Hypothesis_Testing_Werner} for these results).


\subsubsection{Strong Converse Exponent}\label{Sec:Strong Converse}

In the strong converse scenario, we aim to minimize the decay rate of the success probability of the first kind given an exponential constraint on the error probability of the second kind, i.e.\ we are interested in the rate exponent
\begin{equation}
    \zeta_\text{SC}^\mathcal{M}(\rho,\sigma; r) := \inf_{\stackrel{ \{T_n\}_n }{ T_n \in \mathcal{M}_n }} \left\{ \limsup_{n\to\infty} - \frac{1}{n}\log(1-\alpha_n(T_n))  \; \middle| \; \liminf_{n\to\infty} -\frac{1}{n} \log \beta_n(T_n) \geq r \right\} \, .
\end{equation}

The following lemma gives a universal lower bound, i.e.\ an optimality result, on the strong converse exponent.

\begin{lemma}
    Let \(r\geq 0\). With definitions as above, we have
    \begin{align}
        \liminf_{n\to\infty} -\frac{1}{n} \log(1-\alpha_n(T_n)) \geq \sup_{\alpha>1} \frac{\alpha-1}{\alpha}\left[ r-D_\alpha^{\boldsymbol{\mathcal{M}}} \left(\boldsymbol{\rho}\middle\|\boldsymbol{\sigma}\right)\right]
    \end{align} 
    for all sequences of tests \(\{T_n\}_n\) with \(T_n \in \mathcal{M}_n\) that satisfy the rate constraint
    \begin{equation}
        \limsup_{n \to \infty} \frac{1}{n} \log \beta_n(T_n) \leq -r \, .
    \end{equation}
\end{lemma}

\begin{proof}
Let \(r \geq 0\), \(\alpha>1\) and consider a test \(T_n \in \mathcal{M}_n\). 
We can rewrite Eq.\ \eqref{Eq:Error_Relationship} into
\begin{align}
    \frac{1}{n} \log(1-\alpha_n(T_n)) \leq \frac{\alpha-1}{\alpha}\left[ \frac{1}{n} D_\alpha^{\mathcal{M}_n} \left(\rho^{\otimes n}\middle\|\sigma^{\otimes n}\right) + \frac{1}{n}\log \beta_n(T_n)  \right] \, .
\end{align}
Per our assumption, we consider sequences of tests \(T_n\) that satisfy the constraint 
\begin{equation}
    \limsup_{n \to \infty} \frac{1}{n}\log \beta_n(T_n) \leq -r \, .
\end{equation}
That is, for each \(\delta >0\) there exists a \(n_\delta\) such that for all \(n \geq n_\delta\), we get the bound
\begin{align}
    \frac{1}{n} \log \big(1-\alpha_n(T_n) \big) \leq \frac{\alpha-1}{\alpha}\left[ \frac{1}{n} D_\alpha^{\mathcal{M}_n} \left(\rho^{\otimes n}\middle\|\sigma^{\otimes n}\right) -r + \delta \right]  \, .
\end{align}
To conclude, we take the limit superior over \(n\) to obtain
\begin{align}
    \limsup_{n\to\infty} \frac{1}{n} \log(1-\alpha_n(T_n)) \leq \frac{\alpha-1}{\alpha}\left[ D_\alpha^{\boldsymbol{\mathcal{M}}} \left(\boldsymbol{\rho}\middle\|\boldsymbol{\sigma}\right) - r + \delta \right]
\end{align} 
and then taking the limit \(\delta \to 0\) and the infimum over \(\alpha>1\) on the right-hand side results in
\begin{align}
    \limsup_{n\to\infty} \frac{1}{n} \log(1-\alpha_n(T_n)) \leq - \sup_{\alpha>1} \frac{\alpha-1}{\alpha}\left[ r-D_\alpha^{\boldsymbol{\mathcal{M}}} \left(\boldsymbol{\rho}\middle\|\boldsymbol{\sigma}\right)\right] \, .
\end{align} 

\end{proof}

For the achievability direction, we prove the following lemma which gives upper bounds on the strong converse exponent. 
\begin{lemma}\label{Lem:Strong_Converse}
    Let \(r\geq 0\) and \(k\in\mathbbm{N}\). With definitions as above, given a measurement \(M_k \in \mathcal{M}_k\), there exists a sequence of test \( \{ T_n \}_n \) with \(T_n \in \mathcal{M}_n\) that achieves
    \begin{equation}
        \limsup_{n \to \infty} \frac{1}{n} \log \beta_n(T_n)  \leq -r
    \end{equation}
    and
    \begin{align}
        \limsup_{n\to\infty} -\frac{1}{n}\log(1-\alpha_n(T_n)) \leq \sup_{\alpha>1}  \frac{\alpha-1}{\alpha} \bigg[ r - \frac{1}{k}  D_\alpha \left( \mu_{\rho^{\otimes k}}^{M_k} \middle\| \mu_{\sigma^{\otimes k}}^{M_k} \right) \bigg] \, .
\end{align} 
\end{lemma}

\begin{proof}
    Consider a measurement \(M_k\in\mathcal{M}_k\) over an alphabet \(\mathcal{Z}\). The classical strong converse result of Han-Kobayashi \cite{Han} implies the existence of a sequence of acceptance regions \( \{ \mathcal{A}_{k,m} \}_m\) with \(\mathcal{A}_{k,m} \subseteq \mathcal{Z}^{\otimes m}\) for the hypothesis testing problem between the i.i.d.\ distributions \( \left( \mu_{\rho^{\otimes k}}^{M_k} \right)^{\otimes m}\) and \( \left( \mu_{\sigma^{\otimes k}}^{M_k} \right)^{\otimes m} \) such that
    \begin{equation}
        \limsup_{m \to \infty} \frac{1}{m} \log \left[ \left( \mu_{\sigma^{\otimes k}}^{M_k} \right)^{\otimes m}(\mathcal{A}_{k,m}) \right] \leq - kr
    \end{equation}
    and 
    \begin{align}
        \liminf_{m\to\infty} \frac{1}{m} \log \left[ \left( \mu_{\rho^{\otimes k}}^{M_k} \right)^{\otimes m}(\mathcal{A}_{k,m}) \right] &= \inf_{\alpha>1} \frac{\alpha-1}{\alpha} \bigg[ D_\alpha \left( \mu_{\rho^{\otimes k}}^{M_k} \middle\| \mu_{\sigma^{\otimes k}}^{M_k} \right) - kr \bigg] \, .
    \end{align}

    Hence, given an arbitrary \(\delta>0\) there exists \(m_\delta\) such that for \(m \geq m_\delta\), we have for the error probability of the second kind
     \begin{equation}
        \frac{1}{m} \log \left[ \left( \mu_{\sigma^{\otimes k}}^{M_k} \right)^{\otimes m}(\mathcal{A}_{k,m}) \right] \leq - kr + \delta
    \end{equation}
    and for the success probability of the first kind
    \begin{equation}
        \frac{1}{m} \log \left[ \left( \mu_{\rho^{\otimes k}}^{M_k} \right)^{\otimes m}(\mathcal{A}_{k,m}) \right] \geq \inf_{\alpha>1} \frac{\alpha-1}{\alpha} \bigg[ D_\alpha \left( \mu_{\rho^{\otimes k}}^{M_k} \middle\| \mu_{\sigma^{\otimes k}}^{M_k} \right) - kr \bigg] - \delta \, .
    \end{equation}

    Let us define the quantum test \(T_n := (M_k^\dagger)^{\otimes m}(\mathcal{A}_{k,m}) \otimes 1_l \in \mathcal{M}_n\) for \( n = km +l\) with \(l = \{0,1, ..., k-1\}\). With this sequence of tests \(\{T_n\}\), we get for \(n \geq k m_\delta\) the bound 
    \begin{equation}
       \frac{1}{n} \log \tr[ T_n \sigma^{\otimes n} ] = \frac{1}{n} \log \left[ \left( \mu_{\sigma^{\otimes k}}^{M_k} \right)^{\otimes m}(\mathcal{A}_{k,m}) \right] \leq - \frac{km}{km+l }r + \frac{m}{km+l} \delta \, ,
    \end{equation}
    which in turn implies
    \begin{equation}
        \limsup_{n \to \infty} \frac{1}{n} \log \beta_n(T_n)  \leq -r + \frac{\delta}{k} \, .
    \end{equation}

    Moreover, we also have for \(n\geq k m_\delta\) the bound
    \begin{align}
         - \frac{1}{n} \log(1-\alpha_n(T_n))
         &= - \frac{1}{n} \log \tr[ T_n \rho^{\otimes n}] \\
         &= - \frac{1}{km+l} \log \left[ \left( \mu_{\rho^{\otimes k}}^{M_k} \right)^{\otimes m}(\mathcal{A}_{k,m}) \right] \\
        &\leq \frac{m}{km+l} \left( \sup_{\alpha>1}  \frac{\alpha-1}{\alpha} \bigg[ kr - D_\alpha \left( \mu_{\rho^{\otimes k}}^{M_k} \middle\| \mu_{\sigma^{\otimes k}}^{M_k} \right) \bigg] \right) + \frac{m}{km+l} \delta \, ,
    \end{align}
    which implies
    \begin{align}
        \limsup_{n\to\infty} -\frac{1}{n} \log( 1-\alpha_n(T_n) ) &\leq \frac{1}{k} \left( \sup_{\alpha>1}  \frac{\alpha-1}{\alpha} \bigg[ kr - D_\alpha \left( \mu_{\rho^{\otimes k}}^{M_k} \middle\| \mu_{\sigma^{\otimes k}}^{M_k} \right) \bigg] \right) + \frac{1}{k} \delta \\
        &= \sup_{\alpha>1}  \frac{\alpha-1}{\alpha} \bigg[ r - \frac{1}{k}  D_\alpha \left( \mu_{\rho^{\otimes k}}^{M_k} \middle\| \mu_{\sigma^{\otimes k}}^{M_k} \right) \bigg] + \frac{\delta}{k}
    \end{align}

    Taking the limit \(\delta\to0\), we obtain our claim.
\end{proof}

In Sec.\ \ref{Sec:Isotropic_States} and App.\ \ref{App:Werner_States}, we show that for examples of data hiding states the optimal sequence of measurements is given by the i.i.d.-version of the single copy optimal measurement, which is independent of \(\alpha\). Lemma \ref{Lem:Strong_Converse} then enables us to give a tight bound on the strong converse exponent and compute it explicitly for the these states (see Sec.\ \ref{Sec:Hypothesis_Testing_Isotropic} and App.\ \ref{Sec:Hypothesis_Testing_Werner} for these results). For these special cases, we then have the relationship
\begin{equation}
    \zeta_\text{SC}^\mathcal{M}(\rho,\sigma; r) = \sup_{\alpha>1} \frac{\alpha-1}{\alpha}\left[ r-D_\alpha^\mathcal{M} \left(\rho\middle\|\sigma\right)\right] \, .
\end{equation}


\subsection{Isotropic States}\label{Sec:Isotropic_States}

The first family of states for which we study the hypothesis testing problem under restricted measurments are the isotropic states \cite{Horodecki2}. Their defining property is that they are invariant under unitaries of the form \(U \otimes \bar{U}\), where the bar denotes complex conjugation. A general isotropic state can be characterized completely by a single parameter \(p\in[0,1]\) as 
\begin{equation}\label{Def:Isotropic_State}
    \dutchcal{i}(p) := p \Phi + (1-p) \Phi^\perp \, ,
\end{equation}
where \(\Phi\) denotes the maximally entangled state and \(\Phi^\perp\) is its orthogonal complement. Using the canonical basis \(\{\ket{i}\}_{i=1}^d\) on \(A\) and \(B\), respectively, these extremal isotropic states are given w.l.o.g.\ by
\begin{align}
    \Phi &= \frac{1}{d} \sum_{i,j =1}^d \ketbra{i}{j}_A \otimes \ketbra{i}{j}_B & \text{and} && \Phi^\perp &= \frac{1_{AB} - \Phi}{d^2 -1} \, .
\end{align} 
It is well-known that \(\dutchcal{i}(p)\) is separable and has PPT for \(p\in[0,1/d]\) and else it is entangled \cite{Horodecki2}. This class of states was previously studied in the context of LOCC distinguishability e.g.\ in \cite{Li, Christandl, Cheng}.


\subsubsection{Local Distinguishability}\label{Sec:Iso_States_LOCC}

We start with the most simple scenario of \(\rho=\Phi\) and \(\sigma=\Phi^\perp\). Notice that these states are orthogonal, so following the discussion in \ref{Sec:Boundedness} we can perform the test \( T_\Phi := \left\{ \Phi, 1_{AB} - \Phi \right\} \) to obtain  
\begin{align}
    \mu_{\Phi}^{T_\Phi} &= \{1,0\} & \text{and} && \mu_{\Phi^\perp}^{T_\Phi} &= \{0,1 \} \, ,
\end{align}
and therefore \( D^\text{ALL}_\alpha \left( \Phi \middle\| \Phi^\perp \right) = +\infty \) for all orders \(\alpha > 0\). 

However, it is important to note that the test \(T_\Phi\) employs the POVM element \(\Phi\) which has a negative partial transpose and as such is not a member of any of the locally-measured classes introduced in Sec.\ \ref{Sec:Preliminaries}. It turns out that if we restrict the available measurements to any of these classes they all perform equally poorly.
\begin{proposition}\label{Prop:Max_versus_Orthogonal}
    Let \( \mathcal{M} = \{ \text{LO}, \text{LOCC}_1, \text{LOCC}, \text{SEP}, \text{PPT} \} \) and \(\alpha > 0 \). With definitions as above, we have 
    \begin{equation}
        D_\alpha^\mathcal{M} \left(\Phi \middle\| \Phi^\perp \right) = \log(d+1) \, .
    \end{equation}
\end{proposition}

\begin{proof} 
    Note that the case \(\alpha=1\) was originally proven in \cite[Proposition 4]{Li}. The general idea of the proof is to lower-bound \(D^\text{LO}_\alpha\) and then find a matching upper-bound on \(D^\text{PPT}_\alpha\) via the variational formula. This strategy follows the meta idea that was used in \cite{Cheng,Christandl, Li} for similar computations.
    
    Let us start with the lower bound on \(D^\text{LO}_\alpha\). For this, we pick the local basis measurement with POVM elements \(L^{(i,j)} := \ketbra{i}{i}_A \otimes \ketbra{j}{j}_B\). Clearly, we have \(L \in \text{LO}\). The induced probability distributions are given by
    \begin{align}
        \mu_\Phi^{L}(i,j) &= \frac{\delta_{i,j}}{d} & \text{and} && \mu_{\Phi^\perp}^{L}(i,j) &= \frac{1}{d^2-1}\left( 1-\frac{\delta_{i,j}}{d}\right) \, ,
    \end{align}
    where \(\delta_{i,j}\) is the Kronecker delta. This yields by direct calculation the lower bound of
    \begin{equation}
        D^\text{LO}_\alpha \left(\Phi \middle\| \Phi^\perp \right) \geq D_\alpha \left( \mu_\Phi^L \middle\| \mu_{\Phi^\perp}^L \right) = \log(d+1) \, .
    \end{equation}
    
    For the second part, we employ the variational characterization to compute a bound on \(D^\text{PPT}_\alpha \). Since the bound is independent of \(\alpha\), we can focus on the case \(\alpha = +\infty \) due to the monotonicity in \(\alpha\). Here, we showed in Proposition \ref{Lem:Exact_Variational_Max} that the following holds
    \begin{equation}
        D^\text{PPT}_{\max}\left(\Phi\middle\|\Phi^\perp\right) = \sup_{\omega>0} \left\{ \log\tr[ \Phi \omega] \; \middle| \; \tr[ \Phi^\perp \omega ] = 1 \wedge \omega^\Gamma \geq 0 \right\} \, .
    \end{equation}
    Since we test two isotropic states against each other there is an inherent symmetry in the problem. We can take advantage of this similar to the approaches used in \cite{Cheng,Li,Christandl}. Namely, we have that \(\tr[ \dutchcal{i}(p) \omega ] = \tr[  \dutchcal{i}(p) \mathcal{I}(\omega) ] \) for any \(\omega \in \mathcal{P}\), whereby \(\mathcal{I}\) denotes the isotropic twirling operation. Its action on a linear operator \(X_{AB}\) is given by
    \begin{equation}
        \mathcal{I}\left(X_{AB}\right) = \int_{\mathcal{U}(d)} \left( U \otimes \bar{U} \right) X_{AB} \left( U \otimes \bar{U} \right)^\dagger \dd{\mu_H(U)} \, ,
    \end{equation}
    where the integration is with respect to the Haar measure \(\dd{\mu_H(U)}\) on the unitary group \(\mathcal{U}(d)\). Moreover, we note that the PPT property is preserved under application of the twirling channel.\footnote{This follows from local unitary invariance and convexity of the set of PPT operators.}
    
    Thus, we may restrict the optimization w.l.o.g.\ to isotropic PPT operators. We can decompose these as a conic combination of the projector on the maximally entangled state and its orthogonal complement, i.e.\ we can write
    \begin{equation}
        \mathcal{I}(\omega) = c_1 \Phi + c_2 (1_{AB} - \Phi) 
    \end{equation}
    for coefficients \(c_1,c_2 \geq 0\). Note that these have the PPT property if \(c_1 \leq (d+1)c_2\). 
    
    This allows us to rewrite the optimization problem as
    \begin{equation}
        \text{maximize} \; \; \log c_1 \; \; \text{s.t.} \; \; c_2 = 1 \wedge 0 < c_1 \leq (d+1)c_2 \, .
    \end{equation}
    This clearly has the solution \(\log(d+1)\) with gives us the matching upper bound.
\end{proof}

Our proof strategy generalizes to the scenario of testing \(\rho=\Phi\) against an arbitrary isotropic state \(\sigma=\dutchcal{i}(q)\). Note that since the states commute the optimal measurement in the unrestricted case is again \(T_\Phi\), which yields 
\begin{equation}
    D_\alpha^\text{ALL} \left(\Phi \middle\| i(q) \right) = - \log(q) \, .
\end{equation}
There then exists a finite gap to the locally-measured R\'enyi divergences.
\begin{proposition}\label{Prop:Max_versus_Isotropic}
    Let \( \mathcal{M} = \{ \text{LO}, \text{LOCC}_1, \text{LOCC}, \text{SEP}, \text{PPT} \} \), \(q\in[0,1]\) and \(\alpha > 0 \). With definitions as above, we have 
    \begin{equation}
        D_\alpha^\mathcal{M} \left(\Phi \middle\| i(q) \right) = \log(\frac{d +1}{qd+1}) \, .
    \end{equation}
\end{proposition}
\begin{proof}   
    We apply the same proof strategy as for Proposition \ref{Prop:Max_versus_Orthogonal}. First, observe that for \(q=1\) the claim reduces to \(D_\alpha^\mathcal{M} \left(\Phi \middle\| \Phi \right) = 0 \) which trivially holds. Thus, we assume \(q\not=1\) in the following. Next, observe that the local basis measurement \(L\) defined in the proof of Proposition \ref{Prop:Max_versus_Orthogonal} applied on a general isotropic state \(\dutchcal{i}(q)\) yields the probability distribution
    \begin{align}
        \mu_{\dutchcal{i}(q)}^{L}(i,j) = q \mu_\Phi^L(i,j) + (1-q)\mu_{\Phi^\perp}^L(i,j) = q \frac{\delta_{i,j}}{d} + \frac{1-q}{d^2-1}\left(1 - \frac{\delta_{i,j}}{d} \right)  \, .
    \end{align}
    By direct calculation, we then get the lower bound
    \begin{equation}
        D_\alpha \left( \mu_\Phi^L  \middle\| \mu_{\dutchcal{i}(q)}^L \right) = - \log( q + \frac{1 - q}{d+1} ) = \log(\frac{d +1}{qd+1}) \, .
    \end{equation} 

    For the upper bound, we use the twirling technique to rewrite the variational bound into
    \begin{equation}
        \text{maximize} \; \; \log c_1 \; \; \text{s.t.} \; \; q c_1 + (1-q) c_2 = 1 \wedge 0 < c_1 \leq (d+1)c_2 \, .
    \end{equation}
    The constraints can be recast into \(0 < c_1 \leq \frac{d+1}{qd+1}\) and this shows that the solution gives the matching upper bound.
\end{proof}

In Proposition \ref{Prop:Isotropic_Local} of App.\ \ref{App:Variational_Bound_Tight}, we additionally prove a result for the most general case of testing \(\rho=\dutchcal{i}(p)\) against \(\sigma=\dutchcal{i}(q)\). 


\subsubsection{Additivity on I.I.D.\ States}

Up to now, we only discussed comparing single copies of isotropic states. In order to apply our results to the problem of hypothesis testing, however, we need to consider the \(n\)-copy case as well. Here, the dual variational characterization of \(D_{\max}^\text{PPT}(\rho \| \sigma) \) comes in handy and enables us to compute the i.i.d.\ case in an elegant way. The following proposition states our result that the locally-measured R\'enyi divergences are additive on tensor powers of \(\rho=\Phi\) and \(\sigma=\Phi^\perp\).
\begin{proposition}\label{Prop:Additivty_Max_versus_Orthogonal}
Let \( \mathcal{M} = \{ \text{LO}, \text{LOCC}_1, \text{LOCC}, \text{SEP}, \text{PPT} \} \) and \(\alpha > 0 \). With definitions as above, we have 
    \begin{equation}
        D_\alpha^\mathcal{M} \left(\Phi^{\otimes n}\middle\|\left(\Phi^\perp \right)^{\otimes n} \right)  = n \log(d+1) \, .
    \end{equation}
\end{proposition}

\begin{proof}
    For the lower bound, we simply observe that measuring with the product POVM \(L^{\otimes n}\), with \(L\) defined as in the proof of Proposition \ref{Prop:Max_versus_Orthogonal}, directly gives the desired lower bound \(n \log(d+1)\) on \(D_\alpha^{\text{LO}}\).

    For the upper bound, we use the dual characterization of \(D_{\max}^\text{PPT}\). Note that any dual-feasible point gives an upper bound on the primal problem. The point we pick is \( \lambda = (d+1)^n \), \( X = 0 \) and
    \begin{align}
        Y &= \frac{1}{(d-1)^n} \left[ \left(1_{AB} -  \frac{F}{d}\right)^{\otimes n} - (d-1)^n \left( \frac{F}{d} \right)^{\otimes n} \right] \, ,
    \end{align}
    where \(F := \sum_{i,j} \ketbra{i}{j}_A \otimes \ketbra{j}{i}_B \) denotes the swap operator. The feasibility of \(\lambda\) and \(X\) is clear. We have to check that \(Y\) is positive semi-definite and that the side constraint is satisfied.

    For this, note that \(Y\) has eigenvalues 
    \begin{align}
        \frac{1}{(d-1)^n} \left[ \left(1 \mp \frac{1}{d}\right)^n - (d-1)^n \left( \frac{\pm 1}{d} \right)^n \right] \geq \frac{1}{(d-1)^n} \left[ \left(1 - \frac{1}{d}\right)^n - (d-1)^n \left( \frac{1}{d} \right)^n \right] = 0\,,
    \end{align}
    and thus \(Y\) is positive semi-definite. Further, observe that 
    \begin{align}
        Y^\Gamma &= \frac{1}{(d-1)^n} \left[ \left(1_{AB} -  \Phi \right)^{\otimes n} - (d-1)^n \Phi^{\otimes n} \right] \\ 
        &= (d+1)^n \left( \frac{ 1_{AB} -  \Phi}{d^2-1} \right)^{\otimes n} - \Phi^{\otimes n}
        = \lambda \left(\Phi^\perp\right)^{\otimes n} - \Phi^{\otimes n}\,.
    \end{align} 
    Therefore, the point is indeed dual-feasible. With this, we obtain the matching upper bound.
\end{proof}

We can generalize this proof technique to prove an additivity result for the locally-measured max-divergence in the most general case \(\rho=\dutchcal{i}(p)\) and \(\sigma=\dutchcal{i}(q)\).
\begin{proposition}\label{Prop:Additivity_Max_Divergence}
Let \( \mathcal{M} = \{ \text{LO}, \text{LOCC}_1, \text{LOCC}, \text{SEP}, \text{PPT} \} \), \(p \in [0,1]\) and \(q \in [0,1]\). With definitions as above, we have
\begin{align}
    D_{\max}^\mathcal{M} \left( \dutchcal{i}(p) ^{\otimes n} \middle\| \dutchcal{i}(q) ^{\otimes n} \right)  =  n \log \left(\frac{pd+1}{qd+1}\right), \label{eq:isotropub}
\end{align}  
provided one of the following conditions is satisfied:
\begin{enumerate}
    \item \( p \leq \frac{1}{d} \) (separable), \( q \leq \frac{1}{d} \) (separable) and \(q \leq p\),
    \item  \( p \geq \frac{1}{d} \) (entangled), \( q \leq \frac{1}{d} \) (separable) and \(qp\leq \frac{1}{d^2}\).
    \item \( p \geq \frac{1}{d} \) (entangled), \(q \geq \frac{1}{d}\) (entangled) and \( p = q\), i.e.\ they have to be equal.
\end{enumerate} 
\end{proposition}

\begin{proof}

For the lower bound, consider the tensor product measurement \(L^{\otimes n}\) as used in Proposition \ref{Prop:Additivty_Max_versus_Orthogonal}.  This gives the bound
\begin{equation}
    D_{\max}^\text{LO} \left( \dutchcal{i}(p)^{\otimes n} \middle\| \dutchcal{i}(q)^{\otimes n} \right) \geq n D_{\max} \left( \mu_{\dutchcal{i}(p)}^L \middle\| \mu_{\dutchcal{i}(q)}^L\right) = n \log \left(\frac{pd+1}{qd+1}\right)
\end{equation}
where the final equality holds under the condition that \( q \leq p\). 

For the upper bound, we use the dual characterization of \(D_{\max}^\text{PPT}\) from Corollary \ref{Lem:Dual_PPT_Max}. We pick the dual point given by
\begin{align}
    \lambda &= \left(\frac{pd+1}{qd+1}\right)^n, & X &= 0, & Y^\Gamma = \lambda \dutchcal{i}(q)^{\otimes n} - \dutchcal{i}(p)^{\otimes n} \, .
\end{align}

We have to show that \(Y\) is positive semi-definite in order to prove that this is indeed dual-feasible. For this, note that we have
\begin{equation}
    Y= \frac{1}{d^n(d^2-1)^n}\left( \lambda \left( (1-q)d 1_\dutchcal{H} +(d^2q -1) F \right)^{\otimes n}- \left( (1-p)d 1_\dutchcal{H} +(d^2p-1) F \right)^{\otimes n}\right)
\end{equation}
The eigenvalues of \(Y\) are given by
\begin{align}
    \frac{\lambda}{d^n} \frac{(1+qd)^k}{(d+1)^k}\frac{(1-qd)^{n-k}}{(d-1)^{n-k}}-\frac{1}{d^n}\frac{(1+pd)^k}{(d+1)^k}\frac{(1-pd)^{n-k}}{(d-1)^{n-k}} \; \; \text{for} \; \; 0\leq k\leq n \, .
\end{align}

The non-negativity requirement of the eigenvalues for \(0\leq k\leq n\) translates into
\begin{equation}
    \left( \frac{pd+1}{qd+1} \right)^{n-k} (1-qd)^{n-k} \geq (1-pd)^{n-k} \, .
\end{equation}
We then have the following cases:
\begin{enumerate}
    \item If \(p\leq \frac{1}{d}\) and \(q \leq \frac{1}{d}\), these conditions are satisfied if \(q\leq p\). 

    \item If \( p \geq \frac{1}{d}\) and \(q \leq \frac{1}{d}\), they reduce to the requirement that \(pq \leq \frac{1}{d^2} \).

    \item If \( p \geq \frac{1}{d}\) and \(q \geq \frac{1}{d}\), they only hold in the case \(p=q\).
\end{enumerate}
Hence, under the conditions stated in the proposition, we have by weak duality that
\begin{align}
     D_{\max}^\text{PPT} \left( \dutchcal{i}(p)^{\otimes n} \middle\| \dutchcal{i}(q)^{\otimes n} \right) \leq n \log \lambda =n \log \left(\frac{pd+1}{qd+1}\right),
\end{align}
which gives the matching upper bound and completes the proof.
\end{proof}

We remark that it is possible to adapt an idea in the proof of \cite[Proposition 7]{Cheng} to extend the twirling technique we use in App.\ \ref{App:Variational_Bound_Tight} to compute the primal programs for general orders \(\alpha\) and general single-letter isotropic states to the \(n\)-copy case. This offers a potential route to investigate the additivity of the locally-measured R\'enyi divergences on isotropic states when the single-letter characterization is \(\alpha\)-dependent. Note, however, that this has the caveat that the variational bounds do not necessarily give tight bounds in the general case\,---\,as we show in App.\ \ref{App:Variational_Bound_Tight}.


\subsubsection{Application to Hypothesis Testing}\label{Sec:Hypothesis_Testing_Isotropic}

Building on our work in Sec.\ \ref{Sec:Hypothesis_Testing}, we can use these results to generalize a result of Cheng et.\ al.\ in \cite[Corollary 11]{Cheng}. There they computed the optimal Stein's exponent for testing the maximally entangled state versus its orthogonal complement
\begin{equation}
    \zeta_\text{Stein}^\mathcal{M}\left(\Phi,\Phi^\perp;\varepsilon\right) = \log(d+1) \, .
\end{equation}
The additivity result of Proposition \ref{Prop:Additivty_Max_versus_Orthogonal} then allows us to re-derive this hypothesis testing result in a simple manner since the regularized measured relative entropy becomes single-letter for $\rho=\Phi$, $\sigma=\Phi^\perp$ and the strong converse part becomes tight. Moreover, combining the results of Propositions \ref{Prop:Max_versus_Isotropic} and \ref{Prop:Additivity_Max_Divergence} allows us to extend this result to 
\begin{equation}
    \zeta_\text{Stein}^\mathcal{M}\left(\Phi,\dutchcal{i}(q);\varepsilon\right) = \log(\frac{d+1}{qd+1}) \; \; \text{if} \; \; q \leq \frac{1}{d^2}\, .
\end{equation}
Along similar lines, we can generalize the result of Cheng et.\ al.\ for \(q=0\) of the strong converse exponent \cite[Corollary 13]{Cheng}. We obtain
\begin{equation}
    \zeta_\text{SC}^\mathcal{M}\left(\Phi,\dutchcal{i}(q);r\right) = r- \log(\frac{d+1}{qd+1}) \; \; \text{if} \; \; q \leq \frac{1}{d^2}\, 
\end{equation}
for all rates \(r \geq \log(\frac{d+1}{qd+1}) \).


\section{Outlook}

We left open the problem of the Chernoff and Hoeffding exponents for the restricted hypothesis testing problem on data-hiding states (cf.\ again the results in \cite[Table 1]{Cheng}). Here, we already have an achievability result in terms of the regularized measured R\'enyi divergences but the converse part remains open. In that respect, Brandao et.\ al.\ offered in \cite[Conjecture 17]{Brandao} a possible expression for the Chernoff exponent under restricted measurements. This is a rather subtle point though since in the unrestricted case, we know that the regularized measured R\'enyi divergence does not give the optimal error exponent. 

Another possible research direction is to investigate further operational uses of the measured R\'enyi divergence in quantum information theory, e.g.\,, in the direction of\cite{Cheng2}. On the technical side, it would be interesting if an exact variational characterization can be obtained for the measurement classes SEP and PPT for general orders \(\alpha\) and especially for the special case \(\alpha=1\). Moreover, one could consider measured versions of other divergences, a possible example being the Hilbert \(\alpha\)-divergences \cite{Buscemi} (see also \cite{Wilde2}).


\section*{Acknowledgments}

We thank Hao-Chung Cheng, Aditya Nema, and Julius Zeiss for discussions and feedback. MB, TR, and SS acknowledge support from the Excellence Cluster - Matter and Light for Quantum Computing (ML4Q). MB and TR acknowledge funding by the European Research Council (ERC Grant Agreement No. 948139).




\appendix

\section{Infinite-Dimensional Extension}\label{App:Infinite-Dimensional_Version}

Here, we generalize Theorem \ref{Thm:Variational_Q} to the infinite-dimensional case. This result can be seen as a generalization of \cite[Lemma 20]{Ferrari}, which discussed the special case \(\alpha=1\) and \(\mathcal{M}=\text{ALL}\), and was helpful for studying continuous variable resource theories.

First, we need to introduce some notation. Consider a separable Hilbert space $\mathbb{H}$, and quantum density operators $\rho,\sigma$ on $\mathbb{H}$. 
A positive operator-valued measure (POVM) \(M\) over a finite set \(\mathcal{Z}\) is a map from \(\mathcal{Z}\) to \(\mathcal{P}\) which satisfies \(\sum_{z \in \mathcal{Z}} M^z = 1_\mathbb{H} \). For a given \(\rho \in \mathcal{P}\), the POVM induces a positive measure \( \mu_\rho^M\) over \(\mathcal{Z}\) according to the Born rule: \(\mu_\rho^M(z) = \tr[\rho M^z] \). With this, the measured R\'enyi divergences can be defined as in the finite-dimensional case. Moreover, let 
\begin{equation}
    C_{\mathcal{M}} := \bigg\{\omega: \omega=\sum_{z \in \mathcal{Z}} \lambda_z M^z,  \lambda_z \geq 0, ~\forall~z \in \mathcal{Z}, M \in \mathcal{M} \bigg\}
\end{equation}
be the cone of positive definite operators corresponding to $\mathcal{M}$. Note that $C_{\mathcal{M}}  $ is a subset of the space of bounded positive definite  operators on $\mathbb{H}$. 

\begin{theorem}\label{Thm:Infinite_Dimensional}
For \(\rho \in \mathcal{S} \), \(\sigma \in \mathcal{S}\) and \(\mathcal{M} \subseteq \text{ALL}\), we have \\
1) for \(\alpha \in \left(0,1/2\right)\) that
\begin{align}
    Q_{\alpha}^{\mathcal{M}} \left(\rho \middle\| \sigma \right) \geq& \inf_{\stackrel{\omega >0}{\omega \in C_\mathcal{M}}} \alpha \tr[\rho \omega ] + (1-\alpha) \tr[\sigma \omega^{\frac{\alpha}{\alpha-1}}] \label{eq:var1} = \inf_{\stackrel{\omega >0}{\omega \in C_\mathcal{M}}} \tr[ \rho \omega ]^{\alpha} \tr[\sigma \omega^{\frac{\alpha}{\alpha-1}}]^{1-\alpha} \, ,
\end{align} 
2) for \(\alpha \in \left[1/2,1\right)\) that
\begin{align}
    Q_{\alpha}^{\mathcal{M}}\left(\rho \middle\| \sigma \right) \geq& \inf_{\stackrel{\omega >0}{\omega \in C_\mathcal{M}}} \alpha \tr[\rho \omega^{\frac{\alpha-1}{\alpha}}] + (1-\alpha) \tr[\sigma \omega] = \inf_{\stackrel{\omega >0}{\omega \in C_\mathcal{M}}} \tr[\rho \omega^{\frac{\alpha-1}{\alpha}}]^{\alpha} \tr[\sigma \omega]^{1-\alpha} \, ,
\end{align}
and 3) for \(\alpha \in (1,\infty)\) that
\begin{align}
    Q_{\alpha}^{\mathcal{M}} \left( \rho \middle\| \sigma \right) \leq& \sup_{\stackrel{\omega >0}{\omega \in C_\mathcal{M}}} \alpha \tr[\rho \omega^{\frac{\alpha-1}{\alpha}}] + (1-\alpha) \tr[\sigma \omega] = \sup_{\stackrel{\omega >0}{\omega \in C_\mathcal{M}}} \tr[ \rho \omega^{\frac{\alpha-1}{\alpha}}]^{\alpha} \tr[\sigma \omega]^{1-\alpha} \, . \label{eq:var2}
\end{align}
If $\mathcal{M}=\text{ALL}$, we have equality in the above.
\end{theorem}

\begin{proof}
Note that any  POVM $M=\{M^z\}_{z \in \mathcal{Z}} \in \mathcal{M}$ has $|\mathcal{Z}|$  finite by definition. Hence, the proof of the inequalities above follows via analogous steps to the finite dimensional case using the operator Jensen's inequality. Moreover, the equalities between the two variational expressions also follow by the same arguments. Hence, we only need to show the reverse inequality when $\mathcal{M}=\text{ALL}$, i.e., \\
1) for \(\alpha \in \left(0,1/2\right)\):
\begin{align}
    Q_{\alpha}^\text{ALL} \left(\rho \middle\| \sigma \right) \leq& \inf_{\stackrel{\omega >0}{\omega \in C_\text{ALL}}} \alpha \tr[\rho \omega ] + (1-\alpha) \tr[\sigma \omega^{\frac{\alpha}{\alpha-1}}], \label{eq:var1rev}
\end{align} 
2) for \(\alpha \in \left[1/2,1\right)\):
\begin{align}
    Q_{\alpha}^\text{ALL}\left(\rho \middle\| \sigma \right) \leq& \inf_{\stackrel{\omega >0}{\omega \in C_\text{ALL}}} \alpha \tr[\rho \omega^{\frac{\alpha-1}{\alpha}}] + (1-\alpha) \tr[\sigma \omega], \label{eq:var2rev}
\end{align}
and 3) for \(\alpha \in (1,\infty)\): 
\begin{align}
    Q_{\alpha}^\text{ALL} \left( \rho \middle\| \sigma \right) \geq& \sup_{\stackrel{\omega >0}{\omega \in C_\text{ALL}}} \alpha \tr[\rho \omega^{\frac{\alpha-1}{\alpha}}] + (1-\alpha) \tr[\sigma \omega]. \label{eq:var3rev}
\end{align}
We first show \eqref{eq:var3rev}. Note that if $\rho \not\ll \sigma$, then there exists a projector $P_i$ such that $\tr[\rho P_i]>0$ and $\tr[\sigma P_i]=0$. Hence, $Q_{\alpha}^\text{ALL} \left( \rho \middle\| \sigma \right)=+\infty$ and there is nothing to prove. Hence, assume that $\rho \ll \sigma$. We will first show that the supremum in \eqref{eq:var3rev} can be restricted to be over finite rank operators in $C_\text{ALL}$ without changing the value of the supremum. Consider any $\omega \in C_\text{ALL}$. Since $\omega$ is a bounded operator, there exists $m>0$ such that $\omega \leq m 1_\mathbb{H}$. Since $\rho$ and $\sigma$ are density operators, for any $\epsilon >0$, there exists a finite rank projector $P_{\epsilon}$ such that $\norm{P_{\epsilon}\rho P_{\epsilon}-\rho}_1 \leq \epsilon$ and $\norm{P_{\epsilon}\sigma P_{\epsilon}-\sigma}_1 \leq \epsilon$. Then, we have
\begin{align}
 \alpha \tr[\rho \omega^{\frac{\alpha-1}{\alpha}}] + (1-\alpha) \tr[\sigma \omega]& \leq \alpha \tr[P_{\epsilon}\rho P_{\epsilon} \omega^{\frac{\alpha-1}{\alpha}}] + (1-\alpha) \tr[P_{\epsilon}\rho P_{\epsilon}\omega] +m \epsilon \notag \\
 & \leq \alpha \tr[\rho \big(P_{\epsilon} \omega P_{\epsilon}\big)^{\frac{\alpha-1}{\alpha}}] + (1-\alpha) \tr[\sigma P_{\epsilon}\omega P_{\epsilon}] +m \epsilon, \notag
\end{align}
where the final inequality follows by operator Jensen's inequality applied to the concave map $x \mapsto x^{(\alpha-1)/\alpha}$ for $\alpha>1$.
Since $\epsilon>0$ is arbitrary, the claim follows. 

Now consider any $\omega \in C_\text{ALL}$ of finite rank, and let $\omega=\sum_{i=1}^r \lambda_i P_i $ be its spectral decomposition. Then, we have 
\begin{align}
\alpha \tr[\rho \omega^{\frac{\alpha-1}{\alpha}}] + (1-\alpha) \tr[\sigma \omega]=\sum_{i=1}^r \alpha  \lambda_i^{\frac{\alpha-1}{\alpha}}\tr[\rho P_i]+(1-\alpha) \lambda_i\tr[\sigma P_i]
\end{align}
Consider the supremum of each term in the summation in the RHS over $\lambda_i$. When $\tr[\sigma P_i]=\tr[\rho P_i]=0$, $\lambda_i$ can be chosen arbitrarily and the supremum is zero. When $\tr[\rho P_i]=0$ and $\tr[\sigma P_i]>0$, the supremum is $0$ since $1-\alpha <0$  (approached as $\lambda_i \downarrow 0$). For $i$ such that $\tr[\sigma P_i]>0$ and $\tr[\rho P_i]>0$,  the supremum is achieved at $\lambda_i=(\tr[\rho P_i]/\tr[\sigma P_i])^{\alpha}$. From this, we obtain
\begin{align}
    \alpha \tr[\rho \omega^{\frac{\alpha-1}{\alpha}}] + (1-\alpha) \tr[\sigma \omega]&=\sum_{i=1}^r \alpha  \lambda_i^{\frac{\alpha-1}{\alpha}}\tr[\rho P_i]+(1-\alpha) \lambda_i\tr[\sigma P_i] \\
    &\leq \sum_{i} (\tr[\rho P_i])^{\alpha}(\tr[\sigma P_i])^{1-\alpha} \\
    & \leq \sup_{M \in \text{ALL}} Q_\alpha \left( \mu_\rho^M \middle\| \mu_\sigma^M \right), 
    \end{align}
where the sum in the penultimate equation is over all indices $i$ such that $\tr[\sigma P_i]>0$ and $\tr[\rho P_i]>0$. Taking supremum over all $\omega \in C_\text{ALL}$ with finite rank then proves the claim for $\alpha \in (1,\infty)$. 

The proof of \eqref{eq:var1rev} and \eqref{eq:var2rev} are similar. Hence, we only show the former. Note that for $\alpha \in (0,1)$ and $\omega>0$, both the terms $\alpha \tr[\rho \omega ]$ and $(1-\alpha) \tr[\sigma \omega^{\frac{\alpha}{\alpha-1}}]$ are non-negative. Fix an arbitrary $\epsilon>0$. Considering the spectral decomposition $\omega=\sum_{i=1}^{\infty} \lambda_i P_i $,  choose a finite rank projector $P_{\epsilon}=\sum_{i} P_i$ of rank $r$ constructed from the projectors in the spectral decomposition of $\omega$ such that $(\tr[\rho (1_\mathbb{H}-P_{\epsilon})\rho]^{\alpha}(\tr[\sigma(1_\mathbb{H}-P_{\epsilon})])^{1-\alpha} \leq \epsilon$. Then, we have 
\begin{align}
    \alpha \tr[\rho \omega^{\frac{\alpha-1}{\alpha}}] + (1-\alpha) \tr[\sigma \omega]&=\sum_{i=1}^{\infty} \alpha  \lambda_i^{\frac{\alpha-1}{\alpha}}\tr[\rho P_i]+(1-\alpha) \lambda_i\tr[\sigma P_i]   \\
    & \geq \sum_{i=1}^{r} \alpha  \lambda_i^{\frac{\alpha-1}{\alpha}}\tr[\rho P_i]+(1-\alpha) \lambda_i\tr[\sigma P_i], 
\end{align}
where the inequality follows since the omitted terms are positive. For a fixed $i$, the infimum over $\lambda_i$ for each term within the sum is achieved at $\lambda_i=(\tr[\rho P_i]/\tr[\sigma P_i])^{\alpha}$. Substituting this yields that 
\begin{align}
    &\alpha \tr[\rho \omega^{\frac{\alpha-1}{\alpha}}] + (1-\alpha) \tr[\sigma \omega] \\
    &\geq \sum_{i=1}^r (\tr[\rho P_i])^{\alpha}(\tr[\sigma P_i])^{1-\alpha} \\
    & \geq \sum_{i=1}^r (\tr[\rho P_i])^{\alpha}(\tr[\sigma P_i])^{1-\alpha}+(\tr[\rho (1_\mathbb{H}-P_{\epsilon})\rho]^{\alpha}(\tr[\sigma(1_\mathbb{H}-P_{\epsilon})])^{1-\alpha}-\epsilon \notag \\
    & \geq \inf_{M \in \text{ALL}} Q_\alpha \left( \mu_\rho^M \middle\| \mu_\sigma^M \right)-\epsilon. \notag 
\end{align}
Taking infimum over $\omega \in C_\text{ALL}$ and noting that $\epsilon>0$ is arbitrary completes the proof of  the claim. 
\end{proof}

Moreover, it would be neat to extend the
concept of measured R\'enyi divergences to general, continuous POVMs described by measure spaces. Many of the steps in the proof of Theorem \ref{Thm:Infinite_Dimensional} still go through and we would like to point to a Jensen inequality for
operator-valued measures \cite{Farenick1}. It seems, however, that an extension of this inequality would be
needed \cite{Farenick2}.


\section{Variational Bounds for Isotropic States}\label{App:Variational_Bound_Tight}

Recalling our discussion in Sec.\ \ref{Sec:Local_Variational_Formula}, we argued that in general the variational bounds do not give tight bounds. In fact, as we show in the following, we can find parameters \(p\) and \(q\) such that testing the isotropic states \(\rho=\dutchcal{i}(p)\) and \(\sigma=\dutchcal{i}(q)\) yields bounds that are not tight (see Sec.\ \ref{Sec:Isotropic_States} for definitions). 

In order to make this argument we have to compute the locally-measured R\'enyi divergence via a different strategy then in Sec.\ \ref{Sec:Iso_States_LOCC}. Let us first consider the case of unrestricted measurements. Since the states commute, the optimal measurement is the measurement in the eigenbasis, i.e.\ the test \(T_\Phi\) as defined above. This gives for \(\alpha\in(0,1)\cup(1,\infty)\) the result
\begin{equation}\label{Eq:Isotropic_All_Alphha}
    D_\alpha^\text{ALL} \left( \dutchcal{i}(p) \big\| \dutchcal{i}(q) \right) = \frac{1}{\alpha-1} \log( p^\alpha q^{1-\alpha} + (1-p)^\alpha (1-q)^{1-\alpha} )
\end{equation}
with the limiting cases 
\begin{equation}\label{Eq:Isotropic_All_1}
    D^\text{ALL} \left( \dutchcal{i}(p) \big\| \dutchcal{i}(q) \right) = p \log(\frac{p}{q}) + (1-p)\log(\frac{1-p}{1-q})
\end{equation}
and 
\begin{equation}
    D^\text{ALL}_{\max}\left( \dutchcal{i}(p) \big\| \dutchcal{i}(q) \right) = \log \max \left\{ \frac{p}{q}, \frac{1-p}{1-q} \right\} \, .
\end{equation}

For the locality-constrained measurement classes, we adapt the proof strategy of \cite[Proposition 4]{Li} to obtain the optimal measurement for the set PPT. The proof is then completed by the observation that the local basis measurement achieves the same R\'enyi divergence. That is, as in the special cases discussed above, the locality-constrained measurement sets all perform equally well on isotropic states.

\begin{proposition}\label{Prop:Isotropic_Local}
    Let \( \mathcal{M} = \{ \text{LO}, \text{LOCC}_1, \text{LOCC}, \text{SEP}, \text{PPT} \} \), \(p\in[0,1]\), \(q\in[0,1]\) and \(\alpha \in (0,1) \cup (1,+\infty) \). With the definitions as above, we have 
    \begin{equation}
        Q_\alpha^\mathcal{M}\left( \dutchcal{i}(p) \big\| \dutchcal{i}(q) \right) = \frac{d}{d+1} \left[ \left( p + \frac{1}{d} \right)^\alpha \left( q +\frac{1}{d} \right)^{1-\alpha} + (1-p)^\alpha (1-q)^{1-\alpha} \right] \, .
    \end{equation}
    Additionally, we have the limiting cases
    \begin{align}
        D^\mathcal{M}\left( \dutchcal{i}(p) \big\| \dutchcal{i}(q) \right) &= \frac{d}{d+1} \left[ \left(p +\frac{1}{d}\right) \log(\frac{1+pd}{1+qd}) + (1-p)\log(\frac{1-p}{1-q}) \right]\\
        D^\mathcal{M}_{\max}\left( \dutchcal{i}(p) \big\| \dutchcal{i}(q) \right) &= \log \max \left\{ \frac{1+pd}{1+qd}, \frac{1-p}{1-q} \right\} \, .
    \end{align}
\end{proposition}

\begin{proof} 
    The claim for the measured relative entropy was originally proven in \cite[Lemma 24]{Christandl} using the proof strategy of \cite[Proposition 4]{Li}. We adapt the proof idea to the general \(\alpha\)-orders. 

    We can compute \(Q_\alpha^\text{PPT}\left( \dutchcal{i}(p) \big\| \dutchcal{i}(q) \right)\) exactly by taking advantage of symmetry. We start with the case \(\alpha\in(1,+\infty)\). Since we have \(\tr[\dutchcal{i}(p) M] = \tr[\dutchcal{i}(p) \mathcal{I}(M)]\) for any \(M\in\mathcal{P}\), we can restrict the optimization over POVMs w.l.o.g.\ to isotropic PPT measurements. Moreover, an arbitrary isotropic PPT operator can be decomposed into a conic combination of the extremal PPT operators
    \begin{align}
        I_1 &= \Phi + \frac{1}{d+1} \Phi^\perp & \text{and} && I_2 = \frac{d}{d+1} \Phi^\perp \, .
    \end{align}
    Due to the joint-convexity of \(Q_\alpha^\text{PPT}\), we can fine-grain each of the POVMs into operators proportional to the extremal ones and then join them back together into a binary measurement by using that 
    \begin{equation}
        (ax)^\alpha (ay)^{1-\alpha} + (bx)^\alpha (by)^{1-\alpha} = \big((a+b)x \big)^\alpha \big((a+b)y \big)^{1-\alpha} \, .
    \end{equation}  
    The optimal measurement is thus seen to be the binary measurement \(I := \{I_1, I_2\}\). The induced probability distribution of this binary measurement on an arbitrary isotropic state is given by
    \begin{equation}
        \mu_{\dutchcal{i}(p)}^I = \left\{ p + \frac{1-p}{d+1}, \frac{d}{d+1} (1-p) \right\}
    \end{equation}
    which then yields 
    \begin{align}
        Q_\alpha \left( \mu_{\dutchcal{i}(p)}^I \middle\| \mu_{\dutchcal{i}(q)}^I\right) &= \left( p + \frac{1-p}{d+1} \right)^\alpha \left( q + \frac{1-q}{d+1} \right)^{1-\alpha} + \frac{d}{d+1} \left(1-p \right)^{\alpha} \left(1-q\right)^{1-\alpha} \\
        &=\frac{d}{d+1} \left[ \left( p + \frac{1}{d} \right)^\alpha \left( q +\frac{1}{d} \right)^{1-\alpha} + (1-p)^\alpha (1-q)^{1-\alpha} \right] \, .
    \end{align}
    The case \(\alpha\in(0,1)\) can be proven analogously by using the joint concavity of \(Q_\alpha^\text{PPT}\). Then, we have by Property 3 of Lemma \ref{Lem:General_Properties} that 
    \begin{equation}
        D^\text{PPT}\left( \dutchcal{i}(p) \big\| \dutchcal{i}(q) \right) = \sup_{\alpha\in(0,1)} D_\alpha \left( \mu_{\dutchcal{i}(p)}^I \middle\| \mu_{\dutchcal{i}(q)}^I\right) = D \left( \mu_{\dutchcal{i}(p)}^I \middle\| \mu_{\dutchcal{i}(q)}^I\right) 
    \end{equation}
    and 
    \begin{equation}
        D^\text{PPT}_{\max}\left( \dutchcal{i}(p) \big\| \dutchcal{i}(q) \right) = \sup_{\alpha\in(1,\infty)} D_\alpha \left( \mu_{\dutchcal{i}(p)}^I \middle\| \mu_{\dutchcal{i}(q)}^I\right) = D_{\max} \left( \mu_{\dutchcal{i}(p)}^I \middle\| \mu_{\dutchcal{i}(q)}^I\right) \, ,
    \end{equation}
    where the final equalities follow by the continuity of the classical quantity.

    The proof is completed by observing that the local basis measurement \(L\) defined as \(L^{(i,j)} := \ketbra{i}{i}_A \otimes \ketbra{j}{j}_B\) achieves the same \(D_\alpha\). The \(d^2\) probabilities \(\mu_{\dutchcal{i}(p)}^L(i,j)\) can be grouped naturally into two classes. If \(i=j\), we compute 
    \begin{align}
        \sum_{k=1}^{d} \left(\frac{p}{d}+\frac{1-p}{d(d+1)}\right)^{\alpha} \left(\frac{q}{d}+\frac{1-q}{d(d+1)}\right)^{1-\alpha} &= \left( p + \frac{1-p}{d+1} \right)^\alpha \left( q + \frac{1-q}{d+1} \right)^{1-\alpha}
    \end{align}
    and else we have
    \begin{align}
        \sum_{k=1}^{d^2-d} \left(\frac{1-p}{d^2-1}\right)^{\alpha} \left(\frac{1-q}{d^2-1}\right)^{1-\alpha} = \frac{d}{d+1} (1-p)^\alpha (1-q)^{1-\alpha}  \, .
    \end{align}
    Combining these results, we get 
    \begin{equation}
        Q_\alpha^\text{LO}\left( \dutchcal{i}(p) \big\| \dutchcal{i}(q) \right) \geq Q_\alpha \left( \mu_{\dutchcal{i}(p)}^L \middle\| \mu_{\dutchcal{i}(q)}^L\right) = Q_\alpha \left( \mu_{\dutchcal{i}(p)}^I \middle\| \mu_{\dutchcal{i}(q)}^I\right) \, .
    \end{equation}
    This concludes the proof. 
\end{proof}

We are now in the position to show that the variational bounds on the locally-measured R\'enyi divergences are not tight for the orders \(\alpha\in(0,+\infty)\) on general isotropic states. We develop this argument in the following four propositions by providing explicit counterexamples for each range of orders \(\alpha\).


\subsection{Example for \texorpdfstring{\(\alpha\in(0,1/2)\)}{alpha in (0,1/2)}}

\begin{proposition}\label{Prop:Gap}
    Let \(\mathcal{M} := \{ \text{SEP}, \text{PPT} \} \), \(q \in [0,1]\) and \(\alpha\in(0,1/2)\). With definitions as above, we have 
    \begin{equation}
        V_\alpha^\mathcal{M}\left(\Phi, \dutchcal{i}(q) \right) = D_\alpha^\text{ALL}\left(\Phi \middle\| \dutchcal{i}(q) \right)
    \end{equation}
    Consequently, for \(q \in [0,1)\) the gap between the measured R\'enyi divergence and the variational bound is strict, i.e.\ 
    \begin{equation}
        D_\alpha^\mathcal{M} \left( \Phi \middle\| \dutchcal{i}(q) \right) < V_\alpha^\mathcal{M}\left(\Phi, \dutchcal{i}(q) \right) \, .
    \end{equation}
\end{proposition}

\begin{proof}
    Since the objective function \(\eta_\alpha\) is scaling invariant in \(\omega\), we have the relation
    \begin{equation}
        \exp( - V_\alpha^\mathcal{M}\left(\Phi, \dutchcal{i}(q) \right) ) = \inf_{\omega > 0} \left\{ \tr[ \dutchcal{i}(q) \omega^\frac{\alpha}{\alpha-1}] \; \middle| \; \tr[\Phi \omega ] = 1 \wedge \omega \in C_\mathcal{M} \right\} \, .
    \end{equation}

    We show now that we can restrict w.l.o.g.\ to an optimization over isotropic operators \(\omega \in C_\mathcal{M}\). 
    For this, note that for any feasible \(\omega \in C_\mathcal{M}\), \(\mathcal{I}(\omega)\) is a feasible point as well since \(\tr[\Phi \omega ] = \tr[\Phi \mathcal{I}(\omega) ]\) and \(\mathcal{I}(\omega) \in C_\mathcal{M}\).\footnote{This follows from local unitary invariance and convexity of the respective cones.} Moreover, the objective value corresponding to \(\mathcal{I}(\omega)\) is always smaller due to 
    \begin{equation}
        \tr[ \dutchcal{i}(q) \omega^\frac{\alpha}{\alpha-1}] = \tr[ \dutchcal{i}(q) \mathcal{I}\left(\omega^\frac{\alpha}{\alpha-1}\right)] \geq \tr[ \dutchcal{i}(q) \mathcal{I}(\omega)^\frac{\alpha}{\alpha-1}] \, ,
    \end{equation}
    where the inequality step follows by operator convexity of \( f(t) := t^{\frac{\alpha}{\alpha-1}}\) for \(\alpha\in(0,1/2)\) \cite[Section V]{Bhatia} together with the Jensen inequality \cite{Farenick1}. 

    Since isotropic operators are separable if and only if they have PPT, the program we have to evaluate in each case is given by
    \begin{equation}
        \text{minimize} \; \; q + (1-q)c_2^{\frac{\alpha}{\alpha-1}} \; \; \text{s.t.} \; \; c_2 \geq \frac{1}{d+1}
    \end{equation}
    
    Notice that the objective function is monotone decreasing on \(c_2 \in (0,\infty)\) and thus its minimal value is given by \(q\). We then conclude with Eq.\ \eqref{Eq:Isotropic_All_Alphha} that \(V_\alpha^\mathcal{M}\left(\Phi, \dutchcal{i}(q) \right) = D_\alpha^\text{ALL}\left(\Phi, \dutchcal{i}(q) \right)\).

    Lastly, we have with Proposition \ref{Prop:Isotropic_Local} for \(q\in[0,1)\) that
    \begin{equation}
        D_\alpha^\mathcal{M} \left( \Phi \middle\| \dutchcal{i}(q) \right) = \log(\frac{d+1}{qd+1}) < - \log q = V_\alpha^\mathcal{M}\left(\Phi, \dutchcal{i}(q) \right).
    \end{equation}
\end{proof}


\subsection{Example for \texorpdfstring{\(\alpha\in [1/2,1)\)}{alpha in [1/2,1)} }

\begin{proposition}
    Let \(\mathcal{M} := \{ \text{SEP}, \text{PPT} \} \), \(p \in [0,1]\) and \(\alpha\in[1/2,1)\). With definitions as above, we have 
    \begin{equation}
        V_\alpha^\mathcal{M}\left(\dutchcal{i}(p), \Phi \right) = D_\alpha^\text{ALL}\left(\dutchcal{i}(p) \middle\| \Phi \right)
    \end{equation}
    Therefore, for \(p \in [0,1)\), the gap between the measured R\'enyi divergence and the variational bound is strict, i.e.\
    \begin{equation}
        D_\alpha^\mathcal{M} \left( \dutchcal{i}(p) \middle\| \Phi\right) < V_\alpha^\mathcal{M}\left(\dutchcal{i}(p), \Phi \right) \, .
    \end{equation}
\end{proposition}

\begin{proof}
    By scaling invariance, the variational bound is given by 
    \begin{equation}
        \exp( - V_\alpha^\mathcal{M}\left(\dutchcal{i}(p), \Phi \right) ) = \inf_{\omega > 0} \left\{ \tr[ \dutchcal{i}(p) \omega^\frac{\alpha-1}{\alpha}] \; \middle| \; \tr[\Phi \omega ] = 1 \wedge \omega \in C_\mathcal{M} \right\} \, .
    \end{equation}

    By an analogous argument as in the proof for Proposition \ref{Prop:Gap}, we can restrict w.l.o.g.\ to isotropic operators \(\mathcal{I}(\omega)\) since for any \(\omega \in C_\mathcal{M}\), we have
    \begin{equation}
        \tr[ \dutchcal{i}(p) \omega^\frac{\alpha}{\alpha-1}] = \tr[ \dutchcal{i}(p) \mathcal{I}\left(\omega^\frac{\alpha-1}{\alpha}\right)] \geq \tr[ \dutchcal{i}(p) \mathcal{I}(\omega)^\frac{\alpha-1}{\alpha}] 
    \end{equation}
    due to the operator convexity of the function \(f(t) := t^{\frac{\alpha-1}{\alpha}}\) \cite[Section V]{Bhatia}. The resulting program is
    \begin{equation}
        \text{minimize} \; \; p+ (1-p)c_2^{\frac{\alpha}{\alpha-1}} \; \; \text{s.t.} \; \; c_2 \geq \frac{1}{d+1} \, .
    \end{equation}
    Its minimal value is \(p\) and comparing with Eq.\ \eqref{Eq:Isotropic_All_Alphha} we thus have \( V_\alpha^\mathcal{M}\left(\dutchcal{i}(p), \Phi \right) = D_\alpha^\text{ALL}\left(\dutchcal{i}(p), \Phi \right) \). 

    Lastly, by Proposition \ref{Prop:Isotropic_Local} we have for \(p<1\) that
    \begin{equation}
        D_\alpha^\mathcal{M}\left(\dutchcal{i}(p), \Phi \right) = \frac{\alpha}{\alpha-1} \log(\frac{pd+1}{d+1}) < \frac{\alpha}{\alpha-1} \log p = D_\alpha^\text{ALL}\left(\dutchcal{i}(p), \Phi \right) \, .
    \end{equation}
\end{proof}


\subsection{Example for \texorpdfstring{\(\alpha=1\)}{alpha equal 1}}

\begin{proposition}
    Let \(\mathcal{M} := \{ \text{SEP}, \text{PPT} \} \), \(p \in [0,1]\) and \(q \in [0,1]\). With definitions as above, we have 
    \begin{equation}
        V_1^\mathcal{M}\left(\dutchcal{i}(p), \dutchcal{i}(q) \right) = D_1^\text{ALL}\left(\dutchcal{i}(p) \middle\| \dutchcal{i}(q) \right)
    \end{equation}
    if the coefficients satisfy the constraint
    \begin{equation}\label{Eq:Constraint_1}
        q \geq \frac{p}{d+1-pd} \, .
    \end{equation}
    The gap between the measured R\'enyi divergence and the variational bound is then strict for \(p \in [0,1)\) and \(q\) as in Eq.\ \eqref{Eq:Constraint_1}, i.e.\ we have
    \begin{equation}
        D_1^\mathcal{M} \left( \dutchcal{i}(p) \middle\| \dutchcal{i}(q) \right) < V_1^\mathcal{M}\left(\dutchcal{i}(p), \dutchcal{i}(q) \right) \, .
    \end{equation}
\end{proposition}
\begin{proof} 
    The variational bound is given by
    \begin{equation}
         V_1^\mathcal{M}\left(\dutchcal{i}(p), \dutchcal{i}(q) \right) = \sup_{\omega>0} \bigg\{ \tr[\dutchcal{i}(p) \log \omega ] \; \bigg| \; \tr[ \dutchcal{i}(q) \omega] = 1 \wedge \omega \in C_\mathcal{M} \bigg\} \, .
    \end{equation}

    As in the proofs above, we can restrict to an optimization over isotropic operators since
    \begin{equation}
        \tr[ \dutchcal{i}(p) \log \omega] = \tr[ \dutchcal{i}(p) \mathcal{I}\left( \log \omega \right)] \leq \tr[ \dutchcal{i}(p) \log \mathcal{I}(\omega)]
    \end{equation}
    by the operator concavity of \(f(t):=\log t\) \cite[Section V]{Bhatia}. We then have to solve the following program 
    \begin{equation}
        \text{maximize} \; \; p \log c_1 + (1-p) \log c_2 \; \; \text{s.t.} \; \; q c_1 + (1-q) c_2 = 1 \wedge 0 < c_1 \leq (d+1)c_2 \, .
    \end{equation}
    Its easy to see that if \(q=1\), the optimal value is given by \(+\infty\). Assuming \(q\not=1\), we can rewrite the constraints into
    \begin{align}
        c_2 &= \frac{1 - c_1 q}{1-q} && \text{and} & 0 < c_1 \leq \frac{d+1}{qd+1} \, .
    \end{align}
    The objective function is then 
    \begin{equation}
        p \log c_1 + (1-p) \log (1-c_1 q) - (1-p) \log (1-q) \, .
    \end{equation}
    By the derivative test, we obtain the optimal point of the unconstrained problem as \(c_1^\star = \frac{p}{q}\).
    This is also the optimal point of the constrained problem if 
    \begin{align}
        &\frac{p}{q} \leq \frac{d+1}{qd+1} && \text{ which is satisfied if} & q \geq \frac{p}{d+1 -pd} \, .
    \end{align}
    Comparing with Eq.\ \eqref{Eq:Isotropic_All_1}, we obtain the claimed equality \( V_1^\mathcal{M}\left(\dutchcal{i}(p), \dutchcal{i}(q) \right) = D_1^\text{ALL}\left(\dutchcal{i}(p) \middle\| \dutchcal{i}(q) \right)\). 

    Finally, note that with \(D(x\|y) := x(\log x -\log y) \) for \(x,y\in\mathbbm{R}\), we have 
    \begin{equation}
        D\left( p + \frac{1-p}{d+1} \middle\| q + \frac{1-q}{d+1} \right) = D \left( \frac{1}{d+1} + p \frac{d}{d+1} \middle\| \frac{1}{d+1} + q \frac{d}{d+1} \right) \leq \frac{d}{d+1} D(p\|q)
    \end{equation}
    by the joint convexity of \(D(x\|y)\) in the pair \((x,y)\) and \(D(1\|1) = 0\). This shows that 
    \begin{equation}
        D_1^\mathcal{M} \left( \dutchcal{i}(p) \middle\| \dutchcal{i}(q) \right) \leq \frac{d}{d+1} V_1^\mathcal{M}\left(\dutchcal{i}(p), \dutchcal{i}(q) \right) < V_1^\mathcal{M}\left(\dutchcal{i}(p), \dutchcal{i}(q) \right) \, .
    \end{equation}
\end{proof}


\subsection{Example for \texorpdfstring{\(\alpha\in(1,+\infty)\)}{alpha in (1,infinity)}}

\begin{proposition}
    Let \(\mathcal{M} := \{ \text{SEP}, \text{PPT} \} \), \(p \in [0,1]\), \(q \in [0,1]\) and \(\alpha\in(1,+\infty)\). With definitions as above, we have 
    \begin{equation}
        V_\alpha^\mathcal{M}\left(\dutchcal{i}(p), \dutchcal{i}(q) \right) = D_\alpha^\text{ALL}\left(\dutchcal{i}(p), \dutchcal{i}(q) \right)
    \end{equation}
    if the coefficients satisfy
    \begin{equation}
        q \geq \frac{p }{ \left(d+1\right)^\frac{1}{\alpha} - p \left( \left(d+1\right)^\frac{1}{\alpha} - 1 \right)}
    \end{equation}

    So for \(p \in (0,1)\), we have a finite gap 
    \begin{equation}
        D_\alpha^\mathcal{M} \left( \dutchcal{i}(p) \middle\| \dutchcal{i}(q) \right) < V_\alpha^\mathcal{M}\left(\dutchcal{i}(p), \dutchcal{i}(q) \right) \, .
    \end{equation}
\end{proposition}

\begin{proof}

The variational bound is given by
\begin{equation}
    \exp( - V_\alpha^\mathcal{M}\left(\dutchcal{i}(p), \dutchcal{i}(q) \right) ) = \sup_{\omega>0} \left\{ \tr[\dutchcal{i}(p) \omega^{\frac{\alpha-1}{\alpha}} ] \; \middle| \; \tr[ \dutchcal{i}(q) \omega] = 1 \wedge \omega \in C_\mathcal{M} \right\} 
\end{equation}

Using the operator concavity of \(f(t):=t^{\frac{\alpha-1}{\alpha}}\) for \(\alpha\in(1,\infty)\) \cite[Section V]{Bhatia}, we can conclude as above that the we can restrict w.l.o.g.\ to an optimization over isotropic operators. We then have to solve the following program
\begin{equation}
    \text{maximize} \; \; p c_1^\frac{\alpha-1}{\alpha} + (1-p) c_2^\frac{\alpha-1}{\alpha} \; \; \text{s.t.} \; \; q c_1 + (1-q) c_2 = 1 \wedge 0 < c_1 \leq (d+1)c_2 \, .
\end{equation}
If \(q=1\) it is easy to verify that the optimal value is \(+\infty\). Otherwise, we can recast the constraints into 
\begin{align}
    c_2 &= \frac{1 - c_1 q}{1-q} && \text{and} & 0 < c_1 \leq \frac{d+1}{qd+1}
\end{align}
and the objective function then reads
\begin{equation}
    p c_1^\frac{\alpha-1}{\alpha} + (1-p) (1-q)^\frac{1-\alpha}{\alpha} (1-c_1 q)^\frac{\alpha-1}{\alpha} \, .
\end{equation}

By the derivative test, the optimal point of the unconstrained problem is given by
\begin{equation}
    c_1^\star = \frac{1}{q + \left(\frac{q}{p} \right)^\alpha (1-p)^\alpha (1-q)^{1-\alpha}}
\end{equation}
which is the optimal point of the constrained problem if the coefficients satisfy
\begin{align}
    \left(\frac{q}{p} \right)^\alpha (1-p)^\alpha (1-q)^{1-\alpha} \geq \frac{1-q}{d+1} \, ,
\end{align}
which we can reformulate as 
\begin{equation}
    q \geq \frac{p }{ \left(d+1\right)^\frac{1}{\alpha} - p \left( \left(d+1\right)^\frac{1}{\alpha} - 1 \right)} \, .
\end{equation}

Under this constraint, we get the claimed equality \( V_\alpha^\mathcal{M}\left(\dutchcal{i}(p), \dutchcal{i}(q) \right) = D_\alpha^\text{ALL}\left(\dutchcal{i}(p) \middle\| \dutchcal{i}(q) \right)\) by comparing with Eq.\ \eqref{Eq:Isotropic_All_Alphha}.

Lastly, we have that
    \begin{align}
        \left( p + \frac{1-p}{d+1} \right)^\alpha \left( q + \frac{1-q}{d+1} \right)^{1-\alpha} &= \left( \frac{1}{d+1} + p \frac{d}{d+1} \right)^\alpha \left( \frac{1}{d+1} + q \frac{d}{d+1} \right)^{1-\alpha} \\
        & \leq \frac{1}{d+1} + \frac{d}{d+1} p^\alpha q^{1-\alpha}
    \end{align}
    by the joint convexity of \(x^\alpha y^{1-\alpha}\) in the pair \((x,y)\). The inequality is strict if \(p \not = q\). Finally, using the quasi-convexity of the function \(f(t):=\frac{1}{\alpha-1} \log t\) we get  
    \begin{equation}
        D_\alpha^\mathcal{M} \left( \dutchcal{i}(p) \middle\| \dutchcal{i}(q) \right) < V_\alpha^\mathcal{M}\left(\dutchcal{i}(p), \dutchcal{i}(q) \right) \, .
    \end{equation}
\end{proof}


\section{Werner States}\label{App:Werner_States}

The second family of highly symmetric states we consider are the Werner states \cite{Werner}. These are defined by their invariance under conjugation with unitaries of the form \(U\otimes U\). As for isotropic states, a single parameter \(p\in[0,1]\) suffices to characterize them completely. A general Werner states is given by
\begin{equation}\label{Def:Werner_States}
    \dutchcal{w}(p) := p \Theta + (1-p) \Theta^\perp \, ,
\end{equation}
where \(\Theta\) is the completely symmetric and \(\Theta^\perp\) the completely anti-symmetric state. These states are defined as
\begin{align}
    \Theta &:= \frac{1_{AB} + F}{d(d+1)} & \text{and} && \Theta^\perp &:= \frac{1_{AB} - F}{d(d-1)} \, 
\end{align}
with the swap operator \(F := \sum_{i,j} \ketbra{i}{j}_A \otimes \ketbra{j}{i}_B \). It is well-known that the Werner states are separable and have PPT for \(p \in [1/2,1]\) and otherwise are entangled \cite{Werner}. This class of states was studied previously in the context of local distinguishability e.g.\ in \cite{Matthews, Cheng, Matthews2}.


\subsection{Local Distinguishability}

We start with the case \(\rho=\Theta^\perp\) and \(\sigma=\Theta\). Similar to testing the extremal isotropic states against each other, these are perfectly distinguishable if we allow all measurements, i.e.\ \(D_\alpha^\text{ALL} \left( \Theta^\perp \middle\| \Theta \right) = +\infty \) for all orders of \(\alpha > 0\). Their distinguishability under local measurements, however, is strongly reduced.
\begin{proposition}\label{Prop:Antisymmetric_versus_Symmetric}
    Let \( \mathcal{M} = \{ \text{LO}, \text{LOCC}_1, \text{LOCC}, \text{SEP}, \text{PPT} \} \) and \(\alpha > 0 \). With definitions as above, we have 
    \begin{equation}
        D_\alpha^\mathcal{M} \left(\Theta^\perp \middle\| \Theta \right) = \log(\frac{d+1}{d-1}) \, .
    \end{equation} 
\end{proposition}

\begin{proof}
We prove the claim with the proof strategy used for Proposition \ref{Prop:Max_versus_Orthogonal}.

The lower bound on \(D_\alpha^\text{LO}\) using the local basis measurement \(L^{(i,j)} = \ketbra{i}{i}_A \otimes \ketbra{j}{j}_B \) yields the probability distributions
\begin{align}
    \mu_\Theta^L &= \frac{1+\delta_{i,j}}{d(d+1)} & \text{and} && \mu_{\Theta^\perp}^L &= \frac{1-\delta_{i,j}}{d(d+1)}
\end{align}
with which we immediately get the desired lower bound. 

For the upper bound, we focus on the case \(\alpha = \infty\) due to the monotonicity in \(\alpha\). Here, we have to evaluate the optimization problem
\begin{align}
     D_{\max}^\text{PPT}\left(\Theta^\perp \middle\| \Theta \right) = \sup_{\omega>0} \left\{ \log \tr[\Theta^\perp \omega] \; \middle| \; \tr[\Theta \omega] = 1 \wedge \omega^\Gamma \geq 0 \right\}  \, .
\end{align} 
As before, we can take advantage of the symmetry inherent to the Wener states. For this, we use that \( \tr[\dutchcal{w}(p) \omega ] = \tr[ \dutchcal{w}(p) \mathcal{T}(\omega)] \) holds for any \(\omega \in \mathcal{P}\), where \(\mathcal{T}\) denotes the twirling channel. Its action on a linear operator \(X_{AB}\) is given by
\begin{equation}
    \mathcal{T}(X_{AB}) = \int_{\mathcal{U}(d)} \left(U \otimes U\right) X_{AB} \left(U \otimes U\right)^\dagger \dd{\mu_H(U)} \, .
\end{equation}

Since the channel \(\mathcal{T}\) preserves the PPT property of \(\omega\), we can thus restrict the optimization w.l.o.g.\ to PPT Werner operators. These are characterized completely by two scalar coefficients by
\begin{equation}
    \mathcal{T}(\omega) = c_1 \frac{1_{AB} + F}{2} + c_2 \frac{1_{AB} - F}{2} \, .
\end{equation}
The PPT condition translates into the constraint \(c_2(d-1) \leq c_1(d+1)\). As a result of the twirling technique, we are left to evaluate the following optimization problem
\begin{equation}
    \text{maximize} \; \; \log c_2 \; \; \text{s.t.} \; c_1 = 1 \wedge 0 < c_2 \leq \frac{d+1}{d-1} c_1 \, .
\end{equation}
Its solution is \(\log(\frac{d+1}{d-1}) \) which completes the proof.
\end{proof}

Notice that when testing the extremal isotropic operators, the local measurement sets were able to perform better for larger system dimensions. This is not the case when testing the extremal Werner states against each other. Here, we can see that as \(d\to\infty\) the distinguishing power, as quantified by the measured R\'enyi divergence, decreases to zero.

We can generalize our proof to the case \(\rho=\Theta^\perp\) and \(\sigma=\dutchcal{w}(q)\).
\begin{proposition}\label{Prop:Anti_versus_Werner}
    Let \( \mathcal{M} = \{ \text{LO}, \text{LOCC}_1, \text{LOCC}, \text{SEP}, \text{PPT} \} \), \(q\in[0,1]\) and \(\alpha > 0 \). With definitions as above, we have 
    \begin{equation}
        D_\alpha^\mathcal{M} \left(\Theta^\perp \middle\| \dutchcal{w}(q) \right) = \log(\frac{d+1}{d+1 - 2q}) \, .
    \end{equation} 
\end{proposition}
\begin{proof} 
    For \(q=0\), the claim holds trivially so we assume in the following \(q\not=0\). Then, we can show the lower bound by applying the same arguments as in the proof of Proposition \ref{Prop:Antisymmetric_versus_Symmetric}. For the upper bound, we simplify the variational upper bound
    \begin{align}
         D_{\max}^\text{PPT}\left(\Theta^\perp \middle\| \dutchcal{w}(q) \right) = \sup_{\omega>0} \left\{ \log \tr[\Theta^\perp \omega] \; \middle| \; \tr[\dutchcal{w}(q) \omega] = 1 \wedge \omega^\Gamma \geq 0 \right\}  
    \end{align} 
    via the twirling technique. This enables us to rewrite this program as follows
    \begin{equation}\label{Eq:Werner_Program}
        \text{maximize} \; \; \log c_2 \; \; \text{s.t.} \; \; c_1 q + (1-q) c_2 = 1 \wedge 0 < c_2 \leq \frac{d+1}{d-1} c_1 \, .
    \end{equation}
    The constraints can be rephrased into
    \begin{equation}
         0 < c_2 \leq \frac{d+1}{d+1 - 2q} \, ,
    \end{equation}
    which shows that Eq.\ \eqref{Eq:Werner_Program} gives the matching upper bound.  
\end{proof}


\subsection{Additivity on I.I.D.\ States}

As we want to solve the hypothesis testing problem, we additionally need to consider the i.i.d.\ behaviour for these states. For this, we may take advantage of the dual characterization of \(D_{\max}^\text{PPT}\). The locally-measured R\'enyi divergences are additive on tensor powers of \(\rho=\Theta^\perp\) and \(\sigma=\Theta\). 
\begin{proposition}\label{Prop:Additivity_Antisymmetric_versus_Symmetric}
    Let \( \mathcal{M} = \{ \text{LO}, \text{LOCC}_1, \text{LOCC}, \text{SEP}, \text{PPT} \} \) and \(\alpha > 0 \). With definitions as above, we have 
        \begin{equation}
            D_\alpha^\mathcal{M} \left( \left(\Theta^\perp\right)^{\otimes n}\middle\|\Theta^{\otimes n} \right)  = n \log(\frac{d+1}{d-1}) \, .
        \end{equation}
\end{proposition}

\begin{proof} We use the same proof strategy as for Proposition \ref{Prop:Additivty_Max_versus_Orthogonal}. 

    For the lower bound, we just note that the product POVM \(L^{\otimes n}\), where \(L\) is defined as in the proof of Proposition \ref{Prop:Antisymmetric_versus_Symmetric}, achieves the given value.

    For the upper bound, we use the dual characterization of \(D_{\max}^\text{PPT}\). For this example, we pick the dual-feasible point 
    \begin{align}
        \lambda &= \left(\frac{d+1}{d-1} \right)^n, & X &= 0, & Y &= \frac{1}{d^n(d-1)^n} \left[ \left(1_{AB} + d \Phi \right)^{\otimes n} - \left( 1_{AB} - d \Phi \right)^{\otimes n} \right]
    \end{align}
    to obtain the matching upper bound. The feasibility of \(\lambda\) and \(X\) is clear. Moreover, observe that \(Y\) has eigenvalues
    \begin{equation}
        \frac{1}{d^n(d-1)^n} \bigg[ \left(1 + d\right)^{k} - \left( 1 - d \right)^{k} \bigg] \geq 0 \; \; \text{for} \; \; 0 \leq k  \leq n
    \end{equation}
    and thus \(Y\) is positive semi-definite. Lastly, we have that 
    \begin{align}
        Y^\Gamma &= \frac{1}{d^n(d-1)^n} \left[ \left(1_{AB} + F \right)^{\otimes n} - \left( 1_{AB} -F \right)^{\otimes n} \right] \\
        &= \left(\frac{d+1}{d-1}\right)^n \left(\frac{1_{AB} + F}{d(d+1)} \right)^{\otimes n} - \left( \frac{1_{AB} -F}{d(d-1)} \right)^{\otimes n} = \lambda \Theta^{\otimes n} - (\Theta^\perp)^{\otimes n} \, .
    \end{align}
    Therefore, this point is indeed dual feasible and the proof is complete.
\end{proof}

We can generalize this proof technique to investigate the additivity properties of the locally-measured max-divergence for the general case \(\rho=\dutchcal{w}(p)\) and \(\sigma=\dutchcal{w}(q)\).

\begin{proposition}\label{Prop:Additivity_Antisymmetric_versus_Werner} Let \( \mathcal{M} = \{ \text{LO}, \text{LOCC}_1, \text{LOCC}, \text{SEP}, \text{PPT} \} \), \(p \in [0,1]\) and \(q \in [0,1]\). With definitions as above, we have
\begin{align}
    D_{\max}^\mathcal{M} \left( \dutchcal{w}(p) ^{\otimes n} \middle\| \dutchcal{w}(q) ^{\otimes n} \right)  =  n \log \left(\frac{d+1-2p}{d+1-2q}\right), \, 
\end{align} 
provided one of the following conditions 
\begin{enumerate}
    \item \(p \geq \frac{1}{2}\) (separable), \(q \geq \frac{1}{2}\) (separable),  and \(p\leq q\)
    \item \(p \leq \frac{1}{2}\) (entangled), \(q\geq \frac{1}{2}\) (separable), and \((2p-1)(2q-1) \leq d(p+q-1)\)
    \item \(p\leq \frac{1}{2}\) (entangled), \(q \leq \frac{1}{2}\) (entangled),  and \(p=q\), i.e.\ they have to be equal.
\end{enumerate}
\end{proposition}

\begin{proof}
    The proof works similar to that of Proposition \ref{Prop:Additivity_Max_Divergence}. 

    For the lower bound, we consider the tensor product measurement \(L^{\otimes n}\) as used in Proposition \ref{Prop:Additivity_Antisymmetric_versus_Symmetric}. This gives the bound
\begin{equation}
    D_{\max}^\text{LO} \left( \dutchcal{w}(p)^{\otimes n} \middle\| \dutchcal{w}(q)^{\otimes n} \right) \geq n D_{\max} \left( \mu_{\dutchcal{w}(p)}^L \middle\| \mu_{\dutchcal{w}(q)}^L\right) = n \log \left(\frac{d+1-2p}{d+1-2q}\right) \, ,
\end{equation}
where the final equality holds under the condition that \( p \leq q\). 

For the upper bound, we use the dual characterization of \(D_{\max}^\text{PPT}\) from Corollary \ref{Lem:Dual_PPT_Max}. We pick the dual point is given by
\begin{align}
    \lambda &= \left(\frac{d+1-2p}{d+1-2q}\right)^n, & X &= 0, & Y^\Gamma = \lambda \dutchcal{w}(q)^{\otimes n} - \dutchcal{w}(p)^{\otimes n} \, .
\end{align}

Let us investigate when \(Y\) is positive semi-definite. For this, note that we have
\begin{align}
    Y= \lambda \left( \frac{(d+1-2q)   1_\dutchcal{H} + (2qd-(d+1)) d\Phi}{d(d^2-1)}\right)^{\otimes n}- \left( \frac{ (d+1-2p) 1_\dutchcal{H} + (2pd-(d+1)) d\Phi }{d(d^2-1)} \right)^{\otimes n}
\end{align}
The eigenvalues of \(Y\) are of the form
\begin{align}
    \frac{1}{d^n(d^2-1)^{n-k}}\left( \lambda (2q-1)^k(d+1-2q)^{n-k}-(2p-1)^k(d+1-2p)^{n-k}\right) \; \; \text{with} \; \; 0\leq k\leq n \, .
\end{align}
The requirement of non-negativity then translates into
\begin{equation}
    \left(\frac{d+1-2p}{d+1-2q}\right)^k (2q-1)^k \geq (2p-1)^k \; \; \text{for} \; \; 0\leq k\leq n
\end{equation}

We have the following cases:
\begin{enumerate}
    \item If \(p\geq \frac{1}{2}\) and \(q \geq \frac{1}{2}\), these conditions are satisfied if \(p\leq q\). 

    \item If \( p \leq \frac{1}{2}\) and \(q \geq \frac{1}{2}\), they reduce to the requirement that \((2p-1)(2q-1) \leq d(p+q-1)\)

    \item If \( p \leq \frac{1}{2}\) and \(q \leq \frac{1}{2}\), they only hold in the trivial case \(p=q\).
\end{enumerate}
Hence, under the conditions given in the proposition we have by weak duality that
\begin{align}
     D_{\max}^\text{PPT} \left( \dutchcal{w}(p)^{\otimes n} \middle\| \dutchcal{w}(q)^{\otimes n} \right) \leq n \log \lambda =n \log \left(\frac{d+1-2p}{d+1-2q}\right),
\end{align}
which gives the matching upper bound and completes the proof.
\end{proof}


\subsection{Application to Hypothesis Testing}\label{Sec:Hypothesis_Testing_Werner}

Combining the results of Propositions \ref{Prop:Anti_versus_Werner} and \ref{Prop:Additivity_Antisymmetric_versus_Werner} then allows us to extend results for the extremal Werner states obtained in \cite{Matthews2} (see also \cite[Table 1]{Cheng}). Namely, we have by similar arguments as for the isotropic states that 
\begin{equation}
    \zeta_\text{Stein}^\mathcal{M}\left(\Theta^\perp,\dutchcal{w}(q);\varepsilon\right) = \log \left(\frac{d+1}{d+1-2q}\right) \; \; \text{if} \; \; q \geq \frac{d+1}{d+2}\, 
\end{equation}
and 
\begin{equation}
    \zeta_\text{SC}^\mathcal{M}\left(\Theta^\perp,\dutchcal{w}(q);\varepsilon\right) = r - \log \left(\frac{d+1}{d+1-2q}\right) \; \; \text{if} \; \; q \geq \frac{d+1}{d+2}\, 
\end{equation}
for all \(r \geq \log \left(\frac{d+1}{d+1-2q}\right)\). 



\begin{thebibliography}{24}

\bibitem{Renyi}
    {A. Rényi, “On Measures of Information and Entropy”, in Proceedings of the 4th Berkeley Symposium on Mathematical Statistics and Probability, vol. 4.1, pp. 547–561, 1961.}

\bibitem{KullbackLeibler}
    {S. Kullback and R. A. Leibler, “On Information and Sufficiency”, Annals of Mathematical Statistics, vol. 22, no. 1, pp. 79–86, 1951.}

\bibitem{Matthews}
    {W. Matthews, S. Wehner, and A. Winter, "Distinguishability of Quantum States Under Restricted Families of Measurements with an Application to Quantum Data Hiding", Communications in Mathematical Physics, vol. 291, pp. 813–843, 2009.}

\bibitem{Chitambar}
    {E. Chitambar, D. Leung, L. Mančinska, M. Ozols and A. Winter, "Everything You Always Wanted to Know About LOCC (But Were Afraid to Ask)", Communications in Mathematical Physics, vol. 328, pp. 303–326, 2014.}

\bibitem{Brandao2}
    {F. Brandao, M. Christandl and J. Yard, "Faithful Squashed Entanglement", Communications in Mathematical Physics, vol. 306, pp. 805-830, 2011.}
    
\bibitem{Piani}
    {M. Piani, "Relative Entropy of Entanglement and Restricted Measurements", Physical Review Letters, vol. 103, no. 16, p. 160504, 2009.}

\bibitem{Berta2}
    {M. Berta and M. Tomamichel, "Entanglement Monogamy via Multivariate Trace Inequalities", Communications in Mathematical Physics, vol. 405, no. 29, 2024.}
    
\bibitem{Donald}
    {M. J. Donald, "On the relative entropy", Communications in Mathematical Physics, vol. 105, pp. 13–34, 1986.}

\bibitem{Hiai}
    {F. Hiai and D. Petz, "The Proper Formula for Relative Entropy and its Asymptotics in Quantum Probability", Communications in Mathematical Physics, vol. 143, pp. 99–114, 1991.}

\bibitem{Berta1}
    {M. Berta, O. Fawzi, and M. Tomamichel, "On Variational Expressions for Quantum Relative Entropies", Letters in Mathematical Physics, vol. 107, pp. 2239–2265, 2017.}

\bibitem{Mosonyi2}
    {M. Mosonyi and F. Hiai, "Test-Measured Rényi Divergences," in IEEE Transactions on Information Theory, vol. 69, no. 2, pp. 1074-1092, 2023.}

\bibitem{Sinha}
    {S. Sinha and S. Aravinda, "Generalized \(\alpha\)-Observational Entropy", 2023. [Online]. Available: \url{https://arxiv.org/abs/2312.03572}.}

\bibitem{George}
    {I. George and E. Chitambar, "Cone-Restricted Information Theory", 2024. [Online]. Available: \url{https://arxiv.org/abs/2206.04300}.}

\bibitem{Müller-Lennert}
    {M. Müller-Lennert, F. Dupuis, O. Szehr, S. Fehr and M. Tomamichel, "On Quantum Rényi Entropies: A New Generalization and Some Properties", Journal of Mathematical Physics, vol. 54, no. 12, p. 122203, 2013.}

\bibitem{Wilde}
    {M. M. Wilde, A. Winter and D. Yang, "Strong Converse for the Classical Capacity of Entanglement-Breaking and Hadamard Channels via a Sandwiched Rényi Relative Entropy", Communications in Mathematical Physics, vol. 331, pp. 593–622, 2014.}

\bibitem{Petz}
    {D. Petz, "Quasi-Entropies for Finite Quantum Systems", Reports on Mathematical Physics, vol. 23, no. 1, pp. 57-65, 1986.}

\bibitem{Umegaki}
    {H. Umegaki, "Conditional Expectation in an Operator Algebra", Tohoku Mathematical Journal, vol. 6, no. 2-3, pp. 177-181, 1954.}

\bibitem{Renner}
    {R. Renner, "Security of Quantum Key Distribution", PhD thesis, ETH Zurich, 2005. [Online]. Available: \url{https://arxiv.org/abs/quant-ph/0512258}.}

\bibitem{Datta}
    {N. Datta, "Min- and Max-Relative Entropies and a New Entanglement Monotone", in IEEE Transactions on Information Theory, vol. 55, no. 6, pp. 2816-2826, 2009.}

\bibitem{Mosonyi}
    {M. Mosonyi and T. Ogawa, "Quantum Hypothesis Testing and the Operational Interpretation of the Quantum Rényi Relative Entropies", Communications in Mathematical Physics, vol. 334, no. 3, pp. 1617–1648, 2015.}

\bibitem{Tomamichel}
    {M. Tomamichel, "Quantum Information Processing with Finite Resources", Springer Briefs in Mathematical Physics, 2016.}
    
\bibitem{Csiszar}
    {I. Csiszar, "Generalized Cutoff Rates and R\'enyi's Information Measures", in IEEE Transactions on Information Theory, vol. 41, no. 1, pp. 26-34, 1995.}

\bibitem{Prugovecki}
    {E. Prugovečki, "Information-Theoretical Aspects of Quantum Measurement", International Journal of Theoretical Physics, vol. 16, pp. 321–331, 1977.}

\bibitem{Brandao}
    {F. G. Brandão, A. W. Harrow, J. R. Lee, and Y. Peres, "Adversarial Hypothesis Testing and a Quantum Stein's Lemma for Restricted Measurements", in IEEE Transactions on Information Theory, vol. 66, no. 8, pp. 5037-5054, 2020.}

\bibitem{Terhal}
    {B. M. Terhal, D. P. DiVincenzo, and D. W. Leung, "Hiding Bits in Bell States", Physical Review Letters, vol. 86, no. 25, pp. 5807-5810, 2001.}

\bibitem{Eggeling}
    {T. Eggeling and R. F. Werner, “Hiding Classical Data in Multipartite Quantum States,” Physical Review Letters, vol. 89, no. 9, p. 097905, 2002.}

\bibitem{DiVincenzo}
    {D. P. DiVincenzo, D. W. Leung and B. M. Terhal, "Quantum Data Hiding," in IEEE Transactions on Information Theory, vol. 48, no. 3, pp. 580-598, 2002.}
    
\bibitem{DiVincenzo2}
    {D. P. DiVincenzo, P. Hayden and B. M. Terhal, "Hiding Quantum Data", Foundations of Physics, vol. 33, no. 11, pp. 1629–1647, 2003.}

\bibitem{Gilardoni}
    {G. L. Gilardoni, "On Pinsker's and Vajda's Type Inequalities for Csiszár's  f-Divergences", in IEEE Transactions on Information Theory, vol. 56, no. 11, pp. 5377-5386, 2010.}

\bibitem{Hansen}
    {F. Hansen and G. K. Pedersen, "Jensen's Operator Inequality", Bulletin of the London Mathematical Society, vol. 35, no.4, pp. 553-564, 2003.}

\bibitem{Bhatia}
    {R. Bhatia, "Matrix Analysis", Springer, 1997.}

\bibitem{Alberti}
    {P. M. Alberti, "A Note on the Transition Probability over C*-Algebras", Letters in Mathematical Physics, vol. 7, pp. 25–32, 1983.}

\bibitem{Holevo}
    {A. S. Holevo, "Quantum Systems, Channels, Information: A Mathematical Introduction", De Gruyter, 2013.}

\bibitem{Watrous1}
    {J. Watrous, "Advanced Topics in Quantum Information Theory", Lecture Notes, 2021. [Online]. Available: \url{https://cs.uwaterloo.ca/~watrous/QIT-notes/}.}

\bibitem{Watrous2}
    {J. Watrous, "Theory of Quantum Information", Lecture Notes, 2011. [Online]. Available: \url{https://cs.uwaterloo.ca/~watrous/TQI-notes/}. }

\bibitem{Skowronek}
    {L. Skowronek, "Dualities and Positivity in the Study of Quantum Entanglement", International Journal of Quantum Information, vol. 8, no. 5, pp. 721-754, 2010. }
    
\bibitem{Horodecki}
    {M. Horodecki, P. Horodecki and R. Horodecki, "Separability of Mixed States: Necessary and Sufficient Conditions", Physics Letters A, vol. 223,  no. 1-2, pp. 1-8, 1996. }

\bibitem{Boyd}
    {S. P. Boyd and L. Vandenberghe, "Convex Optimization", Cambridge University Press, 2004.}

\bibitem{Ogawa}
    {T. Ogawa and H. Nagaoka, "Strong Converse and Stein’s Lemma in Quantum Hypothesis Testing", in Asymptotic Theory of Quantum Statistical Inference, pp. 28-42, 2005.}

\bibitem{Nagaoka}
    {H. Nagaoka, "Strong Converse Theorems in Quantum Information Theory", in Asymptotic Theory of Quantum Statistical Inference, pp. 64-65, 2005.}

\bibitem{Han}
    {T. Han and K. Kobayashi, "The Strong Converse Theorem for Hypothesis Testing," in IEEE Transactions on Information Theory, vol. 35, no. 1, pp. 178-180, 1989.}

\bibitem{Horodecki2}
    {M. Horodecki and P. Horodecki, "Reduction Criterion of Separability and Limits for a Class of Protocols of Entanglement Distillation", Physical Review A, vol. 59, no. 6, pp. 4206-4216, 1997.}

\bibitem{Li}
    {K. Li and A. Winter, "Relative Entropy and Squashed Entanglement", Communications in Mathematical Physics, vol. 326, pp. 63–80, 2014.}

\bibitem{Christandl} 
    {M. Christandl, R. Ferrara and C. Lancien, "Random Private Quantum States", in IEEE Transactions on Information Theory, vol. 66, no. 7, pp. 4621-4640, 2020.}

\bibitem{Cheng}
    {HC. Cheng, A. Winter, and N. Yu, "Discrimination of Quantum States Under Locality Constraints in the Many-Copy Setting", Communications in Mathematical Physics, vol. 404, pp. 151–183, 2023.}

\bibitem{Cheng2}
    {HC. Cheng and L. Gao, "On Strong Converse Theorems for Quantum Hypothesis Testing and Channel Coding", 2024. [Online]. Available: \url{https://arxiv.org/abs/2403.13584}.}

\bibitem{Buscemi}
    {F. Buscemi and G. Gour, "Quantum Relative Lorenz Curves", Physical Review A, vol. 95, no. 1, p. 012110, 2017.}
    
\bibitem{Wilde2}
    {M. M. Wilde, M. Berta, C. Hirche and E. Kaur, "Amortized Channel Divergence for Asymptotic Quantum Channel Discrimination", Letters in Mathematical Physics, vol. 110, pp. 2277–2336, 2020.}

\bibitem{Li2}
    {K. Li and G. Smith, "Quantum de Finetti Theorem under Fully-One-Way Adaptive Measurements", Physical Review Letters, vol. 114, no. 16, p. 160503, 2015.}

\bibitem{Ferrari}
    {G. Ferrari, L. Lami, T. Theurer and M. Plenio, "Asymptotic State Transformations of Continuous Variable Resources", Communications in Mathematical Physics, vol. 398, pp. 291-351, 2023.
    }

\bibitem{Farenick1}
    {D. R. Farenick and F. Zhou, "Jensen's Inequality Relative to Matrix-Valued Measures", Journal of Mathematical Analysis and Applications, vol. 327, no. 2, pp. 919-929, 2007. }

\bibitem{Farenick2}
    {D. R. Farenick, “Private Communication”, 2015.}

\bibitem{Werner}
    {R. F. Werner, “Quantum States with Einstein-Podolsky-Rosen Correlations Admitting a Hidden-Variable Model”, Physical Review A, vol. 40, no. 8, pp. 4277–4281, 1989.}

\bibitem{Matthews2}
    {W. Matthews and A. Winter, “On the Chernoff Distance for Asymptotic LOCC Discrimination of Bipartite Quantum States”, Communications in Mathematical Physics, vol. 285, no. 1, pp. 161–174, 2008.}  

\bibitem{Owari1}
    {M. Owari and M. Hayashi, “Asymptotic Local Hypothesis Testing Between a Pure Bipartite State and the Completely Mixed State”, Physical Review A, vol. 90, no. 3, p. 032327, 2014.}

\bibitem{Owari2}
    {M. Owari and M. Hayashi, "Local Hypothesis Testing Between a Pure Bipartite State and the White Noise State", in IEEE Transactions on Information Theory, vol. 61, no. 12, pp. 6995-7011, 2015.}

\bibitem{Hayashi1}
    {M. Hayashi and M. Owari, "Tight Asymptotic Bounds on Local Hypothesis Testing Between a Pure Bipartite State and the White Noise State", in IEEE Transactions on Information Theory, vol. 63, no. 6, pp. 4008-4036, 2017.}

\bibitem{Calsamiglia1}
    {J. Calsamiglia, R. Muñoz-Tapia, L. Masanes, A. Ac\'in and E. Bagan, “Quantum Chernoff Bound as a Measure of Distinguishability Between Density Matrices: Application to Qubit and Gaussian States”, Physical Review A, vol. 77, no. 3, p. 032311, 2008.}

\bibitem{Calsamiglia2}
    {J. Calsamiglia, J. I. de Vicente, R. Muñoz-Tapia and E. Bagan, “Local Discrimination of Mixed States”, Physical Review Letters, vol. 105, no. 8, p. 080504, 2010.}

\bibitem{Nathanson}
    {M. Nathanson, “Testing for a Pure State with Local Operations and Classical Communication”, Journal of Mathematical Physics, vol. 51, no. 4, p. 042102, 2010.}

\bibitem{Yu1}
    {N. Yu, R. Duan and M. Ying, “Four Locally Indistinguishable Ququad-Ququad Orthogonal Maximally Entangled States”, Physical Review Letters, vol. 109, no. 2, p. 020506, 2012.}

\bibitem{Yu2}
    {N. Yu, R. Duan and M. Ying, "Distinguishability of Quantum States by Positive Operator-Valued Measures With Positive Partial Transpose", in IEEE Transactions on Information Theory, vol. 60, no. 4, pp. 2069-2079, 2014.}

\bibitem{LiWang}
    {Y. Li, X. Wang and R. Duan, “Indistinguishability of Bipartite States by Positive-Partial-Transpose Operations in the Many-Copy Scenario”, Physical Review A, vol. 95, no. 5, p. 052346, 2017.}

\bibitem{Akibue}
    {S. Akibue and G. Kato, “Bipartite Discrimination of Independently Prepared Quantum States as a Counterexample to a Parallel Repetition Conjecture”, Physical Review A, vol. 97, no. 4, p. 042309, 2018.}

\end{thebibliography}
\end{document}